\documentclass[acmsmall,nonacm]{acmart}
\usepackage{setspace} %
\usepackage[T1]{fontenc}
\usepackage{microtype} %
\usepackage{enumitem}
\usepackage{hyperref}
\usepackage{aliascnt}
\usepackage{cleveref}

\usepackage[dvipsnames]{xcolor}
\usepackage{graphicx}  %
\usepackage{amsmath,amsthm} %

\newtheorem{theorem}{Theorem}
\newaliascnt{lemma}{theorem}
\newtheorem{lemma}[lemma]{Lemma}
\aliascntresetthe{lemma}
\newaliascnt{claim}{theorem}
\newtheorem{claim}[claim]{Claim}
\aliascntresetthe{claim}
\newaliascnt{corollary}{theorem}
\newtheorem{corollary}[corollary]{Corollary}
\aliascntresetthe{corollary}
\newaliascnt{proposition}{theorem}
\newtheorem{proposition}[proposition]{Proposition}
\aliascntresetthe{proposition}
\theoremstyle{definition}
\newaliascnt{definition}{theorem}
\newtheorem{definition}[definition]{Definition}
\aliascntresetthe{definition}
\newaliascnt{example}{theorem}
\newtheorem{example}[example]{Example}
\aliascntresetthe{example}
\AddToHook{env/example/end}{%
  \begingroup
  \qed
  \endgroup
}
\theoremstyle{remark}
\newaliascnt{remark}{theorem}
\newtheorem{remark}[remark]{Remark}
\aliascntresetthe{remark}

\usepackage{thm-restate}
\usepackage{subcaption}
\usepackage{xspace}
\usepackage{todonotes}
\usepackage{tikz}
\usetikzlibrary{positioning}
\setuptodonotes{
    inline,
    color=teal!20,         %
    bordercolor=black,   %
    textcolor=black,       %
    tickmarkheight=0pt,    %
    size=\small,           %
}

\makeatletter
\providecommand*{\cupdot}{%
  \mathbin{%
    \mathpalette\@cupdot{}%
  }%
}
\newcommand*{\@cupdot}[2]{%
  \ooalign{%
    $\m@th#1\cup$\cr
    \hidewidth$\m@th#1\cdot$\hidewidth
  }%
}
\makeatother

\usepackage{tabularx, environ, caption}
\makeatletter
\newcolumntype{\expand}{}
\long\@namedef{NC@rewrite@\string\expand}{\expandafter\NC@find}

\NewEnviron{problem}[2][]{%
  \def\problem@arg{#1}%
  \def\problem@framed{framed}%
  \def\problem@lined{lined}%
  \def\problem@doublelined{doublelined}%
  \ifx\problem@arg\@empty%
    \def\problem@hline{}%
  \else%
    \ifx\problem@arg\problem@doublelined%
      \def\problem@hline{\hline\hline}%
    \else%
      \def\problem@hline{\hline}%
    \fi%
  \fi%
  \ifx\problem@arg\problem@framed%
    \def\problem@tablelayout{|>{\itshape}lX|c}%
    \def\problem@title{\multicolumn{2}{|l|}{%
        \raisebox{-\fboxsep}{\textsc{#2}}%
      }}%
  \else
    \def\problem@tablelayout{>{\itshape}lXc}%
    \def\problem@title{\multicolumn{2}{l}{%
        \raisebox{-\fboxsep}{\textsc{\large #2}}%
      }}%
  \fi%
  \par\noindent%
  \begin{center}
  \begin{tabularx}{0.9\columnwidth}{\expand\problem@tablelayout}%
    \problem@hline%
    \problem@title\\[2\fboxsep]%
    \BODY\\\problem@hline%
  \end{tabularx}%
\end{center}

  \par%
}
\makeatother

\newenvironment{claimproof}[1][Proof of Claim]{%
  \par\addvspace{6pt}%
  \noindent{\itshape #1. }\pushQED{\qed}%
}{%
  \popQED\par\addvspace{6pt}%
}

\newcommand{\tw}{\ensuremath{\mathsf{tw}}\xspace}
\newcommand{\ghw}{\ensuremath{\mathsf{ghw}}\xspace}
\newcommand{\fhw}{\ensuremath{\mathsf{fhw}}\xspace}
\newcommand{\fwidth}{\ensuremath{f\textsf{-width}}\xspace}

\newcommand{\rank}{\ensuremath{\mathit{rank}}\xspace}

\newcommand{\degree}{\ensuremath{\Delta}\xspace}

\newcommand{\bag}{\ensuremath{\mathsf{bag}}\xspace}
\newcommand{\mbag}{\ensuremath{\mathit{mn}}\xspace}
\newcommand{\target}{\ensuremath{\mathsf{TS}}\xspace}

\newcommand{\cg}{\ensuremath{\mathsf{CG}}\xspace}

\newcommand{\tf}{\ensuremath{\mathit{TF}}\xspace}
\newcommand{\mf}{\ensuremath{\mathit{MF}}\xspace}

\newcommand{\splitf}{\mathit{split}\xspace}
\newcommand{\mir}{\mathit{mir}\xspace}

\newtheorem{question}{Open Question}

\usepackage{algorithm}
\usepackage{algpseudocode}

\AtBeginDocument{%
  }

\begin{document}

\title{FPT Parameterisations of Fractional and Generalised Hypertree Width}

\author{Matthias Lanzinger}
\orcid{0000-0002-7601-3727}
\authornote{Supported by the Vienna Science and Technology Fund (WWTF) [10.47379/ICT2201].}
\affiliation{%
  \institution{TU Wien}
  \city{Vienna}
  \country{Austria}
}
\email{matthias.lanzinger@tuwien.ac.at}

\author{Igor Razgon}
\orcid{0000-0002-7060-5780}
\affiliation{%
  \institution{Durham University}
  \city{Durham}
  \country{United Kingdom}
}
\email{igor.razgon@durham.ac.uk}

\author{Daniel Unterberger}
\authornotemark[1]
\affiliation{%
  \institution{TU Wien}
  \city{Vienna}
  \country{Austria}
}
\email{daniel.unterberger@tuwien.ac.at}

\renewcommand{\shortauthors}{Lanzinger, Razgon, and Unterberger}

\begin{abstract}
We present the first fixed-parameter tractable (FPT) algorithms for exact computation of  generalised hypertree width (\ghw) and fractional hypertree width (\fhw). 
Our algorithms are parameterised by the target width, the rank, and the maximum degree of the
input hypergraph. More generally, we show that testing $f$-width is in FPT for a broad class of width functions that we call \emph{manageable}. This class contains the edge cover number $\rho$ and its fractional relaxation $\rho^*$, and thus covers both generalised and fractional hypertree width.
We additionally extend our framework to also obtain an FPT algorithm for computing a discretised version of adaptive width. 
Our approach extends a recent algorithm for treewidth (Boja\'{n}czyk \& Pilipczuk, LMCS 2022) that utilises monadic second-order transductions. 
To extend this idea beyond treewidth we develop new combinatorial machinery around elimination forests in hypergraphs, culminating in a structural normal form for optimal witnesses that makes transduction-based optimisation applicable in the much more general context of manageable width functions.
This yields the first exact FPT algorithms for these measures under any non-trivial parameterisation and provides structural tools that may enable more direct optimisation algorithms.
\end{abstract}



\maketitle

\section{Introduction}
\label{sec:into}

A tree decomposition of a hypergraph $H$ is a tree $T$  with a function $\bag$ which assigns each vertex in $T$ a set of vertices of $H$ such that certain conditions hold. The \emph{\fwidth} of a tree decomposition $(T,\bag)$ is  defined as $\max_{u \in V(T)} f(\bag(u))$. The \fwidth of $H$ is the minimum \fwidth over all tree decompositions of $H$. This generalises various widely studied (hyper)graph parameters in the literature. For instance, when $f$ is the function $B \mapsto |B|-1$, the \fwidth is treewidth~\cite{ROBERTSON1986309}, %
one of the most prominent concepts in the study of graph algorithms. When $f$ is the edge cover number $\rho$ of the bag, or its fractional relaxation $\rho^*$, the \fwidth equals generalised hypertree width (\ghw)~\cite{10.1145/1265530.1265533} and fractional hypertree width (\fhw)~\cite{10.1145/2636918}, respectively. These parameters have been shown central to the complexity of a wide range of problems whose instances are naturally represented as hypergraphs, e.g., conjunctive query evaluation~\cite{DBLP:journals/jcss/GottlobLS02,DBLP:journals/siamcomp/AtseriasGM13,ngo2018worst}, graph and hypergraph counting problems~\cite{10.1145/3457374,DBLP:conf/stoc/0002LR23,bressan2025complexity}, problems in game theory~\cite{DBLP:journals/jair/GottlobGS05}, or combinatorial auctions~\cite{DBLP:journals/jacm/GottlobG13}.

Despite structural similarities to treewidth, deciding the \ghw and \fhw of a hypergraph is significantly harder. Deciding whether $\tw(H)\leq k$ for a hypergraph $H$\footnote{The treewidth of a hypergraph is equivalent to the treewidth of its primal graph, hence classical results on graphs transfer directly. However, even hypergraphs with $\ghw$ 1 can have arbitrarily high treewidth and this is therefore not helpful in deciding many hypergraph width measures.} is famously fixed-parameter tractable (FPT) when parameterised by the width $k$~\cite{DBLP:journals/siamcomp/Bodlaender96, DBLP:journals/jal/BodlaenderK96},  whereas the analogous problems for \ghw and \fhw are both paraNP-hard~\cite{10.1145/1265530.1265533,10.1145/3457374}. Furthermore, Gottlob et al.~\cite{DBLP:conf/wg/GottlobGMSS05} gave an FPT reduction from the set cover problem to checking \ghw which also implies that, in general, the problem is hard to approximate assuming W$[1] \neq $ FPT~\cite{DBLP:journals/jacm/SLM19}. 

Recently significant progress has been made on identifying computationally easier fragments of these problems. Gottlob et al.~\cite{10.1145/3457374,DBLP:conf/mfcs/GottlobLPR20} showed that, under loose structural restrictions on $H$, deciding whether $\ghw(H)\leq k$ and $\fhw(H)\leq k$ is in the complexity class XP  when parameterised by the width $k$, i.e., permitting an $O(n^{f(k)})$ algorithm. Even more recently, Lanzinger and Razgon~\cite{DBLP:conf/stacs/LanzingerR24} gave the first FPT approximation algorithm for \ghw for hypergraphs with bounded edge-intersection cardinality and in parallel Korchemna et al.~\cite{DBLP:conf/focs/KorchemnaL0S024} showed that even polynomial time approximation of $\ghw$ and $\fhw$ is possible under similar (more general)  structural restrictions.
Nonetheless, as of now, no exact fixed-parameter tractable algorithms under any non-trivial parameterisations are known for either problem.

In this paper we present the first exact FPT algorithms for deciding $\ghw(H)\leq k$ and $\fhw(H)\leq k$. Since the natural parameterisation by width $k$ alone is known to be paraNP-hard, obtaining FPT algorithms requires parameterising beyond width. Concretely, we additionally consider the rank (the maximum edge cardinality) and the maximum vertex degree (denoted by $\Delta(\cdot)$) of the input hypergraph as parameters. Formally we study the following class of problems.

\begin{problem}{$f$-width-check}
    Input & Hypergraph $H$ and $k>0$. \\
    Parameters & $\rank(H)$, $\degree(H)$ and $k$ \\
    Output & $\fwidth(H) \leq k$?
\end{problem}

In order to solve this question for \ghw and \fhw, we show that deciding \fwidth in general is FPT in this setting for a large class of functions $f$ that we will call \emph{manageable}.
Informally, a manageable $f$ is additive on non-adjacent disjoint sets and $f(U) \leq k$ implies bounds on the size of $U$, in addition to some minor technical conditions (the notion is introduced formally in~\Cref{sec:techoverview}).
In particular, both \emph{$\rho$ and $\rho^*$ are manageable width functions} and thus \ghw and \fhw are instances of \fwidth with manageable $f$. Our main result then is the following.

\begin{theorem}
\label{thm1}
    \textsc{$f$-width-check} is fixed-parameter tractable for any manageable width function $f$.
\end{theorem}

We approach this problem by building on the transduction-based treewidth algorithm of Bojańczyk and Pilipczuk~\cite{DBLP:journals/lmcs/BojanczykP22}.
At a high level, their result shows how to express the process of turning a potentially suboptimal tree decomposition into an optimal one by a monadic second-order (MSO) transduction,
based on deep combinatorial insights on the fine structure of so-called elimination forests (see~\Cref{sec:prelims}). In our setting, we follow the same blueprint: we construct an MSO transduction that, given a tree decomposition of bounded treewidth, produces a candidate decomposition and an MSO test that the produced
decomposition has $f$-width at most $k$.

The main difficulty in lifting this blueprint from treewidth to $\ghw/\fhw$ is that the width function stops being ``compositional''.
For treewidth, $f(U)=|U|-1$ satisfies a strong separability property: for every partition $U=U_1\cupdot U_2$, the value $f(U)$ is determined by $f(U_1)$ and $f(U_2)$
(via a fixed monotone combination). The construction in~\cite{DBLP:journals/lmcs/BojanczykP22} crucially exploits this kind of compositionality.
In contrast, for $f\in\{\rho,\rho^*\}$, the value on $U$ cannot in general be recovered from the values on an arbitrary partition.
Our key insight is that a weaker form \emph{does} hold for these functions:
composition becomes possible once we restrict attention to the splits that are \emph{edge-disjoint} (no hyperedge meets both sides),
and a substantial part of the paper is devoted to developing the combinatorial machinery needed to make this restriction compatible with the transduction framework.

Our main technical contribution is the extension of the approach of~\cite{DBLP:journals/lmcs/BojanczykP22} to width functions that only permit such weaker compositionality for edge-disjoint partitions.
This requires large-scale changes from the treewidth result, combining multiple additional combinatorial insights with explicit abstraction of various implicit structural observations. 
Our combinatorial analysis concludes in a hypergraph dealternation principle (later called Reduced Hypergraph Dealternation) that shows that optimal witnesses have, in a certain technical sense, bounded structural complexity relative to a fixed (potentially suboptimal) tree decomposition.

The logical formulation behind our first algorithm also lets us vary the width function itself. We use this flexibility to tackle adaptive width, where the decomposition must adapt to the choice of vertex weights.
Formally, we lift the results for manageable functions to produce MSO tests for $f$-width $\le k$ with vertex classes as free variables. This lets us go beyond a single $f$
and handle more complex width functions that go beyond the $f$-width notion.
In particular, we consider variants of \emph{adaptive
width}~\cite{DBLP:journals/mst/Marx11}, a width measure central to the parameterised complexity of conjunctive query evaluation. It is defined as the supremum over all $f$-widths for fractional independent sets $f$.
Intuitively, these width parameters over families of functions are algorithmically more challenging. Indeed, we are not aware of any algorithm for adaptive width (or the closely related submodular width) in the literature beyond the immediate representation as a linear program of doubly exponential size. 

We define the class $\mathcal{F}_\delta$  that consists of fractional independent sets whose vertex weights are multiples of some resolution $1/\delta$. Taking  $\mathcal F_\delta$-width to mean $\sup_{f \in \mathcal{F}_\delta} f\text{-width}$ we show that this discretised notion of adaptive width is in fact checkable in FPT when parameterised only by the width $k$ and the resolution $\delta$.

\begin{theorem}\label{thmadapt1}
For every integer $\delta\ge\rank(H)$, $\mathcal{F}_\delta$-width can be
computed exactly, or certified to exceed a given threshold $k>0$, in FPT time
parameterised by $k$ and $\delta$.
\end{theorem}

In consequence we also obtain similar results for further variants. Furthermore, we also obtain a factor-$(1+\varepsilon)$ approximation of standard adaptive width for graphs in FPT time parameterised by width $k$, degree $\Delta$, and $\varepsilon$.

\smallskip

It was already proposed in~\cite{DBLP:journals/lmcs/BojanczykP22} that the combinatorial principles underlying the  control of elimination forests used in the proof of \Cref{thm1} are closely tied to the classic FPT algorithm for treewidth by Bodlaender and Kloks~\cite{DBLP:journals/jal/BodlaenderK96}.
For our final main result we make this intuition concrete and present an FPT dynamic programming algorithm that computes the optimal $f$-width of a hypergraph, for manageable $f$.
The algorithm uses the same structural bounds that lead to \Cref{thm1} to represent partial elimination forests as states over a given tree decomposition. Dynamic programming along this tree decomposition then identifies optimal width elimination forests that emerge as admissible combinations of the partial ones.

We believe this second FPT algorithm provides two valuable additions to the first proof of \Cref{thm1} based on MSO model checking. First, it provides us with a concrete, and in principle implementable, algorithm, as well as an elementary (although quadruple exponential for $\ghw$ and $\fhw$) upper bound on the parameter-dependent term in the runtime. Second, the algorithm gives a direct realisation of the connection between elimination forest dealternation and Bodlaender--Kloks style dynamic programming, potentially making this connection useful for related width measures.

The rest of this paper is structured as follows. We formally define tree decompositions and key width measures, as well as further necessary technical preliminaries in \Cref{sec:prelims}. We present an extended technical overview of our main results and the central statements leading to them in~\Cref{sec:techoverview}. In particular, the lifting of our framework to sets of functions and the corresponding results for adaptive width are presented in \Cref{sec:adw}. 
Following that we develop the main combinatorial statement in \Cref{sec:lemdealt.main}.
We present our FPT dynamic programming algorithm for finding optimal $f$-width in \Cref{sec:bk}.
 Closing remarks and discussion of interesting new combinatorial problems raised by our results are given  in~\Cref{sec:conclusion}. To improve presentation some further proof details follow in the appendix. %

\section{Preliminaries}
\label{sec:prelims}

We say a family $(U_i)_{i \in I}$ of sets is a \emph{weak partition} of a set $V$ if $\cup_{i\in I}U_i = V$ and $U_i \cap U_j = \emptyset$ for $i \neq j$. If no $U_i$ is empty, we say it is a \emph{partition} of $V$.

\textbf{Graphs \& Hypergraphs \quad } 
A \emph{hypergraph} $H$ is a tuple $(V,E)$ where $V$ is a finite set and $E$ is a set of non-empty subsets $e \subseteq V$. If all $e \in E$ have exactly two elements, we refer to $H$ as a (simple) graph. We call $V$ the vertices of $H$ and $E$ the edges of $H$. For a hypergraph $H$, we denote the vertices of $H$ by $V(H)$ and the edges by $E(H)$. We write $I(v) :=\{e \in E(H) \mid v\in e\}$ for the set of edges incident to vertex $v$, and we say two vertices $u \neq v$ are adjacent if $I(u) \cap I(v) \neq \emptyset$.  For a vertex subset $U \subseteq V(H)$, we call the hypergraph $H'$ with $V(H') = U$ and $E(H') = \{ U \cap e : e \in E(H),~ U \cap e \neq \emptyset\}$ the \emph{subgraph} of $H$ \emph{induced} by $U$, denoted by $H[U]$. For graphs, we only keep edges of cardinality $2$ in the \emph{induced subgraph}, i.e. only edges contained in $U$. We assume throughout that hypergraphs have no isolated vertices, i.e., every vertex in a hypergraph is contained in some edge.

A path in a hypergraph is a sequence of distinct vertices with consecutive vertices adjacent. A hypergraph is \emph{connected} if every pair of vertices is joined by a path.

Let $G$ be a graph. A path from $v_1$ to $v_\ell$ in $G$ is a sequence $(v_1, v_2, \cdots, v_\ell)$ of distinct vertices such that $\{v_i,v_{i+1}\} \in E$ for $1 \leq i \leq \ell-1$. A cycle is a path except that $v_1 = v_\ell$. A graph is \emph{connected} if, for any $v_i ,v_j \in V$, there is a path from $v_i$ to $v_j$. A connected graph without cycles is a \emph{tree}, and a disjoint union of trees is called a \emph{forest}.

A \emph{rooted tree} is a tree $T$ together with a designated vertex $r \in V(T)$ called the root. For a vertex $v \in V(T)$, we refer to the vertices on the unique path from $r$ to $v$ as the \emph{ancestors} of $v$. If $v_1$ is an ancestor of $v_2$, we refer to $v_2$ as a \emph{descendant} of $v_1$. We will denote the subtree of $T$ induced by the  set of all descendants of a vertex $x$ in $T$ as $T_x$. A disjoint union of rooted trees is a \emph{rooted forest} $F$, and we refer to the roots of the trees as the roots of $F$. Note that this implies that every vertex is both an ancestor and a descendant of itself.

\smallskip
\noindent\textbf{Hypergraph Decompositions\quad }
A \emph{tree decomposition} of hypergraph $H$ is a tree $T$ together with a function $\bag : T \to 2^{V(H)}$ such that, for every edge $e \in E(H)$, there is a vertex $n \in V(T)$ such that $e \subseteq \bag(n)$ (the \emph{containment} condition)
. Additionally, the \emph{connectedness} condition requires that the subgraph (of $T$) induced by $\{n \in V(T) : v \in \bag(n)\}$ is connected for each $v \in V(H)$. We will refer to the set $\{\bag(n) : n \in T\}$ as the bags of $T$. Throughout the paper we will consider tree decompositions to be rooted. Note that this is not a technical restriction as any arbitrary rooting is sufficient.

A \emph{width function} $f$ is a function that assigns a hypergraph $H$ together with a subset $U \subseteq V(H)$ of its vertices a non-negative real number $f(H,U)$. 
We will commonly write this as $f_H(U)$, and, if the underlying hypergraph is clear from context, simply as $f(U)$. 
For a given tree decomposition ${\bf t}=(T,\bag)$ of $H$, the \fwidth of ${\bf t}$ is the maximal width of any bag, i.e. $\max_{v \in V(T)} f(H,\bag(v))$. The \fwidth of $H$ is the minimal \fwidth over all tree decompositions of $H$. We denote this as 
$\fwidth({\bf t})$ and 
$\fwidth(H)$ respectively.
Now let $H$ be a hypergraph and $U \subseteq V(H)$. We say a function $\lambda : E(H) \to \{0,1\}$ is an edge cover of $U$ if $\sum_{\substack{e: v \in e}} \lambda(e) \geq 1$ for all $v \in U$. In other words, we choose a set of edges such that each vertex is included in at least one of them.
The weight of an edge cover is $\sum_{e \in E} \lambda(e)$, and the edge cover number $\rho(H,U)$ is the minimal weight of an edge cover of $U$. The generalised hypertree width (\ghw) of a hypergraph $H$ is the $\fwidth(H)$ where $f = \rho$.
The fractional cover number $\rho^*(H,U)$ is defined analogously, except that we can assign non-integer values to edges, i.e. $\lambda : E(H) \to [0,1]$. The fractional hypertree width (\fhw) of a hypergraph $H$ is the $\fwidth(H)$ where $f = \rho^*$.

\smallskip
\noindent
\textbf{Factors and Elimination Forests \quad}
Our technical development involves extensive combinatorial arguments about the structure of certain rooted forests. In particular we argue on the structure in terms of so-called \emph{factors} as introduced in~\cite{DBLP:journals/lmcs/BojanczykP22}.
\begin{restatable}[Factors]{definition}{factors}\label{def:factors}
For a rooted forest $T$ we define the following three kinds of factors:
\begin{description}[topsep=0pt,noitemsep]
\item[Tree Factor] A \emph{tree factor} is a set $V(T_x)$ for some $x \in V(T)$. We call $x$ \emph{the root} of the factor.
    \item[Forest Factor] 
    A forest factor is a {\bf nonempty} 
union of tree factors whose roots are siblings. In particular, 
a tree factor is also a forest factor but the empty set is not. 
The roots of the tree factors included in the union are 
called the \emph{roots} of the 
factor. 
If they are not roots of components of $T$, we define their 
joint parent as the \emph{parent} of the factor. Note that we consider two nodes of $T$  siblings when they
have a joint parent or if they are both roots of components of $T$. This convention is chosen to remain consistent with the notation of~\cite{DBLP:journals/lmcs/BojanczykP22}.
\item[Context Factor] A context factor is a set $U=X \setminus Y$ where $X$ is a 
tree factor and $Y$ is a forest factor whose roots are strict
descendants of the root of $X$. The \emph{root} of $U$, denoted by $rt(U)$, is the root of $X$
and the roots of $Y$ are called the \emph{appendices} of $U$. 
Note that a tree factor is {\bf not} a context factor. 
We will sometimes refer to
$X$ as the \emph{base tree factor} of $U$ and to $Y$ as the \emph{removed forest} of $U$.  For clarity we sometimes denote $X$ and $Y$ as $X(U)$ and $Y(U)$.
We also denote the parent of the appendices by $pr(U)$ and let 
the \emph{spine} of $U$, denoted by $spine(U)$, be the path 
from $rt(U)$ to $pr(U)$.
\end{description}
\end{restatable}

\smallskip
We say that two factors are  \emph{intersecting} if they have at least one vertex in common. Note also that a tree factor is also always a forest factor.

Let $X$ be a set of nodes of a forest $T$. An \emph{$X$-factor} (w.r.t. $T$) is a factor $F$ with $F\subseteq X$.
A \emph{factorisation} of $X$ is a partition of $X$ into $X$-factors. An $X$-factor is maximal if no other $X$-factor contains it as a strict subset. Let $\mf(X)$ denote the set of maximal $X$-factors.  Note that by Lemma~3.2~\cite{DBLP:journals/lmcs/BojanczykP22} we have that $\mf(X)$  is always a factorisation of $X$.

We will be interested in particular in factors of \textit{elimination forests}. Intuitively, this is a tree-shaped arrangement of the vertices of a hypergraph that takes adjacency into account. We will see that this representation also naturally induces tree decompositions.

\begin{restatable}[Elimination Forest]{definition}{elimforest}%
Let $G$ be a hypergraph. An \emph{elimination forest} $F$ for $G$ is a rooted forest over $V(G)$ such that every $e\in E(G)$ is linearly ordered by the ancestor relation of $F$.
An elimination forest $F$ induces a function $\bag_F:$
 $V(F) \to 2^{V(G)} $ as follows: $\bag_F(u)$ contains exactly $u$ and all those ancestors $v$ that are adjacent to $u$ or a descendant of $u$.  No hyperedge meets two components of $F$.  Consequently, connecting the roots of $F$ arbitrarily and keeping the bags $\bag_F$ fixed yields a tree decomposition of $G$; equivalently, the added ancestor relations do not change any bag because there is no adjacency between distinct components.
In what follows, when we refer to the \fwidth of an
 elimination forest, we mean the \fwidth of this induced tree decomposition.
 \end{restatable}

\begin{figure}
\begin{center}
        \begin{tikzpicture}
            \node at (0,0) {$v_1$};
            \node at (1.5,0) {$v_3$};
            \node at (0,-1.5) {$v_4$};
            \node at (3,0) {$v_5$};
            \node at (-1.5,0) {$v_2$};

            \draw[rounded corners = 8mm, very thick] (-0.5,0.5) -- (-0.5,-2)  -- (2,-2) -- (2,0.5) -- cycle;
            \draw[rounded corners = 5mm, very thick] (1,0.5) -- (1,-0.5)  -- (3.5,-0.5) -- (3.5,0.5) -- cycle;
            \draw[rounded corners = 5mm, very thick] (-2,0.5) -- (-2,-0.5)  -- (0.5,-0.5) -- (0.5,0.5) -- cycle;

            \node at (5,0.5) (v1){$v_1$};
            \node at (4,-0.75) (v2){$v_2$};
            \node at (6,-0.75) (v3){$v_3$};
            \node at (5,-2) (v4){$v_4$};
            \node at (7,-2) (v5){$v_5$};

            \draw[very thick] (v1) -- (v2);
            \draw[very thick] (v1) -- (v3);
            \draw[very thick] (v3) -- (v4);
            \draw[very thick] (v3) -- (v5);

            \node at (9.5,0.5) (b1){$\{v_1\}$};
            \node at (8.5,-0.75) (b2){$\{v_1,v_2\}$};
            \node at (10.5,-0.75) (b3){$\{v_1,v_3\}$};
            \node at (9.5,-2) (b4){$\{v_1,v_3,v_4\}$};
            \node at (11.5,-2) (b5){$\{v_3,v_5\}$};

            \draw[very thick] (b1) -- (b2);
            \draw[very thick] (b1) -- (b3);
            \draw[very thick] (b3) -- (b4);
            \draw[very thick] (b3) -- (b5);
        \end{tikzpicture}
\end{center}
\caption{An elimination forest for the depicted hypergraph (left) and its induced tree decomposition (right).}
\end{figure}

\smallskip

Tree decompositions induced by elimination forests will play a central role in our technical developments. 
We will now fix some auxiliary definitions to simplify reasoning over tree decompositions. To that end, for tree decomposition $\mathbf{t}=(T,B)$ we define the following.
For a non-root node $x$ of $T$ with parent $p$, the \emph{adhesion} is $adh(x) :=\bag(x)\cap \bag(p)$. 
The adhesion of the root node is the empty set.
The \emph{margin} of a node $x$ of $T$ consists of those vertices of $\bag(x)$ that are not in the adhesion of $x$, i.e., $mrg(x) :=\bag(x) \setminus adh(x)$.
Finally, the \emph{component} of node $x$ of $T$ is $cmp(x):=\bigcup_{y \in V(T_x)} mrg(y)$. 
The \emph{target sets} of $\mathbf{t}$ are the components of nodes of $T$. That is $\target(\textbf{t}) = \{cmp(x) \mid x\in V(T)\}$.
When there is an underlying elimination forest that is not clear from the description, we state it in the subscript: say $adh_T(x)$.

\begin{figure}
    \usetikzlibrary {fit}
    \begin{subfigure}{0.45 \textwidth}
    \begin{tikzpicture}[scale = 1,
        customnode/.style={scale = 1,
         circle,
         fill=black,
         fill opacity = 1,
         text opacity=1,
         inner sep=0pt,
         minimum size=5pt,
         font=\small}]

        \node[customnode, fill = red] at (0,0) (v1) {};
        \node[customnode, fill = green] at (0.3,0) (v2) {};

        \node[customnode, fill = red] at (0,-1.03) (v3) {};
        \node[customnode, fill = green] at (0.3,-1.03) (v4) {};
        \node[customnode, fill = orange] at (0.15,-1.26) (v5) {};

        \node[customnode, fill = purple] at (1,-2.11) (v6) {};
        \node[customnode, fill = orange] at (1.3,-2.11) (v7) {};

        \node[customnode, fill = red] at (-1,-2.03) (v9) {};
        \node[customnode, fill = blue] at (-0.7,-2.03) (v10) {};
        \node[customnode, fill = orange] at (-0.85,-2.26) (v11) {};

        \draw[thick] (0.15,0) ellipse (9pt and 9pt);
        
        \draw[thick] (0.15,-1.11) ellipse (9pt and 9pt);

        \node at (0.7,-0.9) (x) {{\Large $x$}};

        \draw[thick] (1.15,-2.11) ellipse (9pt and 9pt);

        \draw[thick] (-0.85,-2.11) ellipse (9pt and 9pt);

        \draw[thick] (0.15,-0.3) -- (0.15, -0.8);
        \draw[thick] (0.33,-1.36) -- (1.1, -1.8);
        \draw[thick] (-0.03,-1.36) -- (-0.8, -1.8);

        \node at (3.5,-0.5) (){$adh(x):$};
        \node[customnode, fill = red] at (4.4,-0.5) {};
        \node[customnode, fill = green] at (4.65,-0.5) {};

        \node at (3.5,-1) (){$mrg(x):$};
        \node[customnode, fill = orange] at (4.4,-1) {};
        
        \node at (3.5,-1.5) (){$cmp(x):$};
        \node[customnode, fill = orange] at (4.4,-1.5) {};
        \node[customnode, fill = blue] at (4.65,-1.5) {};
        \node[customnode, fill = purple] at (4.9,-1.5) {};

    \end{tikzpicture}
    \caption{The adhesion, margin and component of the node $x$.}
    \end{subfigure}
    \hfill
    \begin{subfigure}{0.45 \textwidth}
    \begin{tikzpicture}[scale = 0.8,
        customnode/.style={scale = 1,
         circle,
         fill=black,
         fill opacity = 1,
         text opacity=1,
         inner sep=0pt,
         minimum size=5pt,
         font=\small}]

        \node[customnode] at (0,0) (v1) {};
        \node[customnode] at (-0.7,-1) (v2) {};
        \node[customnode, fill = white, draw] at (0,-1) (v3) {};
        \node[customnode] at (0.7,-1) (v4) {};
        \node[customnode] at (-0.7,-2) (v5) {};
        \node[customnode] at (0,-2) (v6) {};
        \node[customnode] at (0,-3) (v7) {};

        \draw (v1) -- (v2);
        \draw (v1) -- (v3);
        \draw (v1) -- (v4);
        \draw (v3) -- (v5);
        \draw (v3) -- (v6) -- (v7);

        \draw[draw= red, fill = red!60, thick,opacity = 0.2, rounded corners] (0,-0.7) -- (-1,-2) -- (-0.3,-3) -- (0,-3.3) -- (0.3,-3) -- (0.3,-1) -- cycle; 
        
         \node[customnode] at (2.8,0) (u1) {};
        \node[customnode, fill = white, draw] at (2.3,-1) (u2) {};
        \node[customnode, fill = white, draw] at (3.3,-1) (u3) {};
        \node[customnode] at (2,-2) (u4) {};
        \node[customnode] at (2.6,-2) (u5) {};
        \node[customnode] at (3.3,-2) (u6) {};
        \node[customnode] at (3.3,-3) (u7) {};

        \draw (u1) -- (u2);
        \draw (u1) -- (u3);
        \draw (u2) -- (u4);
        \draw (u2) -- (u5);
        \draw (u3) -- (u6);
        \draw (u6) -- (u7);

        \draw[draw= green, fill = green!60, thick,opacity = 0.2, rounded corners] (2.3,-0.7) -- (2,-1) -- (1.7,-2) -- (3.3, -3.3) -- (3.6,-3) -- (3.6, -1) -- (3.2,-0.7) -- cycle; 

        \node[customnode] at (5.6,0) (w1) {};
        \node[customnode] at (5.1,-1) (w2) {};
        \node[customnode, fill = white, draw] at (6.1,-1) (w3) {};
        \node[customnode] at (6.1,-2) (w4) {};
        \node[customnode, fill = orange!70, draw] at (5.4,-3) (w5) {};
        \node[customnode, fill = orange!70, draw] at (6.1,-3) (w6) {};
        \node[customnode] at (6.8,-3) (w7) {};

        \draw (w1) -- (w2);
        \draw (w1) -- (w3);
        \draw (w3) -- (w4);
        \draw (w4) -- (w5);
        \draw (w4) -- (w6);
        \draw (w4) -- (w7);

        \draw[draw= blue, fill = blue!60, thick,opacity = 0.2, rounded corners] (6.1,-0.7) -- (5.8,-1) -- (5.8,-2) -- (6.1,-2.3) -- (6.8,-3.3) -- (7.1,-3) -- (6.4,-1) -- cycle;
        
    \end{tikzpicture}
    \caption{A Tree, Forest and Context Factor. Roots marked in white, appendices in orange.}
    \end{subfigure}
    \caption{Illustrations of factors and key tree decomposition notions used in this paper.}
\end{figure}

\smallskip
\noindent\textbf{Monadic Second-Order Logic\quad}
We use monadic second-order logic (MSO) over finite relational structures,
following Courcelle and Engelfriet~\cite{DBLP:books/daglib/0030804}.
MSO extends first-order logic by allowing quantification over sets of
universe elements. A signature specifies the relation symbols and their
arities; a structure interprets these symbols over a finite universe.
We use two signatures to encode hypergraphs together with tree
decompositions or elimination forests.
Intuitively, $\mathit{Vertex}$, $\mathit{Edge}$, and $\mathit{Incident}$
describe the hypergraph, $\mathit{Node}$ and $\mathit{Parent}$ describe
the decomposition tree, and $\mathit{Bag}$ links each tree node to the
hypergraph vertices in its bag.

For a hypergraph $H$ and a rooted tree decomposition
${\bf t}=(T,{\bf B})$, the structure $M(H,{\bf t})$ has signature
$\tau_{td}$ with unary predicates
$\mathit{Vertex},\mathit{Edge},\mathit{Node}$ and binary predicates
$\mathit{Incident},\mathit{Bag},\mathit{Parent}$.
Its universe is the disjoint union $V(H)\cup E(H)\cup V(T)$,
with the three unary predicates selecting the respective sets.
The relation $\mathit{Incident}(v,e)$ holds exactly when $v\in e$ for $e\in E(H)$,
$\mathit{Bag}(t,v)$ exactly when $v\in{\bf B}(t)$ where $t\in V(T)$, and
$\mathit{Parent}(t,s)$ exactly when $s$ is the parent of $t$ in $T$.

For elimination forests, $\mathit{Vertex}$, $\mathit{Edge}$, and
$\mathit{Incident}$ again describe the hypergraph, while $\mathit{Child}$
describes the forest directly on the hypergraph vertices.
For a hypergraph $H$ and an elimination forest $F$, the structure
$M(H,F)$ has signature $\tau_{ef}$ with unary predicates
$\mathit{Vertex},\mathit{Edge}$ and binary predicates
$\mathit{Incident},\mathit{Child}$.
Its universe is the disjoint union $V(H)\cup E(H)$, with the unary
predicates selecting the respective sets and $\mathit{Incident}$
encoding incidence as above. The relation $\mathit{Child}(u,v)$ holds
exactly when $v$ is a child of $u$ in $F$.
For a structure $M$ encoding a hypergraph and a decomposition, we write
$H(M)$ for the encoded hypergraph and ${\bf t}(M)$ or $F(M)$ for the
encoded tree decomposition or elimination forest, respectively.

\section{Technical Overview}
\label{sec:techoverview}

In this section we present a technical overview of our main result. 
Because parts of the proof require a substantial amount of rather technical combinatorial arguments we will present the main proof ideas, the core combinatorial statements, and the proof of our main result in this section. We then discuss the proofs for the combinatorial statements in \Cref{sec:lemdealt.main}.

\subsection{The Main Theorem} \label{subsec:main}

We isolate the minimal properties of a width function that we need in order to extend the
transduction-based optimisation framework of Boja\'nczyk and Pilipczuk~\cite{DBLP:journals/lmcs/BojanczykP22}. We call such
width functions \emph{manageable}. This class includes the edge-cover number $\rho$ and its
fractional relaxation $\rho^\ast$, and thus subsumes generalised and fractional hypertree width.
At the same time, it is broad enough to plausibly capture further hypergraph width measures
studied in the literature (e.g., variants based on modular width functions~\cite{DBLP:journals/mst/Marx11}).

A key obstacle compared to treewidth is that the relevant width functions are not compositional
with respect to arbitrary partitions of a bag. We therefore distinguish a strong and a weak
notion of separability.

\begin{definition}[Monotone separability]\label{def:weaksep}
Let $f$ be a width function. We say that $f$ is \emph{monotone separable} if there exists a
function $g_f:\mathbb{R}_{\ge 0}^2\to \mathbb{R}_{\ge 0}$ that is monotone in both arguments such that,
for every hypergraph $H$, every $U\subseteq V(H)$, and every partition $U=U_1\cupdot U_2$, we have
\[
  f_H(U) = g_f\bigl(f_H(U_1), f_H(U_2)\bigr).
\]
We say that $f$ is \emph{weakly monotone separable} if the same condition holds for every
\emph{edge-disjoint} partition $U=U_1\cupdot U_2$, i.e.\ no hyperedge of $H$ intersects both $U_1$ and $U_2$.
We refer to $g_f$ as the \emph{combiner} for $f$. We write $g_f(n_1,n_2,n_3)$ for $g_f(g_f(n_1,n_2),n_3)$.
\end{definition}

For instance, treewidth corresponds to $f_H(U)=|U|-1$ and is monotone separable with combiner
$g_f(x,y)=x+y+1$.
In contrast, neither generalised nor fractional hypertree width is monotone separable:
for $f=\rho$ (and similarly $f=\rho^\ast$), the value on $U$ cannot be recovered from the values
on an arbitrary partition $U_1\cupdot U_2$. However, the monotone separability is restored under edge-disjoint
partitions: if $U_1$ and $U_2$ are edge-disjoint, then $\rho_H(U)=\rho_H(U_1)+\rho_H(U_2)$ (and likewise
for $\rho^\ast$), so these functions are weakly monotone separable.

The treewidth algorithm of~\cite{DBLP:journals/lmcs/BojanczykP22}
crucially relies on the underlying width function being monotone separable.
Our main technical step is to
replace it by weak separability, which requires substantial technical extension and additional structure to control how decompositions
interact with edge-disjoint splits. We capture the extra requirements in the following definition.

\begin{restatable}{definition}{managewidth}\label{def:manageablewidth}
We call a width function $f$ \emph{manageable} if it satisfies all of the following.

\begin{description}[topsep=0pt,noitemsep]
\item[Monotonicity.]
For every hypergraph $H$ and all $U\subseteq V\subseteq V(H)$,
\(
  f_H(U)\le f_H(V).
\)

\item[Weak monotone separability.]
$f$ is weakly monotone separable, as in Definition~\ref{def:weaksep}.

\item[Bounded size.]
There exists a computable positive integer-valued function $\beta$ such that for every $k\ge 0$, every hypergraph $H$,
and every $U\subseteq V(H)$, if $f_H(U)\le k$, then
\(
  |U|\le \beta\bigl(k,\rank(H)\bigr).
\)
We refer to any such $\beta$ as a \emph{size bound} for $f$.

\item[Invariant locality.]
For all hypergraphs $H_1,H_2$ and all $U_1\subseteq V(H_1)$, $U_2\subseteq V(H_2)$, if
$H_1[U_1]\cong H_2[U_2]$, then $f_{H_1}(U_1)=f_{H_2}(U_2)$.

\item[Automatisability.]
There exists an algorithm that, given a hypergraph $H$, computes $f_H(V(H))$ in time at most
$t(|V(H)|)$ for some computable function $t$.
\end{description}
\end{restatable}

Intuitively, invariant locality means that $f_H(U)$ depends only on the isomorphism type of the induced
subhypergraph $H[U]$.
Together with bounded size, this yields a \emph{bounded number of values} phenomenon: for fixed
$k$ and $r$, there are only finitely many isomorphism types of induced subhypergraphs $H[U]$
with $\rank(H)\le r$ and $f_H(U)\le k$, and hence only finitely many values $\le k$ that $f$ can
take on such bags (bounded by a function of $k$ and $r$).
Finally, automatisability is necessary only at the final stage, to effectively enumerate (up to isomorphism) all small hypergraphs of $\!f$-width at most $k$.

We are now in a position to state our main result regarding
checking \fwidth in full formality.

\begin{restatable}[Extended version of \Cref{thm1}]{theorem}{fwidthcheck} \label{th:mainalgo}\label{thm:main}
\textsc{$f$-width-check} is fixed-parameter tractable for every manageable width function $f$.
In particular,
there is a $g(k,\rank(H),\degree(H)) \,\|H\|^{O(1)}$ time algorithm,  where $g$ is computable, that takes as input a hypergraph $H$ and a parameter $k$ and accepts exactly if $\fwidth(H) \leq k$.
\end{restatable}
\smallskip

Recall, $\ghw$ is precisely \fwidth when $f$ is the edge cover number $\rho$, and $\fhw$ is $\fwidth$ when $f$ is the fractional edge cover number $\rho^*$. 
It is trivial to verify that  $\rho$ is indeed manageable.
For $\rho^*$, the bounded size follows from
\begin{equation}
    |U| \leq \sum_{v \in U} \underbrace{\sum_{e: v \in e} \lambda(e)}_{\geq 1} \le \sum_{e \in E} \sum_{v \in e} \lambda(e) \leq \sum_{e \in E} \rank(H) \lambda(e) \leq k \cdot \rank(H).
\end{equation}
The other conditions of manageability are immediate. 

Exact fixed-parameter tractable algorithms for \ghw and \fhw then follow directly from \Cref{thm:main}.

\begin{corollary}
    For input hypergraph $H$ and $k>0$, it is fixed-parameter tractable to check $\ghw(H)\leq k$ and $\fhw(H)\leq k$ parameterised by $k, \rank(H), \degree(H)$.
\end{corollary}

In the rest of this section, we will provide an overview of the proof of \Cref{thm:main}. Subsequent sections will then expand on the technical details of the statements made here.

\subsection{The Structure of Optimal Elimination Forests} \label{subsec:optelim}

In this subsection, we develop the main combinatorial engine behind
Theorem~\ref{thm:main}, which we call the \emph{Reduced Hypergraph Dealternation Lemma}. The key to our algorithm is that we can check manageable \fwidth on elimination forests whose structure is simple. However, this simplicity is relative in the following sense. Let ${\bf t}=(T,{\bf B})$ be a tree decomposition
of $H$ whose treewidth is not much larger than $\beta(k,\rank(H))$.
Then there exists an \emph{optimal} elimination forest $F$ for $H$ (i.e., $\fwidth(F)=\fwidth(H)$)
whose induced decomposition is  \emph{aligned} with~${\bf t}$.
This alignment in particular refers to how factorisations of the elimination forest match up with ${\bf t}$.
We will therefore be interested in measures that relate the factorisation of elimination forests to specific tree decompositions. In the following, we refer to a pair of a tree decomposition and an elimination forest for the same hypergraph $H$ as a \emph{\tf-pair}.

We introduce two new measures to quantify the relevant structural alignment of \tf-pairs. The first factor in such efficient computability is that we need each target set of ${\bf t}$ to have a small factorisation in $F$. Intuitively, this restricts the possible ways to construct $F$ from ${\bf t}$ 
(recall, target sets $\target$ and maximal factors $\mf$ are defined in \Cref{sec:prelims})
.

\begin{definition} \label{def:twinforest}
For \tf-pair ${\bf t} ,F$, we define the \emph{split} of the pair as the size of the largest maximal factorisation of target sets of $\bf t$, i.e., $\splitf({\bf t},F) := \max_{U \in \target({\bf t})} |\mf(U)|$.
\end{definition}

The second aspect to this structural alignment is that the individual factors themselves follow a kind of top-down monotonicity. Formally, this means that we want to avoid the presence of many context-factors in our representation.
For the definition, it will be helpful to extend our notation. We call $U\subseteq V(H)$ \emph{well-formed} w.r.t.\ $F$ if $\mf(U)$ contains no context factor
(equivalently, it consists only of forest factors).

\begin{definition}
For \tf-pair ${\bf t},F$ we define the \emph{irregularity} of a node $x$ of ${\bf t}$ as the number of children $y$ of $x$ such that $cmp(y)$ is not well-formed w.r.t. $F$. We denote the maximum irregularity of a node of $T$ by $\mir({\bf t},F)$. 
\end{definition}

These two measures of structural complexity, $\splitf$ and $\mir$,  are key to our framework. 
In key contrast to \cite{DBLP:journals/lmcs/BojanczykP22}, we isolate split and irregularity as the only interface parameters necessary to build a modular, general theoretical framework for width checking using MSO transductions. It allows us to decouple various statements for treewidth into independent parts, connected only through split and irregularity. This abstraction and generalisation ultimately allow us to establish core properties also for width functions that are not monotone separable.

\begin{lemma}[Reduced Hypergraph Dealternation]\label{lem:dealt}
Assume that $f$ is manageable. There exists a computable monotone function $\gamma$ such that for every
hypergraph $H$ and every tree decomposition ${\bf t}$ of $H$, there is a reduced elimination forest $F$ of $H$
with $\fwidth(F)=\fwidth(H)$ and
\[
  \splitf({\bf t},F),\ \mir({\bf t},F)\ \le\
  \gamma\bigl(\tw({\bf t}),\, \fwidth(H),\, \rank(H),\, \Delta(H)\bigr).
\]
\end{lemma}
\begin{proof}[Proof Sketch]
We analyse partitions $(Red,X,Blue)$ of $V(H)$, where $X$ separates $Red$ from $Blue$ (i.e., no hyperedge has both a red and a blue vertex). In the application, $X$ is chosen as a subset of a bag $B_x$ of
${\bf t}$, and hence $|X|\le \tw({\bf t})+1$. We further consider only \emph{reduced} elimination forests: an elimination
forest $F$ is reduced if every parent--child edge $(u,v)$ of $F$ is witnessed by some hyperedge, i.e.,
there exists $e\in E(H)$ with $u\in e$ and $e\cap V(F_v)\neq\emptyset$, where $F_v$ is the subtree rooted
at $v$. We show that one can transform every elimination forest into a reduced one of the same width, hence this restriction can be made without loss of generality.

Fix a reduced optimal elimination forest $F$ and a coloured context factor $U\subseteq Red\cup Blue$.
Along the spine of $U$, consider the maximal monochromatic subpaths (intervals). The number of such
intervals is a coarse measure of how often the spine alternates between the two sides of the cut,
and thus of how ``entangled'' this factor is with the separator $X$: many intervals indicate a complicated
interaction that we will later rule out in a well-aligned optimal forest. A \emph{swap} exchanges two
consecutive intervals along the spine. In the alternating pattern of four consecutive intervals,
swapping the middle two makes adjacent intervals merge, reducing the total number of intervals by~2
while preserving connectivity and monochromaticity (cf.\ \Cref{fig:swap_main}).

\begin{figure}[t]
\centering
\begin{minipage}{0.49\linewidth}
\centering
\medskip
\resizebox{.9\linewidth}{!}{%
\begin{tikzpicture}[
  mycircle/.style={circle,draw=black,fill=black,inner sep=0pt,minimum size=5pt,font=\small}
]
  \node[mycircle] at (0.3,0) (v1) {};
  \node[mycircle] at (1.7,0) (v2) {};
  \node[mycircle] at (2.3,0) (v3) {};
  \node[mycircle] at (3.7,0) (v4) {};
  \node[mycircle] at (4.3,0) (v5) {};
  \node[mycircle] at (5.7,0) (v6) {};
  \node[mycircle] at (6.3,0) (v7) {};
  \node[mycircle] at (7.7,0) (v8) {};

  \draw[red, very thick] (v1) -- (v2);
  \draw[very thick] (v2) -- (v3);
  \draw[blue, very thick] (v3) -- (v4);
  \draw[very thick] (v4) -- (v5);
  \draw[red, very thick] (v5) -- (v6);
  \draw[very thick] (v6) -- (v7);
  \draw[blue, very thick] (v7) -- (v8);

  \node at (1,-0.5) {$P_{i-1}$};
  \node at (3,-0.5) {$P_{i}$};
  \node at (5,-0.5) {$P_{i+1}$};
  \node at (7,-0.5) {$P_{i+2}$};
\end{tikzpicture}%
}
\end{minipage}\hfill
\begin{minipage}{0.49\linewidth}
\centering
\resizebox{.9\linewidth}{!}{%
\begin{tikzpicture}[
  mycircle/.style={circle,draw=black,fill=black,inner sep=0pt,minimum size=5pt,font=\small}
]
  \node[mycircle] at (0.3,0) (v1) {};
  \node[mycircle] at (1.7,0) (v2) {};
  \node[mycircle] at (2.3,0) (v3) {};
  \node[mycircle] at (3.7,0) (v4) {};
  \node[mycircle] at (4.3,0) (v5) {};
  \node[mycircle] at (5.7,0) (v6) {};
  \node[mycircle] at (6.3,0) (v7) {};
  \node[mycircle] at (7.7,0) (v8) {};

  \draw[red, very thick] (v1) -- (v2);
  \draw[very thick] (v2) edge[bend right = 20] (v5);
  \draw[blue, very thick] (v3) -- (v4);
  \draw[very thick] (v6) edge[bend right = 20] (v3);
  \draw[red, very thick] (v5) -- (v6);
  \draw[very thick] (v4) edge[bend right = 20] (v7);
  \draw[blue, very thick] (v7) -- (v8);

  \node at (1,-0.5) {$P_{i-1}$};
  \node at (3,-0.5) {$P_{i}$};
  \node at (5,-0.5) {$P_{i+1}$};
  \node at (7,-0.5) {$P_{i+2}$};
\end{tikzpicture}%
}
\end{minipage}
\caption{A swap operation exchanging two consecutive monochromatic intervals.}
\label{fig:swap_main}
\end{figure}

The key claim is that sufficiently long alternation forces the existence of a \emph{neutral} swap, i.e.,
a swap that does not increase $\fwidth$. To prove this, we classify intervals by their \emph{interface}
with the separator $X$. An \emph{$X$-ball} expresses how the bags induced along an interval intersect~$X$.
Manageability implies that, for fixed $|X|$, $\fwidth(H)$, $\rank(H)$, and $\Delta(H)$, only boundedly many
different $X$-balls can occur (\Cref{sec:xintervalbounded}). In particular, this follows from the \emph{bounded size} property limiting the size of the interface, and \emph{invariant locality} limiting the number of isomorphism types. Hence, if a spine contains many intervals,
then many of them share the same $X$-ball. On such a homogeneous region, the effect of swapping depends
only on bounded local information; a pigeonhole argument yields two consecutive intervals whose swap is
neutral (\Cref{sec:xhomogenousneutral}).

Finally, we simplify an optimal forest using a local potential measuring excess monochromatic intervals.
A neutral swap strictly decreases this potential, and we apply these local simplifications in bottom-up
order, preserving previously obtained factorisations. This bounded alternation translates into bounded maximal factorisations of target sets and
bounded irregularity, giving the desired bounds on $\splitf({\bf t},F)$ and $\mir({\bf t},F)$.
All bounds used in this construction arise from effective finite
constructions; taking their monotone closure therefore yields a computable
function $\gamma$.
   The full proof of the lemma is the content of \Cref{sec:lemdealt}.
\end{proof}

\subsection{From Dealternation to Algorithms}
 \label{subsec:transd}\label{subsec:mainproof}

We move on to discussing how to use this main combinatorial observation algorithmically and prove \Cref{thm:main}. 

For a tree decomposition $(T,\bag)$ of a hypergraph $H$, $\mbag$ is a function from $V(H)$ to $V(T)$ where $u$ is mapped to the (unique) node $x$ with $u \in mrg(x)$; we say $x$ is the \emph{margin node of $u$}.
We will primarily be interested in the paths in the tree decomposition induced by two margin nodes. For vertices $u,v$ of $H$, let $\mathit{Path}_\mathit{\mbag}(u,v)$ be the set of vertices on the shortest path from $\mbag(u)$ to $\mbag(v)$ in $T$. Using this definition of $\mbag$, we can now define the key to transitioning from irregularity and split to our algorithmic applications. For each hypergraph vertex $u$ and a given elimination forest $F$, we want to capture the part of the tree decomposition that connects the margin node of $u$ to the margin nodes of its children in $F$. We will refer to such regions of a decomposition as stains (of $u$), and importantly, overlapping stains mean that two vertices’ connections to their children pass through a common part of the tree decomposition.

\begin{definition} \label{def:stain}
Let ${\bf t},F$ be a \textsc{tf}-pair of $H$ and let $u \in V(H)$. The \emph{stain} of $u$ (w.r.t. ${\bf t},F$),
denoted by $\mathit{Stain}(u,({\bf t},F))$, is the set consisting 
of $\mbag(u)$ and  the vertices of $\mathit{Path}_\mathit{\mbag}(u,v)$
for each child $v$ of $u$ in $F$.
That is,
$$
\mathit{Stain}(u,({\bf t},F))=\{\mbag(u)\} \cup \bigcup_{v \in \mathit{Child}_F(u)} \mathit{Path}_\mathit{\mbag}(u,v)
$$
\end{definition}

Note that Definition \ref{def:stain} is different from
the definitions of irregularity and split above in that it considers a part of $F$ in the context of 
${\bf t}$ rather than a part of ${\bf t}$ in the context of $F$. Intuitively, the stain of a vertex $u$ represents a subtree of $\bf t$ that interacts with the part of $F$ rooted at $u$. 

Algorithmically this gives us a way to quantify how subtrees are coupled with each other. More formally, we capture this ``coupling'' in terms  of the \emph{conflict graph} of a \tf-pair ${\bf t},F$, denoted by 
$\cg({\bf t},F)$, which has $V(H)$ as the set of vertices
and $u$ is adjacent to $v$ if and only if their stains
intersect. Analogous to~\cite{DBLP:journals/lmcs/BojanczykP22}, we use the chromatic number of $\cg({\bf t}, F)$ as the key parameter to obtain our algorithms. Intuitively, this measures how complicated the interaction between an optimal $F$ and the given decomposition $\bf t$ can be, and thereby controls the complexity of searching for $F$.

This is precisely the point where our earlier combinatorial developments come in. We show that it is possible to bound the chromatic number of $\cg$ in terms of the two central parameters of our framework, $\splitf$ and $\mir$. Combined with Lemma \ref{lem:dealt}, we can observe that the complexity of searching for an optimal $F$ is controlled by $\tw({\bf t})$, as well as rank, degree and actual $f$-width of $H$.

\begin{theorem} \label{th:cgchrom}
Let $H$ be a hypergraph, let $\mathbf t$ be a tree decomposition of $H$, and let $F$ be an elimination forest of $H$. Then the chromatic number of $\cg(\mathbf t,F)$
is at most $\tw(\mathbf t)+1+(\mir(\mathbf t,F)+3)\splitf(\mathbf t,F)$.
\end{theorem}

The full proof of the theorem is the content of \Cref{sec:comb}. To give a general idea, observe that the graph $\cg({\bf t}, F)$ is an intersection graph of connected subgraphs of a forest, and therefore also chordal. Chordal graphs are perfect graphs, meaning their chromatic number is the size of their largest clique. We then use the fact that $\cg({\bf t}, F)$ also enjoys the Helly property and further observations on the ``density of stains'' to ultimately bound the size of the largest clique in $\cg({\bf t}, F)$.

\begin{remark}
Let $f$ be manageable, put $k=\fwidth(H)$, $r=\rank(H)$,
$\Delta=\Delta(H)$, $t=\tw(\mathbf t)+1$, and write $\beta=\beta(k,r)$.
Combining the bounds on $\kappa_0$ and $\kappa_1$ in
Lemmas \ref{lem:boundedxintervals} and \ref{lem:homogenousneutral}
with the local-dealternation bound in the proof of Lemma \ref{lem:localdealt}
and the bottom-up construction in Lemma \ref{lem:dealt} gives a bound on
the split of an optimal reduced elimination forest $F$.
Using the irregularity bound of Lemma \ref{lem:splitmir},
Theorem \ref{th:cgchrom} therefore yields
\[
    \chi(\cg(\mathbf t,F))
    \le 2^{2^{O(\beta)}}\,t^{O(\beta)}\,
    \bigl(1+\Delta(r-1)\bigr)^{O(\beta^2)}.
\]
Recall that for generalised and fractional hypertree width, one may take $\beta(k,r)=\lceil kr\rceil$. 
\end{remark}

\paragraph{The MSO path}
Stains and \Cref{th:cgchrom} are the key to all of our algorithmic results. Later in \Cref{sec:bk} we present a dynamic programming algorithm for checking manageable $f$-width that directly realises the high-level ideas described above. In this section we focus on establishing the existence of an FPT algorithm through the comparably simpler path of MSO model checking and Courcelle's theorem inspired by the approach of Boja\'nczyk and Pilipczuk~\cite{DBLP:journals/lmcs/BojanczykP22} for treewidth. We choose this presentation of two alternative arguments as the MSO approach is more flexible and adaptable. For instance, in \Cref{sec:adw} we will present further applications that will build on the same underlying translation to MSO.

Recall that we use $\tau_{ef}$ and $\tau_{td}$ for the vocabularies of the elimination forests and tree decompositions, respectively.
Towards the MSO model checking path we establish the following two statements. First, we observe that we can express whether an elimination forest $F$ has $f$-width at most $k$ in MSO (for manageable $f$). Second, we can construct an MSO formula that for given $H$, and tree decomposition $\bf t$, decides whether there exists an elimination forest with $\chi(\cg(\mathbf t,F))$ within some limit, such that $F$ satisfies a given MSO condition. 
The first check is elementary for manageable $f$. 
The second builds on the transduction approach of~\cite{DBLP:journals/lmcs/BojanczykP22}, namely that a bound on the chromatic number allows us to construct a transduction that includes every elimination forest satisfying that bound.
 We defer the proof details for both lemmas to~\Cref{sec:mso}.

To accommodate later uses uniformly, \Cref{lem:mso-blackbox} is stated in a more general form where the MSO condition may depend on a partition of the vertices. For our main result on manageable width functions it is sufficient to consider the trivial partition with $U_1=V(H)$.

\begin{lemma}\label{lem:width-condition}\label{lem:iso4}
Let $f$ be a manageable width function with size bound $\beta$.
Given $k\ge0$ and $r>0$, one can construct an MSO sentence
$\psi_{f,k,r}$ over $\tau_{ef}$ such that, for every elimination forest
$F$ of a hypergraph $H$ of rank at most $r$,
\[
M(H,F)\models\psi_{f,k,r}
\quad\Longleftrightarrow\quad
\text{the $f$-width of $F$ is at most $k$}.
\]
If $f$ is evaluated on hypergraphs with $\ell$ vertex classes and
satisfies the same locality, automatisability, and size conditions
for labelled hypergraphs, the analogous formula
$\psi_{f,k,r}(U_1,\dots,U_\ell)$ has these classes as free variables
and satisfies the equivalence for every weak partition of $V(H)$.
The construction time is bounded by a computable function of
$k,r,\ell,\beta$ and the evaluation algorithm for $f$.
\end{lemma}

\begin{lemma}\label{lem:mso-blackbox}
Given an integer $q>0$ and an MSO formula $\psi(U_1,\dots,U_\ell)$
over $\tau_{ef}$, there is an MSO formula
$\varphi(U_1,\dots,U_\ell)$ over $\tau_{td}$ such that, for a hypergraph
$H$, a tree decomposition $\mathbf t$, and a weak partition
$(U_1,\dots,U_\ell)$ of $V(H)$,
\[
M(H,\mathbf t)\models\varphi(U_1,\dots,U_\ell)
\quad\Longleftrightarrow\quad
\exists F:\ \chi(\cg(\mathbf t,F))\leq q
\ \land\ M(H,F)\models\psi(U_1,\dots,U_\ell),
\]
where $F$ ranges over elimination forests of $H$ and the same vertex
sets $U_i$ are used in both structures.

Moreover, $\varphi$ can be constructed in time depending only on $q,\ell,\|\psi\|$.
\end{lemma}

Putting everything together, we are now ready to prove \Cref{th:mainalgo}. Starting from a tree decomposition of bounded treewidth, we use MSO model checking to decide whether an elimination forest of $f$-width at most $k$ exists.

\begin{proof}[Proof of Theorem \ref{th:mainalgo}]
The algorithm first computes $p=\beta(k,rank(H))$ 
and runs an off-the-shelf algorithm for exact treewidth computation parameterised by $p-1$. 
We note that, due to the bounded size property,
if $\fwidth(H) \leq k$ then $\tw(H) \leq p-1$.  
Therefore, if the treewidth computation algorithm rejects,
our algorithm also rejects. 

Otherwise, we have 
a tree decomposition ${\bf t}=(T,B)$ of $H$ with treewidth at most $p-1$. 
We then compute $p_1=\gamma(\tw({\bf t}), k,rank(H), \Delta(H))$
with $\gamma$ as in Lemma \ref{lem:dealt}, and $q=\tw({\bf t})+1+(p_1+3)p_1$
as in Theorem \ref{th:cgchrom}.

\begin{claim} \label{clm:algo1}
\sloppy
Assume that $\fwidth(H) \leq k$.  
Then there is an elimination forest $F$ of $H$ such that
$\fwidth(F)=\fwidth(H)$ and $\chi(\cg({\bf t},F))\leq q$.
\end{claim}

\begin{claimproof}
By Lemma \ref{lem:dealt} and the monotonicity of function $\gamma$
as in the lemma, there is an elimination forest $F$ of $H$
with $\fwidth(F)=\fwidth(H)$ and such that both
$\splitf({\bf t},F)$ and $\mir({\bf t},F)$ are at most $p_1$.
Next, by Theorem \ref{th:cgchrom} and the definition of $q$, the chromatic number of $\cg({\bf t},F)$
is at most $q$. 
\end{claimproof}

Next, apply Lemma~\ref{lem:width-condition} to construct
$\psi_{f,k,r}$, where $r=\rank(H)$.
We apply the algorithm of Lemma~\ref{lem:mso-blackbox} with $\psi_{f,k,r}$ and $U_1=V(H)$ to obtain $\phi$. If $M(H,{\bf t}) \models \phi$, then by construction $\fwidth(H) \le k$.  If $M(H,{\bf t}) \not \models \phi$ then combining with Claim~\ref{clm:algo1} we see that $\fwidth(H) > k$.

It remains to verify the FPT time bound.
By our initial choice of $p$, $\tw({\bf t})\leq\beta(k,r)-1$.
The parameter $q$ is bounded by a computable function of
$k,r,\Delta(H)$, and $\|\psi_{f,k,r}\|$ is bounded by a computable
function of $k,r$ for fixed $f$, by Lemma~\ref{lem:width-condition}.
The same parameter bounds control the construction and size of
$\phi$ by Lemma~\ref{lem:mso-blackbox}.
For model checking, Lemma~\ref{lem:mthtw} gives
$\tw(M(H,{\bf t}))\leq\tw({\bf t})+2$, so Courcelle's theorem
provides the required FPT bound
\[
    g(q,1,\|\psi_{f,k,r}\|,\tw({\bf t}))\,
    \|M(H,{\bf t})\|^{O(1)}
\]
for some computable function $g$.
The initial decomposition may be assumed to have $O(\|H\|)$ many bags,
and $\|M(H,{\bf t})\|$ can be bounded in terms of $(1+\beta(k,r))\|H\|^{O(1)}$.
Thus the overall runtime is bounded by $g'(k,r,\Delta(H))\,\|H\|^{O(1)}$ for a computable function $g'$.

\end{proof}

\subsection{Computation and Approximation of Adaptive Width Variants}
\label{sec:adw}

In this section we prove a specialisation of \Cref{thm:main} to a more restricted class of width functions. As discussed earlier, we identify monotone separability as the critical technical property that is missing for generalised and fractional hypertree width, but implicitly needed for the proof of Boja\'nczyk and Pilipczuk~\cite{DBLP:journals/lmcs/BojanczykP22}.
As a consequence of our developments to relax the requirement to weak monotone separability, our framework also reveals additional applications to width functions beyond treewidth that are in fact monotone separable. Additionally, studying our framework in this restrained setting also illustrates that the dependence on the degree in the parameter in our proof is specifically a consequence of handling weak monotone separability.

We begin by introducing the specific kinds of width functions we consider in this section.

\begin{definition}[Strong manageability]\label{def:strongly-manageable}
We call a manageable width function $f$ \emph{strongly manageable} if it is
monotone separable (cf. \Cref{def:weaksep}). 
\end{definition}

We also consider labelled versions of these notions. We call a width function (strongly) $\ell$-manageable if it satisfies the corresponding manageability conditions on hypergraphs whose vertices are equipped with a weak partition into $\ell$ labelled classes. Invariant locality is understood with respect to isomorphisms preserving these classes, and automatisability with respect to labelled hypergraphs. 
Importantly, for strongly $\ell$-manageable functions we can prove a version of the dealternation lemma that does not depend on the degree. Hence, following a similar path towards an algorithm as in \Cref{subsec:mainproof} we will also obtain parameterised algorithms that are not parameterised by the degree.

\begin{restatable}[Reduced Hypergraph Dealternation for Strongly Manageable Functions]{lemma}{StrongDealternation}
\label{lem:dealt-strict}
Assume that $f$ is \emph{strongly} $\ell$-manageable. There exists a computable monotone
function $\gamma_{\mathrm{str}}$ such that for every hypergraph $H$ and every
tree decomposition ${\bf t}$ of $H$, there is a reduced elimination forest $F$ of $H$
with $\fwidth(F)=\fwidth(H)$ and
\[
  \splitf({\bf t},F),\ \mir({\bf t},F)
  \le \gamma_{\mathrm{str}}\bigl(\tw({\bf t}),\fwidth(H),\rank(H),\ell\bigr).
\]
\end{restatable}

We will use strongly $\ell$-manageable width functions to study versions of a more complex width measure than generalised and fractional hypertree width, so-called \emph{adaptive width}.
Any function $f : V(H) \to \mathbb{R}_{\ge 0}$ naturally extends to a width function by setting $f(U)=\sum_{v \in U} f(v)$. Such width functions are called \emph{modular}. For hypergraph $H$, let $\mathcal{F}_H$ be the class of all modular functions on $V(H)$ that additionally satisfy \emph{edge domination}, i.e., $f(e)\le 1$ for all $e \in E(H)$. Equivalently, these are the modular functions extended from fractional independent sets on $H$.
The \emph{adaptive width} $adw(H)$ of $H$ is defined as $\sup_{f\in\mathcal{F}_H}\fwidth(H)$. In general for a class of width functions $\mathcal{C}$ we write $\mathcal{C}$-width for $\sup_{f\in\mathcal{C}}\fwidth(H)$.

That is, for checking the adaptive width of $H$ one needs to verify the existence of a low-width decomposition for every modular edge-dominated function on $H$. It is known that $adw(H) \le \fhw(H)$. It was recently shown that $\fhw(H) = O(adw(H)\log adw(H)\log|V(H)|)$~\cite{lanzinger2026}, so the two measures are somewhat surprisingly more closely related than their definitions suggest. Nonetheless, the complexity of deciding adaptive width exactly remains an open problem. The best (and to our knowledge only non-trivial) known algorithms so far express the problem as a double-exponential size linear program~\cite{panda}.
Width measures based on whole classes of functions in this way have become a central element of the study of the complexity of database query answering~\cite{DBLP:journals/mst/Marx11,DBLP:journals/jacm/Marx13,panda}.

In this section, we discuss how algorithms for two slightly simpler variants  of adaptive width follow from our framework for FPT \fwidth checking described above. In particular, we consider the width measures induced by the following classes of functions.

\begin{definition}
Let $H$ be a hypergraph and let $b\ge\delta\ge\rank(H)$ be positive integers.
We define $\mathcal{F}^*_{\delta}$ as the family of modular width functions
induced by fractional independent sets whose vertex weights are either zero
or at least $1/\delta$. We define $\mathcal{F}_{\delta,b}$ as the subfamily
whose positive weights belong to
\[
    \left\{\frac{1}{\delta}+\frac{i}{b}:
    i\in\mathbb{N}_0,\ \frac{1}{\delta}+\frac{i}{b}\le 1\right\}.
\]
We abbreviate $\mathcal{F}_{\delta,\delta}$ by $\mathcal{F}_{\delta}$.
\end{definition}

Thus, $\mathcal{F}_{\delta}$-width is the \emph{discretised} version of 
adaptive width with bounded precision expressed by $\delta$. The additional $b$ parameter gives a finer discretisation of this relaxed family while keeping the minimum positive weight fixed.
$\mathcal{F}^*_{\delta}$-width relaxes this restriction further and only excludes functions that assign very small non-zero weights.

To apply our framework, we first observe that every individual modular width function is monotone separable, with addition as its combiner. The possibility of zero-weight vertices prevents bounded size, but this is easily remedied by working on the positive
support. Let $H^{f+}$ be the induced subhypergraph on the vertices where
$f(v)\neq 0$.
Indeed, it is straightforward to observe that $\fwidth(H)=\fwidth(H^{f+})$:
restricting the bags of a decomposition of $H$ to the positive support
preserves their weights, and adding all zero-weight vertices to every bag
of a decomposition of $H^{f+}$ yields a decomposition of $H$ of the same $f$-width. 
On the positive support, and for the functions considered here, $f(U)\le k$ implies $|U|\le\lfloor\delta k\rfloor$. Moreover, in the discretised setting, the finitely many possible weights can be encoded as vertex labels. With these labels, invariant locality and automatisability hold, so the resulting width function is strongly $\ell$-manageable. In this way we can then take a similar path as before to obtain an FPT algorithm for $\mathcal{F}_{\delta,b}$-width, parameterised only by target width $k$ and precision $b$.

\begin{theorem}\label{thm:adaptmain}
There is an algorithm that takes as input a hypergraph $H$, a positive
rational $k$, and integers $b\ge\delta\ge\rank(H)$, and returns either the
exact $\mathcal{F}_{\delta,b}$-width of $H$, if it is at most $k$, or rejects.
In the latter case, the $\mathcal{F}_{\delta,b}$-width of $H$ is greater than
$k$. The runtime is bounded by
$g(k,b)\,\|H\|^{O(1)}$ for a computable function $g$.
\end{theorem}

\begin{proof}
The algorithm first computes $p=\lfloor\delta k\rfloor$ and runs an
off-the-shelf algorithm for exact treewidth computation parameterised by
$p-1$. We note that the constant assignment $f(v)=1/\delta$ belongs to
$\mathcal{F}_{\delta,b}$. Therefore, if the
$\mathcal{F}_{\delta,b}$-width of $H$ is at most $k$, then $H$ has a tree
decomposition with bags of size at most $p$, and hence $\tw(H)\le p-1$.
If the treewidth computation algorithm rejects, our algorithm also rejects.

If the algorithm proceeds, we have a tree decomposition ${\bf t}=(T,B)$ of $H$ with treewidth at most $p-1$. Set $r=\rank(H)$ and write the allowed positive weights as $\omega_i=1/\delta+(i-1)/b$, for $i\in[\ell]$, where $\ell=\lfloor b(1-1/\delta)\rfloor+1\le b$.
On the positive support, every assignment uses the same size bound
$\beta(x,r)=\max\{1,\lfloor\delta x\rfloor\}$ and $\ell$ vertex labels.
We then compute
\[
    p_1=\gamma_{\mathrm{str}}\bigl(\tw({\bf t}),k,r,\ell\bigr),
    \qquad Q=\tw({\bf t})+1+(p_1+3)p_1,
\]
using the bounds from \Cref{lem:dealt-strict} and \Cref{th:cgchrom}.

\begin{claim}\label{clm:adaptalgo}
Assume that $f\in\mathcal{F}_{\delta,b}$ has nonempty positive support and
$\fwidth(H)\le k$. Let \ ${\bf t}^{f+}$ be the restriction of ${\bf t}$ to
$H^{f+}$. Then there is an elimination forest $F$ of $H^{f+}$ such that
$\fwidth(F)=\fwidth(H)$ and
$\chi(\cg({\bf t}^{f+},F))\le Q$.
\end{claim}

\begin{claimproof}
By \Cref{lem:dealt-strict}, there is an optimal elimination forest $F$ of
$H^{f+}$ whose split and irregularity with respect to ${\bf t}^{f+}$ are
bounded by $\gamma_{\mathrm{str}}(\tw({\bf t}^{f+}),\fwidth(H^{f+}),r,\ell)$.
Since $\tw({\bf t}^{f+})\le\tw({\bf t})$ and
$\fwidth(H^{f+})=\fwidth(H)\le k$, monotonicity bounds both quantities by
$p_1$. Next, by \Cref{th:cgchrom} and the definition of $Q$, the chromatic
number is at most $Q$.
\end{claimproof}

For a weight assignment
$f\in\mathcal{F}_{\delta,b}$, define the weak partition
${\bf U}=(U_0,\dots,U_\ell)$ of $V(H)$ by
\[
    U_0=\{v:f(v)=0\},\qquad
    U_i=\{v:f(v)=\omega_i\}\quad\text{for }i\in[\ell].
\]
Then $(U_1,\dots,U_\ell)$ is a weak partition of $V(H^{f+})$.
Conversely, a weak partition ${\bf U}$ determines a weight assignment
by assigning weight zero to $U_0$ and weight $\omega_i$ to $U_i$.

Next, fix $0\le\lambda\le k$.  Apply the labelled version of Lemma~\ref{lem:width-condition} to
construct $\psi_{\delta,b,\lambda}(U_1,\dots,U_\ell)$ over $\tau_{ef}$.
For an elimination forest $F$ of $H^{f+}$, this formula holds on
$M(H^{f+},F)$ exactly when the $f$-width of $F$ is at most $\lambda$.
The formula depends only on the allowed weights and the threshold,
so the same formula works for every assignment $f$.
We apply the algorithm of Lemma~\ref{lem:mso-blackbox} with $Q$ and
$\psi_{\delta,b,\lambda}$ to obtain
$\varphi^+_{\delta,b,\lambda}(U_1,\dots,U_\ell)$ over $\tau_{td}$.
On $M(H^{f+},{\bf t}^{f+})$, this formula holds exactly when there is
an elimination forest $F$ of $H^{f+}$ with $f$-width at most $\lambda$
and $\chi(\cg({\bf t}^{f+},F))\le Q$.

We can express reducing to the induced subhypergraph on $V(H)\setminus U_0$ in MSO, as well as restricting the decomposition to these vertices in the straightforward way. Through this we then obtain an MSO formula \(\varphi_{\delta,b,\lambda}(U_0,\dots,U_\ell)\) on \(M(H,{\bf t})\) that is equivalent to $\varphi^+_{\delta,b,\lambda}(U_1,\dots,U_\ell)$ on the positive support of $H$ and $\bf t$.
Thus, if $M(H,{\bf t})\models\varphi_{\delta,b,\lambda}({\bf U})$, then by construction $\fwidth(H^{f+})=\fwidth(H)\le\lambda$.
If $M(H,{\bf t})\not\models\varphi_{\delta,b,\lambda}({\bf U})$, then
Claim~\ref{clm:adaptalgo} implies $\fwidth(H)>\lambda$: otherwise the
claim supplies an optimal forest satisfying both conditions above,
since $\lambda\le k$.

It remains to test this condition for every admissible assignment.
Let $\mathrm{Admissible}_{\delta,b}({\bf U})$ be an MSO formula
expressing that the sets weakly partition $V(H)$ and that their weight
assignment is a fractional independent set. For the latter condition,
enumerate the possible label counts on an edge of size at most $r$ and
retain those whose weighted sum is at most $1$. Set
\[
    \Phi_{\delta,b,\lambda}:=
    \forall U_0\cdots\forall U_\ell\,
    \bigl(\mathrm{Admissible}_{\delta,b}({\bf U})
    \rightarrow\varphi_{\delta,b,\lambda}({\bf U})\bigr).
\]
By the preceding equivalence, $M(H,{\bf t})\models\Phi_{\delta,b,\lambda}$
if and only if every $f\in\mathcal{F}_{\delta,b}$ satisfies
$\fwidth(H)\le\lambda$, which is exactly the condition that the
$\mathcal{F}_{\delta,b}$-width of $H$ is at most $\lambda$.

To recover the exact value, we note that
$\omega_i=(b+\delta(i-1))/(\delta b)$, so every bag weight is a multiple
of $1/(\delta b)$. For a fixed assignment, there are only finitely many
possible bag weights, and hence the minimum decomposition width lies on
this grid. Since $\mathcal{F}_{\delta,b}$ is finite, its supremum of
$f$-widths is a maximum and lies on the same grid. Thus, if
$w=\mathcal{F}_{\delta,b}$-width$(H)\le k$, then
\[
    w\in K=\left\{\frac{j}{\delta b}:
    0\le j\le\lfloor\delta b k\rfloor\right\}.
\]
We test $\Phi_{\delta,b,\lambda}$ for the values $\lambda\in K$ in increasing
order and return the first accepted value. Since each test accepts exactly
when $w\le\lambda$, this value is exactly $w$. If none is accepted, we reject,
and the grid property guarantees that $w>k$.

It remains to verify the FPT time bound. By our initial choice of $p$,
$\tw({\bf t})\le\lfloor\delta k\rfloor-1$. The parameter $Q$ is bounded
by a computable function of $k,b$, since $r\le\delta\le b$ and $\ell\le b$.
The same parameters bound the construction and size of
$\psi_{\delta,b,\lambda}$, by \Cref{lem:width-condition}, and therefore
of $\varphi_{\delta,b,\lambda}$ and $\Phi_{\delta,b,\lambda}$, by
\Cref{lem:mso-blackbox} and the constructions above.
For model checking, \Cref{lem:mthtw} gives
$\tw(M(H,{\bf t}))\le\tw({\bf t})+2$, so Courcelle's theorem provides
the required FPT bound. The initial decomposition may be assumed to have
$O(\|H\|)$ bags, and its encoding has parameter-bounded size times a
polynomial in $\|H\|$. Finally, there are at most $kb^2+1$ candidate values
in $K$. Thus the overall runtime is bounded by
$g(k,b)\,\|H\|^{O(1)}$ for a computable function $g$.
\end{proof}

That is, by reducing (strongly) manageable $f$-width checking to MSO, and with some care in the bookkeeping of the partition $U_1,\dots,U_\ell$, we can naturally extend the check to classes of functions. To the best of our knowledge this is the first algorithmic result for any hypergraph width notion of this kind that goes beyond the construction of double-exponential size linear programs.
As a consequence, rounding down to the allowed weights also yields an FPT approximation scheme for $\mathcal{F}^*_{\delta}$-width.

\begin{corollary}\label{cor:adapt1}
There is an algorithm that takes as input a hypergraph $H$, a positive
rational $k$, an integer $\delta\ge\rank(H)$, and a rational
$0<\varepsilon\le1$, and accepts or rejects.
If it accepts, then
$\mathcal{F}^*_{\delta}$-width$(H)\le(1+\varepsilon)k$.
If it rejects, then $\mathcal{F}^*_{\delta}$-width$(H)>k$.
The runtime is bounded by
$g(k,\delta,1/\varepsilon)\,\|H\|^{O(1)}$ for a computable function $g$.
\end{corollary}
\begin{proof}
Set $b=\lceil\delta/\varepsilon\rceil$. Then $b\ge\delta$ and
$\delta/b\le\varepsilon$.
For $f\in\mathcal{F}^*_{\delta}$, keep its zero weights and round each
positive weight down to the largest allowed weight in
$\mathcal{F}_{\delta,b}$. The resulting assignment $f'$ satisfies
$f'\le f$, so it remains a fractional independent set. Moreover, at every
vertex of positive weight,
\[
    f(v)<f'(v)+\frac1b
    \le\left(1+\frac\delta b\right)f'(v),
\]
since $f'(v)\ge1/\delta$. Applying the same inequality to every bag of
an optimal decomposition for $f'$ gives
\[
    \mathcal{F}_{\delta,b}\text{-width}(H)
    \le\mathcal{F}^*_{\delta}\text{-width}(H)
    \le\left(1+\frac\delta b\right)
        \mathcal{F}_{\delta,b}\text{-width}(H).
\]
Run the algorithm of \Cref{thm:adaptmain}, accepting if it returns a value
and rejecting otherwise. The two guarantees follow from these inequalities
and $\delta/b\le\varepsilon$. The runtime is bounded by
$g(k,b)\,\|H\|^{O(1)}$, and $b$ is bounded by a function of
$\delta$ and $1/\varepsilon$, giving the stated bound.
\end{proof}

That is, there is an FPT approximation scheme for adaptive width, excluding width functions that use weights very close to $0$. Given the very limited algorithmic results so far for adaptive width we believe that this represents a significant jump in understanding the theoretical complexity of these problems, and potentially can motivate new avenues of studying these widths. 

It is then a natural question to ask how important near-zero weights are for adaptive width, and in particular what threshold for $\delta$ is good enough such that $\mathcal{F}^*_{\delta}$-width matches adaptive width. For simple graphs $G$ we are in fact able to show that threshold $\Delta(G)+1$ is sufficient, and we therefore obtain an FPT approximation scheme for standard adaptive width on simple graphs. Curiously, this reintroduces the degree in the parameterisation; it remains open at this time if there is a deeper reason for this.

For simple graphs, a minimum positive weight depending only on the degree
already suffices to recover full adaptive width.

\begin{theorem}\label{thm:adaptive-graphs}
Let $G$ be a simple graph with $\tw(G)\ge4$.
Then $adw(G)=\mathcal{F}^*_{\Delta(G)+1}\text{-width}(G)$.
\end{theorem}

\begin{proof}
Let $D=\Delta(G)$ and $a=1/(D+1)$. Since $\tw(G)\ge4$, we have
$D\ge3$ and hence $1-a>1/2$.
Let $f$ be an arbitrary fractional independent set of $G$. We clip its
vertex weights to the interval $[a,1-a]$, defining
\[
    f^*(v)=\min\{1-a,\max\{a,f(v)\}\}.
\]
We first observe that $f^*$ is still a fractional independent set.
Indeed, consider an edge $uv$. If either endpoint has its weight
increased, its new weight is $a$, and the other endpoint has new weight
at most $1-a$. Otherwise, neither weight increases, so
$f^*(u)+f^*(v)\le f(u)+f(v)\le1$.
Thus $f^*\in\mathcal{F}^*_{D+1}$.

\begin{claim}\label{clm:adaptive-graph-clipping}
For every such $f$ and $f^*$,
\(
    \fwidth(G)\le\max\{f^*\text{-width}(G),2\}.
\)
\end{claim}

\begin{claimproof}
Let ${\bf t}=(T,B)$ be a tree decomposition of $G$ of optimal
$f^*$-width, and let $w=f^*\text{-width}(G)$.
Let $U=\{u\in V(G):f(u)>1-a\}$ be the set of vertices whose weights
were decreased. Since $1-a>1/2$, the set $U$ is independent.
For each $u\in U$ and each $v\in N(u)$, edge domination gives
$f(v)\le1-f(u)<a$, and consequently $f^*(v)=a$.
It follows that
\[
    f^*(N(u))=|N(u)|a\le Da=1-a=f^*(u).
\]

We modify ${\bf t}$ by replacing each occurrence of a vertex of $U$
by its neighbours. More precisely, for each $t\in V(T)$, put
\[
    B'(t)=(B(t)\setminus U)\cup
        \bigcup_{u\in B(t)\cap U}N(u).
\]
For each $u\in U$, choose a node $t(u)$ with $u\in B(t(u))$, and attach
a new leaf $t_u$ to $t(u)$ with bag $B'(t_u)=N[u]$.
Such a node exists because $u$ belongs to an edge of $G$.

The resulting pair is a tree decomposition of $G$. Every edge with
neither endpoint in $U$ retains a covering bag, and every edge incident
to $u\in U$ is covered by $N[u]$.
To verify the connectedness condition, consider $v\notin U$.
Its occurrences among the original nodes now consist of its original
subtree together with the occurrence subtrees of its neighbours in $U$.
Each of these additional subtrees intersects the original subtree of
$v$, since a bag of ${\bf t}$ contains both endpoints of the corresponding
edge. Their union is therefore connected.
Each new leaf containing $v$ is attached at a bag containing $v$.
Finally, each $u\in U$ occurs only in its own new leaf, since $U$ is
independent.

It remains to bound the $f$-width of this decomposition.
For each original node $t$, the replacement inequality above gives
$f^*(B'(t))\le f^*(B(t))\le w$.
Moreover, $B'(t)$ contains no vertex of $U$, so $f\le f^*$ on this bag,
and hence $f(B'(t))\le w$.
For each new leaf, edge domination and $|N(u)|\le D$ give
\[
    f(N[u])\le f(u)+D(1-f(u))
    <(1-a)+Da=\frac{2D}{D+1}<2.
\]
Thus the constructed decomposition has $f$-width at most $\max\{w,2\}$,
proving the claim.
\end{claimproof}

Set $W=\mathcal{F}^*_{D+1}\text{-width}(G)$.
The constant assignment $f(v)=1/2$ belongs to $\mathcal{F}^*_{D+1}$
and has $f$-width $(\tw(G)+1)/2\ge2$; hence $W\ge2$.
For every fractional independent set $f$, the claim and
$f^*\in\mathcal{F}^*_{D+1}$ imply
\[
    \fwidth(G)\le\max\{f^*\text{-width}(G),2\}\le W.
\]
Taking the supremum over $f$ gives $adw(G)\le W$.
The reverse inequality follows from
$\mathcal{F}^*_{D+1}\subseteq\mathcal{F}_G$.
\end{proof}

\begin{corollary}\label{cor:adaptive-graphs}
There is an algorithm that takes as input a simple graph $G$ with
$\tw(G)\ge4$, a positive rational $k$, and a rational
$0<\varepsilon\le1$, and accepts or rejects.
If it accepts, then $adw(G)\le(1+\varepsilon)k$.
If it rejects, then $adw(G)>k$.
The runtime is bounded by
$g(k,\Delta(G),1/\varepsilon)\,\|G\|^{O(1)}$ for a computable function $g$.
\end{corollary}

\begin{proof}
Apply \Cref{cor:adapt1} with $\delta=\Delta(G)+1$ and the given
$\varepsilon$. The guarantees follow from \Cref{thm:adaptive-graphs},
and the runtime has the stated dependence on $k$, $\Delta(G)$, and
$1/\varepsilon$.
\end{proof}

In this section we proved the existence of FPT algorithms and approximation schemes for variations of adaptive width. While these results all ultimately rely on MSO model checking using Courcelle's theorem and thus do not yield practical upper bounds, the mere existence of FPT algorithms for these problems strikes us as surprising and of further interest. Moreover, while not explored formally, by using \Cref{lem:dealt-strict} instead of \Cref{lem:dealt}, one directly obtains a version of \Cref{th:mainalgo} for strongly manageable functions without $\Delta(H)$ in the parameterisation.

\section{The Reduced Hypergraph Dealternation Lemma}
\label{sec:lemdealt.main}
\label{sec:lemdealt}
In this section we prove the Reduced Hypergraph Dealternation Lemma
(Lemma~\ref{lem:dealt}). We first derive it from a local statement that simplifies
an elimination forest across a separating partition. We then establish this local
statement using neutral swaps, and prove the existence of these swaps by analysing
bags along context factors. Throughout, $f$ is manageable and $\beta$ denotes its
bounded size witness, chosen monotone in its arguments.

\subsection{From Local Dealternation to the Global Lemma}
\label{sec:dealtglobal}
\label{sec:dealtgather}

The global proof uses two ingredients. 
For the first, consider that a given tree decomposition separates the hypergraph into regions (corresponding to the subtrees of the decomposition), but an optimal elimination forest may arrange the vertices of such regions in a complicated, interleaved way. The first part is then a local dealternation lemma that states that looking at one separator, we can untangle the two sides without increasing the width.  The second is a lemma that bounds irregularity in terms of split. Together this will allow us to upgrade the local statement to a global one.

We introduce the terminology needed to state these ingredients here. We then prove the irregularity bound and reduced hypergraph dealternation lemma based on them. The local dealternation step is the major technical challenge and will then be the subject of \Cref{sec:localdealtproof} to \Cref{sec:xhomogenousneutral}.

\begin{definition}[Reduced elimination forests]
An elimination forest $F$ of a hypergraph $H$ is \emph{reduced}
if for each non-leaf vertex $u$ and each child $v$ of $u$,
$u$ is adjacent in $H$ to $V(F_v)$. 
\end{definition}

\begin{lemma} \label{lem:redconnect}
Let $H$ be a hypergraph and let $F$ be a reduced elimination forest for $H$.
Then, for each $u \in V(H)$, $H[V(F_u)]$ is connected.  
Moreover, for each non-leaf vertex $u$ and each child $v$ of $u$, 
$H[V(F_v) \cup \{u\}]$ is also connected. 
\end{lemma}
\begin{proof}
By bottom-up induction on $F$.
If $u$ is a leaf then the statement is obvious. 
Otherwise, the tree factor of each child $v$ of $u$ is a connected
subset and $u$ itself is connected to each $V(F_v)$ by definition
of a reduced forest. 
\end{proof}

In what follows, we fix a weak partition $(Red,X,Blue)$ of $V(H)$
so that $H$ has no hyperedge $e$ containing both a red and a blue vertex.
We refer to the vertices of $Red$ as \emph{red vertices} and to the vertices
of $Blue$ as \emph{blue vertices}. We say that a set $U \subseteq Red \cup Blue$
is \emph{monochromatic} if either $U \subseteq Red$ or $U \subseteq Blue$.

\begin{definition}[Replacement of elimination forests] \label{def:repls}
A reduced elimination forest $F'$ is a \emph{replacement} of a reduced elimination 
forest $F$ (w.r.t. $(Red,X,Blue)$ if not clear from the context) if each monochromatic
factor of $F$ is also a monochromatic factor of $F'$. 
We say that $F'$ is a strong replacement if $\fwidth(F') \leq \fwidth(F)$. 
\end{definition}

We can then state the local dealternation lemma described above.  Its proof is given in
Section~\ref{sec:localdealtproof}.

\begin{lemma} \label{lem:localdealt}
Let $F$ be a reduced elimination forest of $H$, and put
$k=\fwidth(F)$, $r=\rank(H)$, and $d=\Delta(H)$.
Then there is a strong replacement $F'$ of $F$ such that
the maximal factorisations of $Red$ and $Blue$ under $F'$ each have size
at most
\(
    2(9+3|X|)\bigl(1+\kappa(|X|,k,r,d)\bigr),
\)
where $\kappa$ is defined in Lemma~\ref{lem:existsneutral}.
Denote this bound by $\kappa_2(|X|,k,r,d)$.
\end{lemma}

The second main ingredient follows almost immediately by adapting the proofs of Bojańczyk and Pilipczuk~\cite{DBLP:journals/lmcs/BojanczykP22} to our terminology. The main novelty in our proof is in the proof of \Cref{lem:localdealt} for manageable width functions, and in particular while assuming only weak monotone separability. 

\begin{lemma} \label{lem:transopr7.5}
Let $F$ be a rooted forest on $V(H)$ and let $W_0 \subseteq W_1 \subseteq V(H)$.
Then $|\mf_F(W_0)| \leq 9 \cdot|\mf_F(W_1)|+3\cdot |W_1 \setminus W_0|$.
\end{lemma}

\begin{proof}
Lemma~7.5 of \cite{DBLP:journals/lmcs/BojanczykP22} proves this statement
for arbitrary rooted forests. 
\end{proof}

\begin{lemma} \label{lem:transopt7.10}
Let $F$ be a reduced elimination forest of $H$. 
Let $\ell$ be the treewidth of the tree decomposition induced by $F$.
Let $X,A_1, \dots, A_p$ be a weak partition of $V(H)$ such that for $1 \leq i \neq j \leq p$,
no hyperedge of $H$ intersects both $A_i$ and $A_j$. 
Let $U \in \mf(V(H) \setminus X)$ be a context factor.
Then $U$ intersects at most $\ell+1$ sets among $A_1, \dots A_p$. 
\end{lemma}

\begin{proof}
Empty sets $A_i$ may be discarded. If $X=\emptyset$ or $V(H)\setminus X=\emptyset$,
the maximal factorisation of $V(H)\setminus X$ contains no context factor, so
the assertion is vacuous. Otherwise, Lemma~7.10
of~\cite{DBLP:journals/lmcs/BojanczykP22}, applied to the primal graph of $H$,
gives the statement: elimination forests, reducedness, and induced bags are
the same for $H$ and its primal graph.
\end{proof}

The following lemma is an upgrade of Lemma~7.11
of~\cite{DBLP:journals/lmcs/BojanczykP22}.

\begin{lemma} \label{lem:splitmir}
Let $({\bf t},F)$ be a \textsc{tf}-pair such that 
$F$ is a reduced elimination forest.
Then
\[
    \mir({\bf t},F)\le \beta(\fwidth(F),\rank(H))\cdot\bigl(9\splitf({\bf t},F)+3(\tw({\bf t})+1)\bigr).
\]
\end{lemma}

\begin{proof}
Let ${\bf t}=(T,{\bf B})$. We demonstrate existence of a function $h$ so that 
for an arbitrary node $x$ of $T$,
the irregularity of $x$ is at most $h(\splitf({\bf t},F), \tw({\bf t}),\fwidth(F),\rank(H))$. 

Let $\ell$ be the treewidth of the tree decomposition induced by $F$.
Let $y_1,\dots, y_p$ be the children of $x$. 
For each $1 \leq i \leq p$, let $A_i=cmp(y_i)$. Let $X=V(H) \setminus \bigcup_{i \in [p]} A_i$. 
By the tree-decomposition axioms, no hyperedge intersects two distinct sets $A_i$.
Now, we employ Lemma \ref{lem:transopt7.10} to observe that any maximal factor $B$ of $V(H) \setminus X$
that is a context factor intersects at most $\ell+1$ sets among $A_1, \dots, A_p$,
where $\ell+1$ is bounded by a function of $\fwidth(F)$ and $\rank(H)$. 

We have $|\mf_F(cmp(x))|\le \splitf({\bf t},F)$ by definition. Moreover, observe that
$\bigcup_{i\in[p]}A_i=cmp(x)\setminus mrg(x)$
and $|mrg(x)|\le |\bag(x)|\le \tw({\bf t})+1$.
By Lemma~\ref{lem:transopr7.5},
\[
    |\mf_F(A_1\cup\cdots\cup A_p)|
    \le 9\splitf({\bf t},F)+3(\tw({\bf t})+1).
\]
Every context factor of \(\mf_F(A_1\cup\cdots\cup A_p)\) intersects at most
\(\beta=\beta(\fwidth(F),\rank(H))\) of the sets \(A_i\), by
Lemma~\ref{lem:transopt7.10}. Hence the number of indices \(i\) for which
\(A_i\) intersects a context factor is at most
$\beta\bigl(9\splitf({\bf t},F)+3(\tw({\bf t})+1)\bigr)$.

Let $I\subseteq[p]$ be the set of indices $i$ for which $A_i$ intersects
a context factor of $\mf_F(A_1\cup\cdots\cup A_p)$.
We claim that if $A_i$ is not well-formed, then $i \in I$, thus upper-bounding the irregularity of
$x$ and hence the maximum irregularity since $x$ is taken arbitrarily. It remains to prove the claim. 
The union of tree and forest factors in the factorisation of $\bigcup_{i \in [p]} A_i$
can be seen as the union of tree factors $B_1, \dots, B_r$. 
Let $i \in [p] \setminus I$. We are going to prove that $A_i$ is well-formed;
equivalently, its maximal factorisation contains only forest factors.

Let $J=\{j \in [r]: B_j \cap A_i \neq \emptyset\}$. 
We observe that for each $j \in J$, $B_j \subseteq A_i$. Indeed, $B_j \subseteq A_1 \cup \dots \cup A_p$.
On the other hand, since $H[B_j]$ is connected, $B_j$ cannot intersect two distinct $A_i$. 
It follows that $A_i=\bigcup_{j \in J} B_j$. Now, assume for the sake of contradiction that a maximal
factor $C$ of $A_i$ is a context factor. Let $u$ be the root of $C$. As we have observed $u \in B_j$
for some $j \in J$. But then $V(F_u) \subseteq B_j$. That is $C \subset V(F_u) \subseteq B_j$ in contradiction
to $C$ being maximal. 
\end{proof}

With the main ingredients established we can then combine them to obtain the global reduced hypergraph dealternation lemma.  The argument follows Section 7.4 of~\cite{DBLP:journals/lmcs/BojanczykP22}, with adaptations to preserve optimality for the general width function $f$ and to track the dependence on decomposition width, optimal $f$-width, rank, and degree separately.

\begin{proof}[Proof of the Reduced Hypergraph Dealternation Lemma (Lemma \ref{lem:dealt})]
Choose an elimination forest $F$ with $\fwidth(F)=\fwidth(H)$.
First, we assume that $F$ is reduced because the transformation from an arbitrary
forest to a reduced one can be done so that each bag of the tree decomposition induced by the resulting
forest is a subset of a bag of the tree decomposition induced by the original forest 
(Lemma~7.8 of \cite{DBLP:journals/lmcs/BojanczykP22}). 
Consequently, the $\fwidth$ of the resulting reduced forest does not increase,
and hence it remains optimal.

The next step is to define \emph{canonical} Red-Blue decompositions
w.r.t. the nodes $x$ of $T$ (recall that ${\bf t}=(T,B)$). 
In particular, the decomposition w.r.t. $x$ is $(cmp(x),adh(x),V(H) \setminus (cmp(x) \cup adh(x)))$. 
We let $(x_1, \dots, x_m)$ be the sequence of nodes of $T$ being explored in a bottom-up manner,
that is children occur in this sequence before their parents. 
We observe the following. 
\begin{claim} \label{clm:inv111}
Let $1 \leq i <j \leq m$. Then $cmp(x_i)$ is monochromatic for
the canonical red-blue decomposition w.r.t. $x_j$.
\end{claim}

\begin{claimproof}
If $x_j$ is a predecessor of $x_i$ then $cmp(x_i) \subseteq cmp(x_j)$ 
hence the claim holds. 
Otherwise, we observe that $cmp(x_i) \subseteq V(H) \setminus (adh(x_j) \cup cmp(x_j))$. 
Indeed $cmp(x_i)$ does not intersect with $\bag(x)$ if $x$ is not a successor of $x_i$. We note that $adh(x_j) \cup cmp(x_j)$ is a subset of the union of $\bag(x_j)$ and the bags of 
the proper successors of $x_j$. By assumption $x_j$ is not a predecessor of $x_i$.
By assumption about $(x_1, \dots, x_m)$  $x_j$ is not a successor of $x_i$ either.
It follows that no successor of $x_j$ is a successor of $x_i$. Hence the union of the bags of successors of $x_j$
is disjoint from $cmp(x_i)$ as required.  
\end{claimproof}

Let $F_0=F$ and let $F_1, \dots, F_m$ be reduced elimination forests
of $H$ defined as follows. For each $1 \leq i \leq m$, choose $F_i$ to be a strong
replacement of $F_{i-1}$ supplied by Lemma~\ref{lem:localdealt} for the canonical
red-blue decomposition for $x_i$.
Set
\[
    \kappa_3(w,k,r,d):=\kappa_2(w+1,k,r,d).
\]
For every such decomposition, $X=adh(x_i)$ and
$|X|\le |\bag_{\mathbf t}(x_i)|\le \tw({\bf t})+1$.
Moreover, each intermediate forest has $f$-width at most $\fwidth(F)$.
Monotonicity of $\kappa_2$ therefore gives a uniform bound at every step.
We are going to prove by induction on $1 \leq i \leq m$ that,
for each $j \leq i$, $cmp(x_j)$ has maximal factorisation w.r.t. $F_i$ of size
at most $\kappa_3(\tw({\bf t}),\fwidth(F),\rank(H), \Delta(H))$.

Let $i=1$. Then this holds for $cmp(x_1)$ simply by Lemma \ref{lem:localdealt}.
Since $i$ is the smallest index, the remaining part of the statement is vacuously true.
Now assume that $i>1$. First, let $j=i$. Then the statement for $j$ is immediate from
Lemma \ref{lem:localdealt}.
Assume now that $j<i$. By the induction assumption, the statement about $x_j$ is
true regarding $F_{i-1}$. As $F_i$ is a replacement of $F_{i-1}$ and, by Claim \ref{clm:inv111},
$cmp(x_j)$ is monochromatic for the canonical red-blue decomposition for $x_i$, 
a witnessing factorisation for $cmp(x_j)$ w.r.t. $F_{i-1}$ remains so w.r.t. $F_i$.

It follows from the above reasoning that
$$\splitf({\bf t},F_m) \leq \kappa_3(\tw({\bf t}), \fwidth(F),\rank(H), \Delta(H))$$
as required. Since each replacement is strong,
$\fwidth(H)\le\fwidth(F_m)\le\fwidth(F)=\fwidth(H)$, so $F_m$ is optimal.
It remains to add that $\mir({\bf t},F_m)$ is also upper-bounded
by a function of the same parameters by Lemma \ref{lem:splitmir}. 
\end{proof}

\subsection{Local Dealternation by Repeated Swaps}
\label{sec:localdealtproof}

To prove Lemma~\ref{lem:localdealt}, we will begin by investigating red/blue colourings of the vertices on the two sides of
a separator $X$. For such colourings we will seek to rearrange the forest so that each
colour occupies only boundedly many factors. Reducedness of the elimination forest will allow us to focus on context factors, where the two colours may alternate along the spine. At a high level, we will show how to group
consecutive spine vertices of the same colour into intervals and simplify this
alternation by swapping neighbouring intervals so that intervals of the same
colour merge. At the same time we show that these swaps preserve the required forest structure and all existing monochromatic factors. Most importantly, and the main challenge in this whole process, is to ensure that our swaps do not increase $f$-width. 
In this section we will establish the first two structural properties of swaps and show how these properties, together with neutrality in terms of $f$-width, combine to establish the bounded monochromatic factorisations required by
Lemma~\ref{lem:localdealt}. The two subsequent sections will then establish that this can be established using \emph{neutral} swaps, i.e., swaps that do not increase $f$-width.

For $u \in Red \cup Blue$, an important part of this reasoning is which  of $Red,Blue$ is of the same colour as $u$ and which is of the opposite colour.
For this purpose, we introduce the notation $Native(u)$ and $Foreign(u)$. In particular, if $u \in Red$ then $Native(u)=Red$ and $Foreign(u)=Blue$. 
If, on the other hand, $u \in Blue$ then $Native(u)=Blue$ and $Foreign(u)=Red$.

Recall that $(Red,X,Blue)$ partitions $V(H)$, with $X$ separating the red and blue
vertices, and that $F_u$ denotes the subtree of $F$ rooted at $u$.
Intuitively, a whole subtree containing no separator vertices cannot contain both
colours as reducedness makes its vertex set connected in $H$, whereas red and blue
vertices can only be connected through the separator vertices $X$. The following lemma, which formalises this, and its variant shown afterwards will be important to guide our approach.

\begin{lemma} \label{lem:redmono}
Let $H$ be a hypergraph and let $F$ be a reduced elimination forest for $H$.
Let $U \subseteq Red \cup Blue$. 
If there is $u \in V(H)$ such that $U=V(F_u)$ then
$U$ is monochromatic. Furthermore, if there is $u \in V(H)$
and a child $v$ of $u$ such that $U=V(F_v) \cup \{u\}$  then $U$
is also monochromatic. 
\end{lemma}

\begin{proof}
By Lemma \ref{lem:redconnect}, $U$ is connected. 
On the other hand, there is no edge connecting $U \cap Red$
and $U \cap Blue$. Hence, one of these sets must be empty.
\end{proof}

In order to proceed, we need to recall, for a context factor $B$ of $F$,
the notions of the root $rt(B)$ of $B$, the parent $pr(B)$ of the appendices
of $B$ and the spine $spine(B)$.
Let us note that a context factor is uniquely determined
by its spine and the appendices. 
Indeed, let $P$ be a directed path of $F$ whose final vertex 
is not a leaf. Let $W$ be a non-empty subset of children of the 
final vertex of $P$. 
Then the context factor $Cont_F(P,W)$ is obtained by adding
to $P$ each $V(F_u)$ such that $u$ is either a child of a non-final
vertex of $P$ other than its successor on $P$ or a child of the final
vertex of $P$ that does not belong to $W$.

\begin{lemma} \label{lem:monspine}
Suppose that $F$ is reduced. 
Let $P$ be a directed path of $F$ whose final vertex 
is not a leaf. Let $W$ be a non-empty subset of children of the 
final vertex of $P$. 
Suppose that $Cont_F(P,W) \subseteq Red \cup Blue$
and $V(P)$ is monochromatic. 
Then $Cont_F(P,W)$ is monochromatic. 
\end{lemma} 

\begin{proof}
It is enough to show that
for each $u \in V(P)$ and for each child $v$ of $u$ such 
that $V(F_v) \subseteq Cont_F(P,W)$,
$V(F_v) \cup \{u\}$ is monochromatic. 
Since, by definition, $V(F_v) \cup \{u\} \subseteq Red \cup Blue$,
this is immediate from Lemma \ref{lem:redmono}.
\end{proof}

With these initial simplifying observations on our colourings in hand, we move on to defining the main object of our swapping procedure, colour intervals and interval factors. These will be the objects that we end up swapping.

\begin{definition}[Colour intervals, interval factors]
 \label{def:contmono}
Let $U \subseteq Red \cup Blue$ be a context factor of an elimination forest $F$ 
of a hypergraph $H$.
Let $P$ and $W$ be the spine and the appendices of $U$ 
respectively. 
We call a maximal monochromatic subpath of $P$ a \emph{(colour) interval}.
In other words $P=P_1+ \dots+P_q$ where $P_1, \dots, P_q$ are the intervals of $P$. 
For each $1 \leq i \leq q$, we introduce a set $W_i$ of \emph{local appendices} 
as follows. 
$W_q=W$ and for each $1 \leq i \leq q-1$, $W_i$ is a singleton containing  the first vertex of $P_{i+1}$. 
We call the factors $Cont_F(P_i,W_i)$ for $1 \leq i \leq q$ the \emph{interval factors} of $U$. 
\end{definition}

\begin{corollary} \label{cor:monspine}
With the notation as in Definition \ref{def:contmono}, and
assuming that $F$ is reduced, each interval factor of $U$
is monochromatic. 
\end{corollary}

\begin{proof}
Immediate from Lemma \ref{lem:monspine}.
\end{proof}

A swap exchanges two consecutive internal colour intervals. Since this will imply that the swapped intervals become adjacent to intervals of the same colour, the neighbouring
intervals then merge, reducing the number of intervals by two, see also~\Cref{fig:swap_main}. 
As mentioned before, establishing the existence of swaps that do not increase $f$-width is the final goal of our arguments. We refer to them as \emph{neutral swaps}.

\begin{definition}[Swaps] \label{def:singleswap}
With the notation of Definition \ref{def:contmono}, 
assume that $q \geq 4$. 
Let $1<i<q-1$. 
Let $(u_1,u_2)$ be the edge from the last vertex of $P_{i-1}$ to
the first vertex of $P_i$. 
Let $(v_1,v_2)$ be the edge from the last vertex 
of $P_i$ to the first vertex of $P_{i+1}$.
Finally, let $(w_1,w_2)$ be the edge connecting the
last vertex of $P_{i+1}$ and the first vertex of $P_{i+2}$. 
Let $F'$ be obtained from $F$ by removal of edges 
$(u_1,u_2)$, $(v_1,v_2)$, and $(w_1,w_2)$
and adding edges $(u_1,v_2)$, $(w_1,u_2)$, $(v_1,w_2)$. 
We call $F'$ a \emph{swap} of $F$ and the \emph{intervals
being swapped} are $P_i$ and $P_{i+1}$. 
We say that the swap is \emph{neutral} if $\fwidth(F') \leq \fwidth(F)$.
In particular, in this case we say that $U$ \emph{admits} a neutral swap
(of intervals $P_{i}$ and $P_{i+1}$). 
\end{definition}

We first verify the structural properties of a swap. These will also be used
later when analysing its effect on width.

\begin{lemma} \label{lem:singleswap}
Let $F$ be a reduced elimination forest. Let $B\subseteq Red\cup Blue$ be a context factor of $F$
with the spine $P$ and appendices $W$. 
Let $P_1, \dots, P_q$ be the colour intervals of $P$. 
Let $1 \leq i \leq q-3$.
Let $F'$ be the rooted forest obtained from $F$ by swapping $P_{i+1}$ and $P_{i+2}$.
Then the following statements hold regarding $F'$. 
\begin{enumerate}
\item $F'$ is a reduced elimination forest. 
\item Each monochromatic factor of $F$ remains a monochromatic factor of $F'$ in
the same capacity (that is a tree factor remains a tree factor, a
forest factor remains a forest factor, and a context factor remains a context factor). 
\item Each context factor of $F$ disjoint from $B$ remains a context factor in $F'$
with the same spine and appendices. 
\item $B$ remains a context factor of $F'$ with the same appendices but with the spine
obtained from $P$ by swapping $P_{i+1}$ and $P_{i+2}$ as in Definition \ref{def:singleswap}.
In particular, the number of colour intervals of $B$ w.r.t. $F'$ reduces by $2$.  
\end{enumerate} 
\end{lemma}

\begin{proof}
First, it is not hard to observe that
\begin{claim} \label{clm:singleswap1}
$F$ and $F'$ have the same set of leaves. 
\end{claim}

The swap preserves ancestor relations within both $Red\cup X$ and $Blue\cup X$.
Indeed, the two swapped interval factors are monochromatic by
Corollary~\ref{cor:monspine}, and $B$ contains no vertex of $X$.
Since every hyperedge is contained in one of these two sets, $F'$ is an
elimination forest.

To check reducedness, let $a$ be the last vertex of $P_i$, let $b,c$ be the
first and last vertices of $P_{i+1}$, let $d,e$ be the first and last vertices
of $P_{i+2}$, and let $z$ be the first vertex of $P_{i+3}$.
The new parent--child edges are $(a,d)$, $(e,b)$, and $(c,z)$.
For $(a,d)$, we have $V(F'_d)=V(F_b)$, and for $(e,b)$, we have
$V(F_z)\subseteq V(F'_b)$. Thus the old adjacency witnesses for $(a,b)$ and
$(e,z)$ also witness these two new edges.
For $(c,z)$, the set $V(F_d)\setminus V(F_z)$ is the interval factor of
$P_{i+2}$, which is monochromatic of the opposite colour to $c$.
Consequently, the old adjacency witness for $(c,d)$ belongs to
$V(F_z)=V(F'_z)$ and witnesses $(c,z)$.
For every unchanged parent--child edge, the child subtree either remains
unchanged, gains vertices, or loses only an interval factor of the opposite
colour to the parent. Its old adjacency witness therefore remains in the child
subtree. This proves reducedness.

We further observe that if $V(F_u)$ is monochromatic then for each vertex of $V(F_u)$,
its incident forest edges remain unchanged in $F'$. Hence the monochromatic tree factors of $F$ remain
ones in $F'$. Further on, we notice that if $V(F_u)$ is monochromatic then its parent remains
the same in $F'$. We thus conclude that monochromatic forest factors of $F$
remain ones in $F'$.

In order to demonstrate that monochromatic context factors are invariant in $F'$,
we first observe the following. 
\begin{claim} \label{clm:singleswap2}
A monochromatic path of $F$ remains a monochromatic path of $F'$.
\end{claim}

Now, let $U$ be a monochromatic context factor of $F$ and let $Q$
be its spine. By Claim \ref{clm:singleswap2}, $Q$ is monochromatic in $F'$. 
Every tree factor attached to $Q$ and contained in $U$, together with its attachment edge, remains unchanged.
Any changed child edge from $Q$ leaves its final vertex and leads outside $U$, both before and after the swap.
Thus $F[U]=F'[U]$, and the nonempty set of children leaving $U$ remains confined to the final vertex of $Q$.
Finally, the final vertex of $Q$ is not a leaf by 
Claim \ref{clm:singleswap1}. We conclude that $U$ remains a monochromatic context
factor in $F'$. 

Let $U$ be a context factor of $F$ disjoint from $B$.
We observe that, by construction $F[U]=F'[U]$ and the neighbourhood of the spine of $U$
is also invariant. Hence, the penultimate statement of the lemma holds. 

The last item of the lemma is immediate by construction. 
\end{proof}

To iterate swaps within a fixed factorisation, we use the following relative
notion of strong replacement.

\begin{definition} \label{def:relstrongrepl}
Let $F$ be a reduced elimination forest of $H$ and let 
${\bf U}=\{U_1, \dots U_q\}$ be a factorisation of $Red \cup Blue$
w.r.t. $F$ where each $U_i$ is a monochromatic tree or forest factor
or a (not necessarily monochromatic) context factor.
A reduced elimination forest $F'$ is an \emph{strong replacement} 
of $F$ w.r.t. ${\bf U}$ if the following conditions hold. 
\begin{enumerate}
\item $\fwidth(F') \leq \fwidth(F)$. 
\item Each monochromatic factor of $F$ remains a monochromatic factor in
the same capacity (that is a tree factor remains a tree factor, a
forest factor remains a forest factor, and a context factor remains a context factor). 
\item Each context factor of $F$ that is an element of ${\bf U}$ remains a context factor in $F'$.
\end{enumerate}
\end{definition}

\begin{remark} \label{rem:strongrepl}
Let us make three remarks related to the notion of a strong replacement. 
First, if $F'$ is a strong replacement of $F$ w.r.t. a factorisation
${\bf U}$ then ${\bf U}$ is also a factorisation of $Red\cup Blue$
w.r.t. $F'$ where each factor
retains its capacity. Second, the relation is transitive. That is,
if $F''$ is a strong replacement of $F'$ w.r.t. ${\bf U}$ then $F''$
is a strong replacement of $F$ w.r.t. ${\bf U}$. 
Finally, $F'$ is a strong replacement of $F$ in  the sense of Definition \ref{def:repls}. 
\end{remark}

Our development on swaps then concludes in the already announced existence of neutral swaps. In particular, we show that such swaps exist whenever a context factor has sufficiently many colour intervals. The proof of the following lemma will be the subject of the following two sections. In the meantime we refer forwards to the bounds $\kappa_0$ and
$\kappa_1$ which are specified in Lemmas~\ref{lem:boundedxintervals}
and~\ref{lem:homogenousneutral}.

\begin{restatable}{lemma}{NeutralSwap}\label{lem:existsneutral}
Let $k\ge 0$ and let $F$ be a reduced elimination forest of $H$ with $\fwidth(F)\le k$,
and let $U\subseteq Red\cup Blue$ be a context factor.
If the spine of $U$ has more than %
\[
    \kappa(|X|,k,\rank(H),\Delta(H))
    =
    \kappa_0(|X|,k,\rank(H),\Delta(H)) \cdot
    \kappa_1(k,\rank(H))
\]
monochromatic intervals, then $U$ admits a neutral swap.
\end{restatable}

Assuming this lemma, we can repeatedly simplify context factors while
preserving the other factors and keeping the width bound fixed. This then establishes the reduced hypergraph dealternation lemma.

\begin{lemma} \label{lem:seqswaps}
Let $k\ge 0$ and let $F$ be a reduced elimination forest of $H$ with
$\fwidth(F)\le k$. Let
${\bf U}=\{U_1, \dots U_q\}$ be a factorisation of $Red \cup Blue$
w.r.t. $F$ where each $U_i$ is a monochromatic tree or forest factor
or a (not necessarily monochromatic) context factor.
Then there is a strong replacement $F'$ of $F$ w.r.t. ${\bf U}$ where the number of colour
intervals in the spine of each context factor $U \in {\bf U}$ 
is at most $\kappa(|X|, k,\rank(H), \Delta(H))$
where $\kappa$ is the function as in Lemma \ref{lem:existsneutral}. 
\end{lemma}

\begin{proof}
\sloppy
Keep $k$ fixed throughout the induction. For each $i \in [q]$ define the \emph{excess} $x_i$ as follows.
If $U_i$ is not a context factor or if $U_i$ is a context factor
but the number $y_i$ of colour intervals of its spine is at most 
$\kappa(|X|, k,\rank(H), \Delta(H))$ then $x_i=0$.
Otherwise, we let $x_i=y_i-\kappa(|X|, k,\rank(H), \Delta(H))$.
We prove the lemma by induction on $x=\sum_{i \in [q]} x_i$ which we refer to 
as the \emph{excess} of ${\bf U}$ w.r.t. $F$.
If $x=0$ then simply set $F'=F$. 

So, we assume that  $x>0$. 
We are going to demonstrate existence of a strong replacement $F^*$
of $F$ w.r.t. ${\bf U}$ such that the excess of ${\bf U}$ w.r.t. $F^*$ is smaller than w.r.t. $F$. 
The lemma will then follow from the induction assumption combined 
with the transitivity of strong replacement w.r.t. ${\bf U}$ as specified
in Remark \ref{rem:strongrepl}. 

Since $x>0$, we identify a context factor $U \in {\bf U}$ with the
number of colour intervals greater than 
$\kappa(|X|, k,\rank(H), \Delta(H))$.
By Lemma \ref{lem:existsneutral}, there is an elimination forest $F^*$
obtained from $F$ by a neutral swap of colour intervals of $U$. 
We first note that $\fwidth(F^*) \leq \fwidth(F)$ simply by definition of a neutral swap; in particular, $\fwidth(F^*)\le k$,
so the induction uses the same threshold.
Then we employ Lemma \ref{lem:singleswap} to conclude that $F^*$ is a reduced
elimination forest and that the remaining properties of a strong replacement 
w.r.t. ${\bf U}$ hold as well. It also follows from Lemma \ref{lem:singleswap} that the number
of colour intervals of $U$ w.r.t. $F^*$ is reduced by $2$ and the number of such 
intervals in the other context factors of ${\bf U}$ remains the same.
Thus we conclude that the excess of ${\bf U}$ w.r.t. $F^*$ is smaller than w.r.t. $F$
as required. 
\end{proof}

\begin{proof}[Proof of Lemma \ref{lem:localdealt}]
Let ${\bf U}_0:=\mf_F(Red\cup Blue)$.
Let ${\bf U}$ be the factorisation obtained from ${\bf U}_0$ by replacing 
each non-monochromatic forest factor by two monochromatic forest factors.
The replacement is possible by Lemma \ref{lem:redmono}.
Further on by Lemma \ref{lem:redmono}, each non-context factor of ${\bf U}$ 
is monochromatic. We conclude that ${\bf U}$ satisfies the properties 
of a factorisation in Definition \ref{def:relstrongrepl}. 
We thus apply Lemma \ref{lem:seqswaps} with $k=\fwidth(F)$ to conclude existence of a
strong replacement $F'$ of $F$ w.r.t. ${\bf U}$ such that the spine of each context
factor of ${\bf U}$ w.r.t. $F'$ has at most 
$\kappa(|X|, \fwidth(F),\rank(H), \Delta(H))$ colour intervals.
As each context factor is the union of its interval factors, 
it follows from Corollary \ref{cor:monspine} that there is a factorisation of $Red \cup Blue$
w.r.t. $F'$ into at most $|{\bf U}|+|{\bf U}| \cdot \kappa(|X|, \fwidth(F),\rank(H), \Delta(H))$
monochromatic factors. 

By Lemma~\ref{lem:transopr7.5}, applied with
\(W_0=Red\cup Blue=V(H)\setminus X\) and \(W_1=V(H)\), we have
\[
    |{\bf U}_0|=|\mf_F(Red\cup Blue)|\le 9+3|X|.
\]
After splitting each non-monochromatic forest factor into two monochromatic
forest factors, $|{\bf U}|\le 2(9+3|X|)$.
Hence the final monochromatic factorisation has size at most
\[
    2(9+3|X|)
    \bigl(1+\kappa(|X|,\fwidth(F),\rank(H),\Delta(H))\bigr)
    =\kappa_2(|X|,\fwidth(F),\rank(H),\Delta(H)).
\]

 Finally, in light of Remark~\ref{rem:strongrepl}, we conclude that $F'$ is a strong replacement of $F$ as in  \Cref{def:repls}. 
\end{proof}

\subsection{Bag Decomposition and Bounding $X$-Intervals}
\label{sec:dealtballs}
\label{sec:xintervalbounded}

To find a swap that does not increase $f$-width, we first isolate the parts of
each bag that interact with the separator $X$. We partition bags into unions of
connected components: those meeting $X$, called the $X$-ball, and the remaining
components of either colour. This partition gives us edge-disjoint pieces to
which weak monotone separability applies. We then divide the spine into regions
on which the $X$-ball is constant. Our aim in this subsection is to bound the
number of these regions. Bounded bag size, rank, and degree restrict the possible
$X$-balls, and we show that their recurrence along the spine is also controlled.
Consequently, if the spine has sufficiently many colour intervals, many of them
must lie in a single region with a fixed $X$-ball. The next section will then use
this fixed interaction with the separator to find a neutral swap.

Throughout this subsection, fix a hypergraph $H$, a weak partition
$(Red,X,Blue)$ of $V(H)$ such that no hyperedge meets both $Red$ and $Blue$,
a number $k\ge 0$, and a reduced elimination forest $F$ of $H$ with
$\fwidth(F)\le k$. Also fix a context factor $U\subseteq Red\cup Blue$
of $F$, with spine $P$ and appendices $W$. All definitions and lemmas below
use this fixed setup unless explicitly stated otherwise. Let
\[
    p=\max\{1,\lceil\beta(k,\rank(H))\rceil\}.
\]
The bounded size condition gives $|\bag_F(u)|\le p$ for every $u\in V(H)$.
We first distinguish the parts of a bag by their relation to the separating
partition.

\begin{definition} [$X$-bags, native bags, and foreign bags]
Let $u \in V(H)$. We refer to $\bag_F(u) \cap X$ as the $X$-bag of $u$
and denote it as $X\text{-bag}_F(u)$. 
For $u \in Red \cup Blue$, we refer to $\bag_F(u) \cap Native(u)$
as the \emph{native bag}
of $u$ denoted by $nbag_F(u)$ and to $\bag_F(u) \cap Foreign(u)$ as the \emph{foreign}
bag of $u$ denoted by $fbag_F(u)$. 
\end{definition}

Weak monotone separability does not apply to an arbitrary partition of a bag.
We instead partition the bag into unions of connected components, so that the
resulting parts are edge disjoint.

\begin{definition}[$X$-component, native component, foreign component]
Let $u\in Red\cup Blue$ and let $W$ be a connected component of $H[\bag_F(u)]$.
We say that $W$ is an $X$-\emph{component} if $W \cap X \neq \emptyset$. 
Note that otherwise $W$ is monochromatic as no hyperedge of $H$ connects
both  a red and a blue vertex. If $W$ is of the same colour as $u$
then we say that $W$ is a \emph{native component} of $u$. 
Otherwise, $W$ is a \emph{foreign component} of $u$. 
\end{definition}

\begin{definition}[$X$-ball, native ball, foreign ball]
We call the union of all $X$-components of $H[\bag_F(u)]$
the \emph{$X$-ball} of $u$ (w.r.t. $F$) and denote it by $X\text{-ball}_F(u)$. 
Further on, we refer to the union of all native components 
of $H[\bag_F(u)]$ as the \emph{native ball} of $u$ (w.r.t. $F$) and
denote it by $nball_F(u)$. 
Finally, we refer to the union of foreign components of $\bag_F(u)$
as the \emph{foreign ball} of $u$ (w.r.t. $F$) and denote it by $fball_F(u)$. 
\end{definition}

\begin{proposition} \label{prop:strongadd}
Let $G$ be any elimination forest of $H$ and let $u\in Red\cup Blue$.
Put $A=X\text{-ball}_G(u)\cup nball_G(u)$. Then
\[
    f_H(\bag_G(u))=
    \begin{cases}
        f_H(A), & fball_G(u)=\emptyset,\\
        g_f\bigl(f_H(A),f_H(fball_G(u))\bigr),
            & fball_G(u)\neq\emptyset.
    \end{cases}
\]
\end{proposition}
\begin{proof}
The set $A$ contains $u$ and is therefore nonempty.
If the foreign ball is empty, then $\bag_G(u)=A$.
Otherwise, $A$ and the foreign ball form an edge-disjoint partition
of $\bag_G(u)$ into two nonempty sets, since they are unions of
distinct connected components of $H[\bag_G(u)]$.
Apply weak monotone separability.
\end{proof}

We partition the spine according to changes in its $X$-ball. Within each
resulting region, the interface with $X$ is fixed.

\begin{definition}[$X$-interval]
With the notation as in Definition \ref{def:contmono},
let $P'$ be a maximal subpath of $P$ such that $X\text{-ball}_F(u)$ is the same for all $u \in V(P')$. We refer to $P'$ as an $X$-interval of $U$. 
\end{definition}

\begin{definition}[$X$-homogeneous] \label{def:homofactor}
Let $U_0 \subseteq Red \cup Blue$ be a context factor and let
$P_0$ be the spine of $U_0$. We say that $U_0$ is $X$-\emph{homogeneous}
if all the vertices of $P_0$ have the same $X$-ball. 
\end{definition}

We now bound the number of possible $X$-balls, and then control their recurrence
along the spine. This is the part of the argument where bounded degree is used.

\begin{lemma}\label{lem:boundedint1}
For each $u \in V(P)$, 
each vertex $w \in X\text{-ball}_F(u)$ is at distance at most $p-1$
from some vertex of $X$. 
\end{lemma}

\begin{proof}
Indeed, let $w \in X\text{-ball}_F(u)$. Then there is a connected component $W$ 
of $H[\bag_F(u)]$ such that $W \cap X \neq \emptyset$ and $w \in W$. 
It follows that there is a path $Q$ in the primal graph of $H[\bag_F(u)]$ between $w$ and some vertex $x \in X$
such that $V(Q) \subseteq \bag_F(u)$. 
It follows from the bounded size condition that $p \geq |\bag_F(u)| \geq |V(Q)|$.
In other words $Q$ contains at most $p-1$ edges and hence $w$ is at distance at most 
$p-1$ from $X$ as required. 
\end{proof}

\begin{lemma} \label{lem:boundedint2}
Let ${\bf Y}$ be the set of all $X$-balls of the vertices of $P$. Then
\[
    |{\bf Y}|\le
    \left(1+|X|\bigl(1+\Delta(H)(\rank(H)-1)\bigr)^{p-1}\right)^p.
\]
\end{lemma}

\begin{proof}
Let $W(X)$ be the set of vertices at distance at most $p-1$ from $X$
in the primal graph of $H$. Its maximum degree is at most
$\Delta(H)(\rank(H)-1)$, so
\[
    |W(X)|\le |X|\sum_{j=0}^{p-1}[\Delta(H)(\rank(H)-1)]^j
    \le |X|\bigl(1+\Delta(H)(\rank(H)-1)\bigr)^{p-1}.
\]
By Lemma \ref{lem:boundedint1} and the bounded size condition, each element
of ${\bf Y}$ is a subset of $W(X)$ of size at most $p$.
Fix an ordering of $W(X)$. Encode each such subset by its ordered list of
vertices, padded to length $p$ with a dummy symbol. There are at most
$(|W(X)|+1)^p$ such lists, giving the desired bound.
\end{proof}

For the remainder of this subsection, let $P'_1,\dots,P'_q$ be the
$X$-intervals of the fixed spine $P$, listed from top to bottom, and let
$X_i$ be the common $X$-ball of the vertices of $P'_i$.
Had $X_1, \dots, X_q$ been all distinct,  Lemma \ref{lem:boundedxintervals}
would have immediately followed from Lemma \ref{lem:boundedint2}.
However, this is not necessarily so as, for instance, a sequence 
$\{v_1\}, \{v_1,v_2\}$ and then back $\{v_1\}$
is perfectly possible. This matter can be addressed through the notion of \emph{peaks}.

\begin{definition}
We say that $X_i$ is a \emph{peak} if $i=q$ or 
if $i<q$ and $X_i \setminus X_{i+1} \neq \emptyset$.  
\end{definition}

\begin{lemma} \label{lem:boundedint3}
In the fixed sequence $X_1,\dots,X_q$, if $i<j$ and both $X_i$ and $X_j$
are peaks, then $X_i\neq X_j$.
\end{lemma}

\begin{proof}
Indeed, assume  the opposite. 
Then there are $1 \leq  i<j \leq q$ such that $X_i$ and $X_j$ are peaks
and $X_i=X_j$. 
Let $v \in X_i \setminus X_{i+1}$. 
We note that $X_i \cap X_{i+1} \neq \emptyset$.
Indeed, by definition $X_i \cap X=X_j \cap X \neq \emptyset$.
By the connectivity property of tree decompositions, $X_i \cap X \subseteq X_{i+1}$.
Let $Q_v$ be a path from $X_i \cap X$ to $v$ in $H[X_i]$. 

\begin{claim} \label{clm:boundedint31}
Let $w \in V(P'_{i+1})$. 
Then $V(Q_v)$ is not a subset of $\bag_F(w)$. 
\end{claim}

\begin{claimproof}
Indeed, assume the opposite.
Then in $H[\bag_F(w)]$ $v$ is connected to a vertex of $X$ and hence belongs to $X\text{-ball}_F(w)=X_{i+1}$,
a contradiction. 
\end{claimproof}

Note a subtle point of Claim \ref{clm:boundedint31}.
We do not claim that $v \notin \bag_F(w)$; we claim 
that at least one vertex of $V(Q_v)$ is outside $\bag_F(w)$.

It follows from the tree decomposition axioms that all the vertices of $P$
whose bags w.r.t. $F$ contain $V(Q_v)$ as a subset form a subpath $P_v$ of $P$.  
By assumption both $P'_i$ and $P'_j$ are subpaths of $P_v$. 
But then for each integer $a>0$ such that $i+a \leq j$, 
$P'_{i+a}$  is also a subpath of $P_v$.
In particular, $P'_{i+1}$ is a subpath of $P_v$.
However, this contradicts Claim \ref{clm:boundedint31}.
\end{proof}

Combining the bound on possible balls with the distinctness of peaks gives
the required bound on the number of $X$-intervals.

\begin{lemma} \label{lem:boundedxintervals}
The number of $X$-intervals of $U$ is at most
\[
    \kappa_0(|X|,k,\rank(H),\Delta(H)) := (p+1)\left(1+|X|\bigl(1+\Delta(H)(\rank(H)-1)\bigr)^{p-1}\right)^p,
\]
where $p=\max\{1,\lceil\beta(k,\rank(H))\rceil\}$.
\end{lemma}

\begin{proof}[Proof of Lemma \ref{lem:boundedxintervals}] 
We partition $[q]$ into maximal intervals $I_1, \dots I_a$
such that each $I_b$ is a chain in the following sense.
Let $b_1<b_2 \in I_b$. Then $X_{b_1} \subseteq X_{b_2}$. 
Let $j$ be the final index of $I_b$. Then, due to the
maximality of $I_b$, $X_j$ is a peak.
By Lemmas \ref{lem:boundedint2} and \ref{lem:boundedint3},
\[
    a\le \left(1+|X|\bigl(1+\Delta(H)(\rank(H)-1)\bigr)^{p-1}\right)^p.
\]
Since consecutive $X$-intervals have distinct $X$-balls, each chain is
strictly increasing. Its sets have distinct sizes between $0$ and $p$,
so each chain has at most $p+1$ elements. Therefore
\[
    q=|I_1|+\cdots+|I_a|
    \le (p+1)\left(1+|X|\bigl(1+\Delta(H)(\rank(H)-1)\bigr)^{p-1}\right)^p
\]
as required.
\end{proof}

\subsection{Existence of Neutral Swaps}
\label{sec:xhomogenousneutral}

We move on to showing how a fixed interaction with the separator allows us to find a
swap that does not increase $f$-width. Within an $X$-homogeneous context factor,
each colour interval has a common foreign ball, whose $f$-value we will call its
foreign width. We seek four consecutive intervals for which swapping the
middle two preserves their $X$-balls and native balls and replaces their
foreign balls by ones of no greater foreign width. Weak monotone separability
then ensures that the swap does not increase the width of any bag. The key
is that bounded bag size and invariant locality allow only boundedly many
foreign-width values. A combinatorial argument shows that a sufficiently long
sequence of these values must contain four consecutive intervals satisfying
the required inequalities. Finally, we combine this result with the bound on
$X$-intervals from the previous subsection: sufficiently many colour intervals
force a long homogeneous region, and hence a neutral swap. This supplies the
missing ingredient for the proof of \Cref{lem:localdealt}.

We next find a neutral swap inside a sufficiently long homogeneous context
factor. Fix an $X$-homogeneous context factor $U_0$ with spine $P_0$. First we
show that each colour interval has a well-defined foreign bag and foreign ball.

\begin{lemma} \label{lem:neutralaux0}
Let $U' \subseteq Red \cup Blue$ be a context factor of $F$
and let $P'$ be a colour interval of its spine. 
Then all the vertices of $P'$ have the same foreign bags.
\end{lemma}

\begin{proof}
Let $u_1,u_2 \in V(P')$. 
Let $w \in fbag_F(u_1)$. 
Since $u_1$ and $u_2$ are selected arbitrarily,
it is enough to demonstrate that $w \in fbag_F(u_2)$. 
As $w$ is not of the same colour as $u_1$,
$w$ is a proper ancestor of $u_1$ and, in fact, $w$
is a proper ancestor of all the vertices of $P'$.
In particular, $w$ is a proper ancestor of $u_2$.
Next, $w$ is not adjacent to $u_1$ since the vertices are of distinct colours.
Therefore, $w$  is adjacent to a proper successor of $u_1$ and, in fact,
to a proper successor of all  of $P'$ (in light of Corollary \ref{cor:monspine}). In  particular,
$w$ is adjacent to a proper successor of $u_2$. 
We conclude that $w \in \bag_F(u_2)$.
Since $u_2$ and $w$ are of different colours, $w \in fbag_F(u_2)$. 
\end{proof}

With the notation as in Lemma \ref{lem:neutralaux0},
we introduce the notation $fbag_F(P')$ 
as $fbag_F(u)$ for an arbitrary $u \in V(P')$.

\begin{lemma} \label{lem:neutralaux1}
Let $P'$ be a colour interval of $P_0$. 
Then all the vertices of $P'$ have the same foreign balls. 
\end{lemma}

\begin{proof}
Let $u_1,u_2 \in V(P')$. 
Let $W$ be a foreign component of $u_1$. 
Since vertices $u_1$ and $u_2$ are selected arbitrarily,
it is enough to demonstrate that $W$ is a foreign component of $u_2$. 
It is immediate from Lemma \ref{lem:neutralaux0} that
$W \subseteq \bag_F(u_2)$. 
As $W$ is a connected component of $H[\bag_F(u_1)]$, $W$ is a
connected set of $H[\bag_F(u_2)]$ but possibly not a maximal one. 
In this case, there is a vertex $w' \in \bag_F(u_2) \setminus W$ such that
$w'$ is adjacent to $W$. 
We observe that $w'$ cannot have the same colour as $u_2$ since
$H$ has no hyperedges including vertices of two different colours. 
If $w'$ has the same colour as $W$ then by Lemma \ref{lem:neutralaux0},
$w' \in \bag_F(u_1)$. But then $W \cup \{w'\}$ is a connected subset
of $H[\bag_F(u_1)]$ in contradiction to the maximality of $W$.

It remains to assume that $w' \in X$.
But then $W$ is a subset of the $X$-ball of $u_2$.
But this is a contradiction since $W$ is disjoint from the $X$-ball of $u_1$
and the latter is equal to the $X$-ball of $u_2$ since both $u_1$ and $u_2$
belong to the same $X$-interval.
\end{proof}

We denote the common foreign ball on $P'$ by $fball_F(P')$.

We refer to $f_H(fball_F(u))$ as the foreign width of $u$
(w.r.t. $F$) and denote it by $fw_F(u)$
\footnotemark \footnotetext{The naming $fwidth_F$ could be easily confused with
$\fwidth$. Therefore we chose another, easily distinguishable notation.}. 
It follows from Lemma \ref{lem:neutralaux1} that
the foreign width is invariant for the vertices of 
a colour interval. Therefore, we use 
the notation $fw_F(P')$ where $P'$ is a colour
interval to denote the foreign width w.r.t. $F$ of
an arbitrary vertex of $P'$.

The numerical criterion for a neutral swap is captured by a mutable
quadruple. The following combinatorial statement guarantees such a quadruple
in a sufficiently long sequence taking only finitely many values.

\begin{definition}
Let $S=(a_1, \dots, a_r)$ be a sequence of numbers. 
The quadruple  $(b_1,b_2,b_3,b_4)$ of four \emph{consecutive}
elements of $S$ is called \emph{mutable} if $b_3 \geq b_1$
and $b_2 \geq b_4$.  
\end{definition}

\begin{theorem} \label{theor:exmut}
Let $I$ be a non-empty finite set of reals
and let $S=(s_1,\dots, s_r)$ be a sequence such
that $s_i \in I$ for each $i \in [r]$. 
Suppose that $r \geq 4|I|$. 
Then $S$ has a mutable quadruple. 
\end{theorem}

\begin{proof}
By induction on $|I|$. 
If $|I|=1$ then we observe that $4$ consecutive $a,a,a,a$ form
a mutable quadruple. 

Assume now that $|I|>1$ and let $a$ be the largest element of $I$. 
Consider a sequence $S$ as specified in the statement of the theorem
and having length at least $4|I|$.

Suppose that $S$ has $5$ consecutive elements of the form $a_1,a_2,a,a_3,a_4$. 
If $a_2 \geq a_3$ then $a_1,a_2,a,a_3$ is a mutable quadruple. 
Otherwise, if $a_2<a_3$ then $a_2,a,a_3,a_4$ is a mutable quadruple. 

It remains to assume that five consecutive elements of the above form do not occur in $S$. 
It follows that $a$ can only occur among the first two elements or 
among the last two elements of $S$. Let $S'$ be the sequence obtained by removing
from $S$ the first two and the last two elements. 
Then each element of $S'$ is contained in $I \setminus \{a\}$ and the length of $S'$ is at least $4(|I|-1)=4|I \setminus \{a\}|$. 
Therefore $S'$, and hence $S$, contains a mutable quadruple by 
the induction assumption. 
\end{proof}

Now, let $P'_1, \dots, P'_q$ be the partition of $P_0$
into colour intervals that occur on $P_0$
in the order listed starting from the root of $U_0$.

For later use, observe that $fball_F(P'_j)\neq\emptyset$ for every
$3\le j\le q$. Indeed, let $t$ be the first vertex of $P_0$,
let $u$ be the first vertex of $P'_j$, and let $v$ be the last
vertex of $P'_{j-1}$. By reducedness, $v$ is adjacent to a
descendant of $u$, so $v\in fbag_F(u)$.
The common $X$-ball on $P_0$ is contained in $\bag_F(t)$,
whereas $v$, being a proper descendant of $t$, does not belong
to $\bag_F(t)$. Thus $v\in fball_F(u)$.
Lemma~\ref{lem:neutralaux1} gives the observation for the whole interval.

In the rest of the proof, we demonstrate that
if a sequence $fw_F(P'_3), \dots, fw_F(P'_q)$
has a mutable quadruple then the desired neutral swap
exists. We then harness Theorem \ref{theor:exmut}
to demonstrate that if $q$ is sufficiently large w.r.t.
$k$ and $\rank(H)$ then such a quadruple does
exist.

\begin{lemma} \label{lem:neutralaux2}
\sloppy
Assume that there is $3 \leq i \leq q-3$
such that $fw_F(P'_i)$,$fw_F(P'_{i+1})$,$fw_F(P'_{i+2})$,$fw_F(P'_{i+3})$
form a mutable quadruple of $fw_F(P'_1), \dots fw_F(P'_q)$.
Then the swap of $P'_{i+1}$ and $P'_{i+2}$ is neutral.
\end{lemma} 

\begin{proof}
Throughout the proof, we will frequently 
refer to the sets of ancestors and descendants of a vertex in an elimination forest. 
We denote these sets by $pred_F(u)$ and $succ_F(u)$, respectively; both include $u$.

Let $F'$ be the elimination forest resulting from the swap. 

Let $u_1$ and $v_1$ be the respective first and last vertices
of $P'_{i+1}$, $u_2$ and $v_2$ be the respective first and last 
vertices of $P'_{i+2}$ and $u_3$ be the first vertex of $P'_{i+3}$
(the paths are explored along the spine of $U_0$ starting from the root). 

\begin{claim} \label{clm:neutralaux21}
Let $u \in (V(H) \setminus V(F_{u_1})) \cup V(F_{u_3})$. 
Then $\bag_{F}(u)=\bag_{F'}(u)$. 
\end{claim}

\begin{claimproof}
It is straightforward to observe that for each $u$ as in the
statement of the claim, $pred_F(u)=pred_{F'}(u)$
and $succ_F(u)=succ_{F'}(u)$.
\end{claimproof}

\begin{claim} \label{clm:neutralaux22}
Let $u \in Cont_F(P'_{i+1}+P'_{i+2},\{u_3\}) \setminus V(P'_{i+1} \cup P'_{i+2})$.
Then $\bag_F(u)=\bag_{F'}(u)$. 
\end{claim}

\begin{claimproof}
We start by observing that $succ_F(u)=succ_{F'}(u)$. 
By Corollary \ref{cor:monspine}, all the vertices of $succ_F(u)$ are of the same colour
as $u$. Therefore, both $\bag_F(u)$ and $\bag_{F'}(u)$ can have only elements of $X$
or elements of the same colour as $u$. 
But $pred_{F}(u) \cap (X \cup Native(u))=pred_{F'}(u) \cap(X \cup Native(u))$.
\end{claimproof}

In light of Claims \ref{clm:neutralaux21} and \ref{clm:neutralaux22}, 
we only need to verify the lemma for $V(P'_{i+1}) \cup V(P'_{i+2})$. 
Our goal is to demonstrate that for each $u \in V(P'_{i+1}) \cup V(P'_{i+2})$,
$X\text{-ball}_{F}(u)=X\text{-ball}_{F'}(u)$, $nball_F(u)=nball_{F'}(u)$ 
and $fw_{F'}(u) \leq fw_F(u)$.

We observe (by construction and Corollary \ref{cor:monspine}) 
that for each $u \in V(P'_{i+1}) \cup V(P'_{i+2})$,
$pred_F(u) \cap Native(u)=pred_{F'}(u) \cap Native(u)$ and
$succ_F(u) \cap (Native(u) \cup X)=succ_{F'}(u) \cap (Native(u) \cup X)$.
We thus conclude that

\begin{claim} \label{clm:neutralaux23}
For each $u \in V(P'_{i+1}) \cup V(P'_{i+2})$, $nbag_F(u)=nbag_{F'}(u)$. 
\end{claim}

\begin{claim} \label{clm:neutralaux24}
$F'$ is a reduced elimination forest. 
$U_0$ is a context factor of $F'$. 
The sequence of colour intervals of $U_0$ in $F'$
is the concatenation of 
${\bf P}_0$, $F'[V(P'_{i}) \cup  V(P'_{i+2})],F'[V(P'_{i+1}) \cup V(P'_{i+3})]$, ${\bf P}_1$.
Here ${\bf P}_0$ is a possibly empty sequence of colour intervals of $P_0$ occurring before $P'_i$
and ${\bf P}_1$ is a possibly empty sequence of colour intervals of $P_0$ occurring after $P'_{i+3}$.
\end{claim}

\begin{claimproof}
The first two statements follow from Lemma \ref{lem:singleswap}.
The third statement is immediate by construction. 
\end{claimproof}

\begin{claim} \label{clm:neutralaux25}
For each $u \in V(P'_{i+1})$, $fbag_{F'}(u)=fbag_{F}(P'_{i+3})$ and
for each $u \in V(P'_{i+2})$, $fbag_{F'}(u)=fbag_F(P'_i)$. 
\end{claim}

\begin{claimproof}
We only prove the statement for $u \in V(P'_{i+1})$; the other case is symmetric.
Denote $F'[V(P'_{i+1}) \cup  V(P'_{i+3})]$ by $P'$. 
By Claim \ref{clm:neutralaux24} and Lemma \ref{lem:neutralaux0}, 
$fbag_{F'}(u)=fbag_{F'}(P')$. On the other hand $P'$ includes vertices of $P'_{i+3}$.
Hence, by Claim \ref {clm:neutralaux21}, $fbag_{F'}(P')=fbag_F(P'_{i+3})$.
\end{claimproof}

\begin{claim} \label{clm:neutralaux26}
For each $u \in V(P'_{i+1}) \cup V(P'_{i+2})$,
$X\text{-bag}_F(u)=X\text{-bag}_{F'}(u)$. 
\end{claim}

\begin{claimproof}
We first observe that $pred_F(u) \cap X=pred_{F'}(u) \cap X$. 
Let $u \in V(P'_{i+2})$. By construction, 
$succ_{F'}(u)=succ_F(u) \cup Cont_F(P'_{i+1},\{u_2\})$. 
Therefore, if the claim is not true then $X\text{-bag}_{F'}(u)$ is a strict
superset of $X\text{-bag}_F(u)$. 
Let $v \in X\text{-bag}_{F'}(u) \setminus X\text{-bag}_F(u)$. 
It follows that there is $w \in Cont_F(P'_{i+1},\{u_2\})$
such that $v$ is adjacent to $w$. 
Let $u_0$ be an arbitrary vertex of $P'_i$.
We observe that $v \in pred_F(u_0)$ and $w \in succ_F(u_0)$. 
Hence, $v \in X\text{-bag}_F(u_0)$. Due to $U_0$ being $X$-homogeneous, 
$X\text{-ball}_F(u_0)=X\text{-ball}_F(u)$ and hence $v \in X\text{-bag}_F(u)$ in contradiction
to our assumption. 

Now, let $u \in V(P'_{i+1})$. 
Then $succ_{F'}(u)=succ_F(u) \setminus Cont_F(P'_{i+2}, \{u_3\})$. 
Hence,  if the claim does not hold, $X\text{-bag}_{F'}(u)$ is a strict subset
of $X\text{-bag}_F(u)$. 
Let $v \in X\text{-bag}_F(u) \setminus X\text{-bag}_{F'}(u)$.
Let $u_0 \in V(P'_{i+3})$. 
Due to the $X$-homogeneity of $U_0$, $v \in X\text{-bag}_{F}(u_0)$. 
By Claim \ref {clm:neutralaux21}, $v \in \bag_{F'}(u_0)$. 
Since $succ_{F'}(u_0) \subseteq succ_{F'}(u)$, we conclude 
that $v \in X\text{-bag}_{F'}(u)$, a contradiction. 
\end{claimproof}

\begin{claim} \label{clm:neutralaux27}
For each $u \in V(P'_{i+1}) \cup V(P'_{i+2})$,
$X\text{-ball}_F(u)=X\text{-ball}_{F'}(u)$. 
\end{claim}

\begin{claimproof}
In light of Claim \ref{clm:neutralaux26},
it is enough to observe that for each $x \in X\text{-bag}_F(u)$,
the connected component of $H[\bag_F(u)]$ containing $x$ 
is also a connected component of $H[\bag_{F'}(u)]$. 

So, let $W$ be such a component. 
We need first to verify that $W \subseteq \bag_{F'}(u)$.
$W \cap (nbag_F(u) \cup X\text{-bag}_F(u)) \subseteq \bag_{F'}(u)$
by combination of Claims \ref{clm:neutralaux23} and \ref{clm:neutralaux26}. 
Let $v \in W \cap fbag_{F}(u)$.

Assume first that $u \in V(P'_{i+2})$. 
Let $u_0 \in V(P'_i)$. 
Due to the $X$-homogeneity of $U_0$ and Claim \ref{clm:neutralaux21},
$v \in \bag_{F'}(u_0)$. As $u$ and $u_0$ are of the same colour,
$v \in fbag_{F'}(u_0)$. By Claim \ref{clm:neutralaux25},
$v \in fbag_{F'}(u) \subseteq \bag_{F'}(u)$. 
If $u \in V(P'_{i+1})$, we apply the same argument 
with $P'_{i+3}$ being used instead of $P'_i$. 

By assumption about $W$, it is a connected set of $H$
and hence a connected set of $H[\bag_{F'}(u)]$ and hence, in turn,
a subset of a connected component of $H[\bag_{F'}(u)]$.
Suppose $W$ is not a connected component of $H[\bag_{F'}(u)]$.
Since $W$ is connected, there exists $v\in\bag_{F'}(u)\setminus W$ adjacent
to $W$. As $W$ is a connected component of $H[\bag_F(u)]$, we have
$v\notin\bag_F(u)$, and hence $v\in\bag_{F'}(u)\setminus\bag_F(u)$.
By combination of Claims \ref{clm:neutralaux23} and \ref{clm:neutralaux26},
$v \in fbag_{F'}(u)$. Assume that $u \in V(P'_{i+2})$.
Let $u_0 \in V(P'_i)$. By Claim \ref{clm:neutralaux25}, $v \in \bag_{F'}(u_0)$. 
On the other hand, by the $X$-homogeneity of $U_0$ and Claim \ref{clm:neutralaux21},
$W$ is a connected component of $H[\bag_{F'}(u_0)]$, a contradiction. 
If $u \in V(P'_{i+1})$ then the argument is symmetric with $P'_{i+3}$ being used 
instead of $P'_i$.
\end{claimproof}

\begin{claim} \label{clm:neutralaux28}
For each $u \in V(P'_{i+1})$, $fball_{F'}(u)=fball_{F}(P'_{i+3})$ and
for each $u \in V(P'_{i+2})$, $fball_{F'}(u)=fball_F(P'_i)$. 
\end{claim}

\begin{claimproof}
We prove the claim for $u \in V(P'_{i+1})$;
for the other case, the proof is symmetric. 
We first observe that $fball_{F'}(u)=fbag_{F'}(u) \setminus X\text{-ball}_{F'}(u)$. 
Let $u_0 \in V(P'_{i+3})$.
By Claim \ref{clm:neutralaux25}, $fbag_{F'}(u)=fbag_{F'}(u_0)$. 
By combination of Claim \ref{clm:neutralaux27}, the $X$-homogeneity of $U_0$
and Claim \ref{clm:neutralaux21}, $X\text{-ball}_{F'}(u)=X\text{-ball}_{F'}(u_0)$. 
We thus conclude that $fball_{F'}(u)=fball_{F'}(u_0)$. 
By Claim \ref{clm:neutralaux21}, $fball_{F'}(u)=fball_F(u_0)$. 
Finally, by Lemma \ref{lem:neutralaux1}, $fball_{F'}(u)=fball_F(P'_{i+3})$.
\end{claimproof}

\begin{claim} \label{clm:neutralaux29}
For each $u \in V(P'_{i+1}) \cup V(P'_{i+2})$,
$nball_F(u)=nball_{F'}(u)$. 
\end{claim}

\begin{claimproof}
We observe that
$nball_{F'}(u)=nbag_{F'}(u) \setminus X\text{-ball}_{F'}(u)=
nbag_F(u) \setminus X\text{-ball}_F(u)=nball_F(u)$,
where the middle equality follows from Claims
\ref{clm:neutralaux23} and \ref{clm:neutralaux27}.
\end{claimproof}

Let $u\in V(P'_{i+1})$ and put
$A=X\text{-ball}_F(u)\cup nball_F(u)$.
By Claims~\ref{clm:neutralaux27} and~\ref{clm:neutralaux29},
$A=X\text{-ball}_{F'}(u)\cup nball_{F'}(u)$.
The foreign balls of $P'_{i+1}$ and $P'_{i+3}$ are nonempty
by the observation above. Hence, by Proposition~\ref{prop:strongadd}
and Claim~\ref{clm:neutralaux28},
\[
\begin{aligned}
f_H(\bag_{F'}(u))
&=g_f\bigl(f_H(A),fw_F(P'_{i+3})\bigr)\\
&\le g_f\bigl(f_H(A),fw_F(P'_{i+1})\bigr)
 =f_H(\bag_F(u)).
\end{aligned}
\]
The inequality follows from the mutable-quadruple assumption
and monotonicity of $g_f$.

For $u \in V(P'_{i+2})$ the argument is the same with $P'_i$
used instead of $P'_{i+3}$ and $P'_{i+2}$ used instead of $P'_{i+1}$.
Together with Claims~\ref{clm:neutralaux21} and~\ref{clm:neutralaux24},
this shows that $F'$ is a reduced elimination forest of no greater $f$-width.
\end{proof}

Bounded size and invariant locality bound the number of possible foreign
widths. The mutable-quadruple theorem and the neutrality criterion therefore
give the following result.

\begin{lemma} \label{lem:homogenousneutral}
Let $p(k,r)=\max\{1,\lceil\beta(k,r)\rceil\}$ and set
$\kappa_1(k,r)=4(p(k,r)+1)2^{2^{p(k,r)}}$.
Let $U_0$ be an $X$-homogeneous context factor with spine $P_0$ that
has at least $\kappa_1(k,\rank(H))$ colour intervals. %
Then $U_0$ admits a neutral swap of two non-terminal (neither first nor last) intervals. 
\end{lemma}

\begin{proof}[Proof of \Cref{lem:homogenousneutral}]
Let $r=\rank(H)$ and $p=p(k,r)$.
For $j\ge 3$, the foreign ball of $P'_j$ is nonempty and has
at most $p$ vertices. By invariant locality, its foreign width
is determined by the isomorphism type of the induced hypergraph
on that ball. The number of such types is at most
\[
    \sum_{n=1}^{p}2^{2^n-1}\le p2^{2^p-1},
\]
since a hypergraph on a fixed set of $n$ vertices is specified
by a collection of non-empty subsets.
Thus the sequence $fw_F(P'_3),\dots,fw_F(P'_q)$ takes at most
$p2^{2^p-1}$ distinct values.
If $q\ge\kappa_1(k,r)=4(p+1)2^{2^p}$, then
\[
    q-2\ge 4p2^{2^p-1}.
\]
By Theorem~\ref{theor:exmut}, this sequence has a mutable quadruple,
whose starting index $i$ satisfies $3\le i\le q-3$.
By Lemma \ref{lem:neutralaux2}, the swap of $P'_{i+1}$ and $P'_{i+2}$
is neutral. We also note that none of these two is a terminal colour interval. 
\end{proof}

Both ingredients are now established: the spine has boundedly many
$X$-intervals, and a homogeneous region with many colour intervals admits a
neutral swap. We combine them to prove the lemma used in the local construction.

\begin{proof}[Proof of Lemma \ref{lem:existsneutral}] 
\sloppy
Let $r$ be the number of colour intervals of $P$. 
Assume that $r \geq \kappa_1(k,\rank(H)) \cdot \kappa_0(|X|,k,\rank(H), \Delta(H))$
where $\kappa_0$ is as in Lemma \ref{lem:boundedxintervals}
and $\kappa_1$ is as in Lemma \ref{lem:homogenousneutral}. 
We prove that in this case $U$ admits a neutral swap. 
Let $P'_1, \dots, P'_q$ be the partition of $P$ into 
$X$-intervals. According to Lemma \ref{lem:boundedxintervals},
$q \leq \kappa_0(|X|,k,\rank(H), \Delta(H))$.
It is not hard to see that $U$ is the disjoint union of context factors
$U_1,\dots, U_q$ with $P'_i$ and $W'_i$ as the respective spine and appendices
for each $U_i$ where $W'_i$ is a singleton consisting of the first vertex of $P'_{i+1}$
if $i<q$ and $W'_q=W$.
For $1 \leq i \leq q$, let $r_i$ be the number of colour intervals
of $U_i$. Clearly, $r \leq r_1+ \dots+r_q$ as each colour interval
of each $U_i$ is a subset of exactly one colour interval of $U$. 
It thus follows from the pigeonhole principle that some $U_i$ has 
at least $\kappa_1(k,\rank(H))$ colour intervals.
By Lemma \ref{lem:homogenousneutral} and the definition of a neutral
swap, $U_i$ admits a swap of two non-terminal colour intervals
that does not increase the \fwidth of the resulting elimination forest $F'$. 
We note that each non-terminal colour interval of $U_i$ is also 
a colour interval of $U$ and two such consecutive intervals are also consecutive
for $U$. We conclude that their swap is also a swap for $U$ that does not increase
the \fwidth of $F$ and hence a neutral swap of $U$ as required. 
\end{proof}

\section{A Direct Dynamic Programming Algorithm}
\label{sec:bk}

In this section we give a direct dynamic programming algorithm for a fixed
manageable width function $f$.  Given a hypergraph $H$ and an integer $k$, it either reports that
$\fwidth(H)>k$ or computes the optimum and an optimal elimination forest.
As in \Cref{sec:techoverview}, we start from a tree decomposition
${\bf t}=(T,B)$ of bounded ordinary width.  We then work from its leaves
to its root, combining partial forests whose shape may differ from that
of $T$.  This proceeds (loosely) in the spirit of the Bodlaender--Kloks algorithm for
treewidth~\cite{DBLP:journals/jal/BodlaenderK96}, where the states summarise candidate solutions relative to the given decomposition.

At a node $x$, we call the vertices of $cmp(x)$, whose margin nodes lie in
$T_x$, \emph{processed}.  A partial certificate assigns a colour to each
processed vertex and records either that the vertex is a root or a colour
representing its parent.  Recording the parent by its colour allows us to
describe connections to vertices that have not yet been processed.  Ultimately, the bound on the chromatic number of the conflict graph from~\Cref{th:cgchrom}
will also guarantee that boundedly many colours suffice in this approach. Intuitively, we use the bound in the sense that whenever
$\fwidth(H)\le k$, then some optimal forest has properly coloured stains using
at most $q$ colours.
For non-root $x$, every parent--child edge crossing between $cmp(x)$ and
its complement has its parent's stain containing the decomposition edge
above $x$.  All these stains intersect, so distinct parents involved in
such connections have distinct colours.

We require various additional bookkeeping data. We need to record which colour names remain available for connections beyond $T_x$; we refer to these as its \emph{active colours}.  Together
with the vertices of $B_x$, these colour names form a bounded
\emph{interface}: the bag vertices identify where hyperedges from different
subtrees can meet, while the colours describe forest connections across
the boundary.  We record which represented vertices are ancestors of
which others in the partial forest.  Several interface names can represent
the same vertex, so this relation is a preorder.  Further connections may identify an unresolved parent with a boundary vertex or introduce new
ancestor relations.  We therefore also retain candidate extensions of this preorder.  

Finally, to determine the contribution
from $T_x$ to the bag of a vertex $u$, we must account for the additional ancestry relations and for which interface vertices are ancestors or
descendants of $u$.  This is necessary even when $u\notin cmp(x)$, since its bag may contain vertices of $cmp(x)$.  To this end the state specifically records the induced subhypergraphs of these contributions.  The bounded interface and the bound on bag size yield only boundedly many states in terms of $k,\rank(H),\Delta(H)$.

\subsection{Setup and DP states}

Fix a hypergraph $H$ and an integer $k$, and let $r:=\rank(H)$, $d:=\Delta(H)$, and
$\beta:=\beta_f(k,r)$.  Then
\begin{equation}
 f_H(U)\le k\quad\Longrightarrow\quad |U|\le\beta,
 \label{eq:bk-size-bound}
\end{equation}
so $\fwidth(H)\le k$ implies $\tw(H)\le\beta-1$.  Let
${\bf t}=(T,B)$ be a width-$(\beta-1)$ tree decomposition, rooted and
binarised with empty root $\rho$.  Use $mrg(x)$, $cmp(x)$, and $\mbag(v)$
as in Section~\ref{sec:prelims}.  Lemma~\ref{lem:dealt} and
Theorem~\ref{th:cgchrom} give a computable $q=q(k,r,d)$ such that, whenever
$\fwidth(H)\le k$, some optimal elimination forest has properly
$q$-coloured stains with respect to ${\bf t}$.

For each $e\in E(H)$, let $a(e)$ be the highest node such that
$e\subseteq B_{a(e)}$, and
set
\[
 \mathcal E_x^0:=\{e\mid a(e)=x\},\qquad
 \mathcal E_x:=\{e\mid a(e)\in V(T_x)\}.
\]
Thus
\begin{equation}
 \mathcal E_x=\mathcal E_x^0\mathbin{\dot\cup}
 \mathop{\dot\bigcup}_{y\text{ child of }x}\mathcal E_y.
 \label{eq:bk-edge-split}
\end{equation}
Let $H_x$ be the hypergraph with vertex set $cmp(x)\cup B_x$ and edge set
$\mathcal E_x$; vertices of $B_x$ may be isolated in $H_x$.

For every non-root node $z\in V(T)$, let $e_z$ denote the edge of $T$ joining
$z$ to its parent.  We encode candidate elimination forests relative to
${\bf t}$.  The maps $c$ and $p$ record the colours of a vertex's own stain
and its parent's stain, respectively.  For each colour $i$, a support forest
$Q_i$ connects the corresponding stains; allowing it to contain additional
nodes and edges makes restriction and composition immediate.  The bit
$\varepsilon_i$ says whether its component at $x$ continues through the parent
edge of $x$.
Throughout, regard $[q]$ as disjoint from $V(H)$.

\begin{definition}[Partial certificates]
Fix $x\in V(T)$.  An \emph{$x$-partial certificate}, also called a
\emph{partial certificate below $x$}, is a tuple
\[
 P=\bigl(c,p,(Q_i,\varepsilon_i)_{i\in[q]}\bigr),
\]
where $c:cmp(x)\to[q]$, $p:cmp(x)\to\{0\}\cup[q]$,
each $Q_i$ is a subgraph of $T_x$, called the \emph{colour-$i$ support
forest}, and each
$\varepsilon_i\in\{0,1\}$.  If $x=\rho$, then $\varepsilon_i=0$ for every
$i\in[q]$.  In all cases, we require
\begin{equation}
 \begin{aligned}
  \varepsilon_i=1&\Rightarrow x\in V(Q_i)
    &&\text{for every $i\in[q]$},\\
  \mbag(v)&\in V(Q_{c(v)})
    &&\text{for every $v\in cmp(x)$},\\
  p(v)\ne0&\Rightarrow \mbag(v)\in V(Q_{p(v)})
    &&\text{for every $v\in cmp(x)$}.
 \end{aligned}
 \label{eq:bk-certificate-incidence}
\end{equation}

Let
\(
 I_x(P):=\{i\in[q]\mid\varepsilon_i=1\}
\) be the \emph{active} colours at $x$.
For $i\in I_x(P)$, let $O_i$ be the component of $Q_i$ containing $x$.
We call the components $O_i$ \emph{open} and all other components of the
support forests \emph{closed}.  A vertex
$u\in cmp(x)$ \emph{anchors} a component $C$ of $Q_i$ if
$c(u)=i$ and $\mbag(u)\in V(C)$.
Call a component \emph{anchored} if it has an anchor and \emph{pending}
otherwise.  We require every support component to have at most one anchor and
every pending component to be open.  For an anchored component $C$, let $u_C$
denote its unique anchor.  Thus every closed component is anchored, while an
open component may be anchored or pending.  Define
\[
 I_x^{\rm anch}(P) :=\{i\in I_x(P)\mid O_i\text{ is anchored}\}.
\]

Represent every component $C$ of $Q_i$ by
\[
 \bar C:=
 \begin{cases}
  u_C,&C\text{ is anchored},\\
  i,&C=O_i\text{ is pending}.
 \end{cases}
\]
For $v\in cmp(x)$ with $p(v)\ne0$, let $C_v$ be the component of
$Q_{p(v)}$ containing $\mbag(v)$.  We require the arcs $(\bar C_v,v)$, over all
such $v$, to form a rooted forest $F_x(P)$ on
$(cmp(x)\cup B_x)\mathbin{\dot\cup}
\{i\in I_x(P)\mid O_i\text{ is pending}\}$.  Put $\bar i:=\bar O_i$ for every
active colour $i$ and $\bar v:=v$ for every actual vertex.  Define
$a\preceq_x(P)b$ if $\bar a$ is an ancestor of $\bar b$ in $F_x(P)$.  Thus an
anchored component
names its anchor, while a pending component
describes an unresolved root.

Let $\mathsf{Cert}(x)$ denote the family of all $x$-partial certificates. We refer to a
$\rho$-partial certificate as a \emph{complete certificate}.
\end{definition}

Fix $x\in V(T)$, let $P\in\mathsf{Cert}(x)$, and let $y$ be a child of $x$.
For $i\in[q]$, put
\[
 Q_i^y:=\bigl(V(Q_i)\cap V(T_y),E(Q_i)\cap E(T_y)\bigr),\qquad
 \varepsilon_i^y:=\mathbf 1[e_y\in E(Q_i)],
\]
and define the \emph{restriction of $P$ to $y$} by
\[
 P|_{T_y}:=\bigl(c|_{cmp(y)},p|_{cmp(y)},
                    (Q_i^y,\varepsilon_i^y)_{i\in[q]}\bigr).
\]

Fix $x\in V(T)$ and $P\in\mathsf{Cert}(x)$.  The \emph{ancestry retained at
$x$} is
\[
 \mathcal A_x(P):=
 \preceq_x(P)|_{B_x\mathbin{\dot\cup}I_x(P)}.
\]
Thus an active colour consistently names its open component: if that
component is anchored, its name is equivalent to the anchor; if it is
pending, its name remains an unresolved ancestor.  The retained ancestry has
ground-set size $|B_x|+|I_x(P)|\le\beta+q$.

Vertices in different components of the partial forest may become comparable
when connections through the unprocessed region are completed.  We will therefore also need to 
record candidate extensions of the ancestry on the interface
\[
 \Gamma_x(P):=B_x\mathbin{\dot\cup}I_x(P).
\]
An \emph{ancestry extension} of a preorder $R$ is a preorder on the same
ground set containing $R$.  Write $\operatorname{rtc}$ for reflexive
transitive closure.  We allow all preorder extensions,
without requiring that each arise from a forest completion.
For an ancestry extension
$\lambda\supseteq\mathcal A_x(P)$, put
\[
 \preceq_{x,\lambda}:=
 \operatorname{rtc}\bigl(\preceq_x(P)\cup\lambda\bigr).
\]
The set of
\emph{chain-valid ancestry extensions} is
\[
 \mathcal C_x(P):=\{\lambda\supseteq\mathcal A_x(P)\mid
 \text{every two vertices of every }e\in\mathcal E_x
 \text{ are comparable under }\preceq_{x,\lambda}\}.
\]

\begin{example}
\label{ex:bk-partial}
\label{ex:bk-completions}
When processing a subtree of the given decomposition, we may encounter a
vertex before its parent in the elimination forest.  The active colours
allow us to record this connection without yet specifying the parent.
For a running example, let $H$ have edges $ab,bs,su,uc$.  Take $\beta\ge3$ and $q\ge5$, and use the
tree decomposition shown on the left below.  The right-hand diagram shows
the partial forest that we will describe at $y$.
\begin{center}
\begin{minipage}[t]{0.53\linewidth}
\centering
\textit{Tree decomposition ${\bf t}$}\par\smallskip
\begin{tikzpicture}[font=\small,draw=black,text=black,
  bag node/.style={draw,rounded corners=2pt,inner sep=3pt}]
 \node[bag node] (rho) at (0,0) {$\rho:\emptyset$};
 \node[bag node] (x) at (0,-0.7) {$x:\{s,u\}$};
 \node[bag node] (y) at (-1.15,-1.4) {$y:\{s,u\}$};
 \node[bag node] (z) at (1.15,-1.4) {$z:\{s,u\}$};
 \node[bag node] (ly) at (-1.15,-2.1) {$\ell_y:\{a,b,s\}$};
 \node[bag node] (lz) at (1.15,-2.1) {$\ell_z:\{u,c\}$};
 \draw (rho)--(x) (x)--(y) (x)--(z) (y)--(ly) (z)--(lz);
 \draw[dashed,rounded corners=3pt] (-2.08,-1.12) rectangle (-0.22,-2.43);
 \node at (-1.15,-2.67) {$T_y$};
\end{tikzpicture}
\end{minipage}\hfill
\begin{minipage}[t]{0.43\linewidth}
\centering
\textit{Partial forest $F_y(P)$}\par\smallskip
\begin{tikzpicture}[font=\small,draw=black,text=black,
  vertex/.style={draw,circle,minimum size=6mm,inner sep=1pt}]
 \useasboundingbox (-0.4,0.2) rectangle (3.55,-2.85);
 \node[vertex] (a) at (0,-0.5) {$a$};
 \node[vertex,dashed] (pending) at (1.05,-0.5) {$3$};
 \node[vertex] (b) at (1.05,-1.4) {$b$};
 \node[vertex] (s) at (2.1,-0.5) {$s$};
 \node[vertex] (u) at (3.15,-0.5) {$u$};
 \draw[->] (pending)--(b);
 \node[align=center] at (1.55,-2.2) {Dashed node: unresolved parent\\of $b$, represented by colour $3$.};
\end{tikzpicture}
\end{minipage}
\end{center}
The bags give
\[
 \begin{gathered}
 cmp(y)=\{a,b\},\quad cmp(z)=\{c\},\quad mrg(x)=\{s,u\},\\
 \mathcal E_y=\{ab,bs\},\quad
 \mathcal E_z=\{uc\},\quad \mathcal E_x^0=\{su\}.
 \end{gathered}
\]
When the computation reaches $y$, consider a partial certificate in which
$a,b$ have colours $1,2$, respectively, $a$ is a root, and the parent of
$b$ is represented by colour $3$.  Thus $c(a)=1$, $c(b)=2$, $p(a)=0$,
and $p(b)=3$. We keep all three colour names available beyond $T_y$.  For each of these
colours, choose the path from $\ell_y$ to $y$ as its support and set its
continuation bit to $1$.  All other supports are empty, with continuation
bits $0$.  The first two supports are anchored at $a,b$.  The third has
no anchor in $cmp(y)$, so it represents the unresolved parent of $b$,
giving the arc $3\to b$.

{
The certificate specifies the colour of $b$'s parent, but has not yet
identified that parent with a vertex.  In particular, the presence of $s$
in $B_y$ does not identify it with the placeholder $3$.  Consider now two
completions in which $s,u,c$ receive colours $3,4,5$, respectively, and
$s$ is the parent of $b$:
\[
 F_A:\ a\to u\to s\to b,\ u\to c,
 \qquad F_B:\ a\to s\to b\to u\to c.
\]
Both completions are compatible with the chosen supports: these may
extend beyond the stains, so colour $2$ can remain
open even when $b$ has no children in $F_A$.

Both completions resolve the parent of $b$ in the same way, but place $u$
differently: above $b$ in $F_A$ and below it in $F_B$.  On the interface
$\Gamma_y=\{s,u,1,2,3\}$, this gives two different ancestry extensions:
\[
 \lambda_A:\ 1\prec u\prec3\equiv s\prec2,
 \qquad \lambda_B:\ 1\prec3\equiv s\prec2\prec u,
\]
\par}
Here $1$ represents $a$, $2$ represents $b$, and $3\equiv s$ records the
resolved parent.  Both extensions make the edges $ab,bs$ of $\mathcal E_y$
chains, so both belong to $\mathcal C_y$.  This remaining freedom matters
when recovering bags: in $F_A$, the edge $ab$ puts $a$ in the bag of $u$,
whereas in $F_B$ both endpoints lie above $u$.  The partial forest alone
therefore does not yet determine its contribution to this bag.  Retaining
$\mathcal C_y$ allows us to record the contributions for the different
ancestry extensions, as we do next.
\end{example}

{
The remaining task is to recover the bags of the eventual elimination
forest.  The bag of a vertex $u$ contains $u$ and the ancestors that have
a neighbour at or below $u$.  Consequently, vertices of $cmp(x)$ can
belong to this bag even if $u$ itself lies outside $cmp(x)$.  It is
therefore not enough to describe contributions to bags of processed
vertices.  We also need contributions to bags of vertices that have not
yet been processed.  These two families of contributions will form the
\emph{internal} and \emph{external bag tables}, respectively. In either
case, we call the vertex whose bag is considered the \emph{marked vertex}.

An ancestry extension describes relations among interface elements.  To
recover the bag of a marked vertex $u$, we also need its relation to those
elements.  If $u\in cmp(x)$, some of these relations are already fixed by
the partial forest.  Its \emph{boundary profile} records which interface
elements represent ancestors of $u$ and which represent descendants:}
\[
 \pi_x(P,u):=
 \bigl(\{a\in\Gamma_x(P)\mid a\preceq_x(P)u\},
       \{a\in\Gamma_x(P)\mid u\preceq_x(P)a\}\bigr).
\]
For $\lambda\in\mathcal C_x(P)$, a \emph{position over $\lambda$} is a pair
$(D,U)$ of subsets of $\Gamma_x(P)$ such that $D$ is downward closed, $U$ is
upward closed, and $d\mathrel\lambda u$ for all $d\in D$ and $u\in U$.
It extends a profile $(D_0,U_0)$ if $D_0\subseteq D$ and $U_0\subseteq U$.

Once the vertices at or above and at or below the
marker have been specified, the edges assigned to $T_x$ determine their
contribution to its bag.  For
$X^-,X^+\subseteq cmp(x)\cup B_x$, put
\begin{equation}
 L_x(X^-,X^+)
 =X^-\cap X^+
  \mathbin\cup
  \{a\in X^-\mid
       \exists e\in\mathcal E_x:\ a\in e,\ e\cap X^+\ne\emptyset\}.
\label{eq:bk-L}
\end{equation}
Writing $L=L_x(X^-,X^+)$, define
\[
 \operatorname{sig}_x(X^-,X^+):=
 \begin{cases}
  \infty,&|L|>\beta,\\
  \bigl(H_x[L\cup B_x],L\cap B_x\bigr),&|L|\le\beta.
 \end{cases}
\]
We identify two finite signatures $(J,S)$ and $(J',S')$ if $S=S'$ and there is
a hypergraph isomorphism $J\cong J'$ that fixes $B_x$ pointwise.

{
The signature retains the induced hypergraph needed to evaluate $f$, as
well as the boundary vertices needed to combine contributions from
different subtrees.  The distinguished subset $L\cap B_x$ records which
of these boundary vertices actually belong to the contribution.  Thus
retaining the boundary does not require treating all its vertices as bag
members.
\par}

For a position $(D,U)$ over $\lambda$, let $A^-(D)$ and $A^+(U)$ be the
downward and upward closures of $D$ and $U$ under $\preceq_{x,\lambda}$,
restricted to $cmp(x)\cup B_x$.  For $u\in cmp(x)$, define
$A^-_u(D)$ and $A^+_u(U)$ in the same way from $D\cup\{u\}$ and
$U\cup\{u\}$.

For an $x$-partial certificate $P$, define its \emph{internal bag table}
$\mathcal B_x^{\rm int}(P)$ as the set of all triples
\[
 \bigl(\pi_x(P,u),(\lambda,D,U),
       \operatorname{sig}_x(A^-_u(D),A^+_u(U))\bigr)
\]
for every $u\in cmp(x)$, every $\lambda\in\mathcal C_x(P)$, and every
position $(D,U)$ over $\lambda$ extending $\pi_x(P,u)$.  Define its
\emph{external bag table} $\mathcal B_x^{\rm ext}(P)$ as the set of all pairs
\[
 \bigl((\lambda,D,U),\operatorname{sig}_x(A^-(D),A^+(U))\bigr)
\]
for every $\lambda\in\mathcal C_x(P)$ and every position $(D,U)$ over
$\lambda$. Repetitions are omitted from both tables.  Put
$\mathcal B_x(P):=(\mathcal B_x^{\rm int}(P),
\mathcal B_x^{\rm ext}(P))$.  The \emph{state of $P$ at $x$} is
\[
 \tau_x(P):=\bigl(I_x(P),I_x^{\rm anch}(P),
 \mathcal A_x(P),\mathcal C_x(P),\mathcal B_x(P)\bigr).
\]
For $x\in V(T)$, let
\(
\mathfrak T_x:=\{\tau_x(P)\mid P\in\mathsf{Cert}(x)\}.
\)

{
\begin{example}
\label{ex:bk-rows}
Continue with the partial certificate from Example~\ref{ex:bk-partial},
taking $u$ as the marked vertex.  Since $u\notin cmp(y)$, its contributions
belong to the external bag table.  We compute the entries corresponding
to $\lambda_A$ and $\lambda_B$.
For $F_A$, the vertices of $H_y$ at or above and at or below $u$ are,
respectively,
\[
 X^-=\{a,u\},\qquad X^+=\{u,s,b\}.
\]
In~\eqref{eq:bk-L}, the intersection $X^-\cap X^+$ contributes $u$, and
the edge $ab$ contributes $a$.  Thus $L_y(X^-,X^+)=\{a,u\}$.
For $F_B$, we instead have $X^-=\{a,s,b,u\}$ and $X^+=\{u\}$.
Both endpoints of $ab$ are now above $u$, and no edge of $\mathcal E_y$
contains $u$, so only the intersection term contributes: $L_y=\{u\}$.
The corresponding rows of the same external table are
\[
\begin{array}{c|c|c|c|c}
 \lambda&D&U&L_y&S_y\\ \hline
 \lambda_A&\{1,u\}&\{u,3,s,2\}&\{a,u\}&\{u\}\\
 \lambda_B&\{1,3,s,2,u\}&\{u\}&\{u\}&\{u\}
\end{array}
\]
Note that the sets $D,U$ both contain $u$ because ancestry is
reflexive and that the equivalent elements $3,s$ occur together.
The signatures have underlying hypergraphs
$J_A=H_y[\{a,s,u\}]$ and $J_B=H_y[\{s,u\}]$,
with $S_y=\{u\}$ in both cases.  Thus $s$ is retained as a boundary vertex
for subsequent gluing, but is not yet part of either bag contribution.
This distinction matters when we incorporate the local edge $su$ at $x$.
\end{example}
}

Let $Y_x$ be the set of children of $x$.  The local data at $x$ are the values
of $c,p$ on $mrg(x)$, the choice of whether $x\in V(Q_i)$, and the bit
$\varepsilon_i$, for each $i\in[q]$.  Discard a choice unless $x$ is chosen for
$Q_i$ whenever $\varepsilon_i=1$ or a child $i$-support is open.  The remaining
data determine the datum below $x$ by taking the unions of the child supports,
adding $x$ to $Q_i$ according to the local choice, and adding $e_y$ to $Q_i$
precisely when the $i$-support is open at $y$.  Every partial certificate is
recovered in this way from its restrictions.

\begin{proposition}
\label{prop:bk-certificate-restriction}
The restriction of an $x$-partial certificate to a child $y$ belongs to
$\mathsf{Cert}(y)$.
\end{proposition}
\begin{proof}
The incidence conditions restrict to $T_y$.  Every closed component of
$Q_i^y$ is already closed in $Q_i$ and hence anchored; its open component lies
in a component of $Q_i$ and has at most one anchor.  Restriction deletes
vertices and arcs from $F_x(P)$.  If the anchor of an open child component lies
outside $cmp(y)$, it is replaced by its colour, which is a new root.  Hence no
directed cycle is created.
\end{proof}

\begin{lemma}
\label{lem:bk-finite-types}
For the fixed manageable function $f$, there is a computable function $N_f$
such that
\[
 |\mathfrak T_x|\le N_f(k,r,d)
\]
for every node $x$ of every rooted tree decomposition of width at most
$\beta_f(k,r)-1$.
\end{lemma}
\begin{proof}
Put $s:=\beta+q$, $R_s:=2^{s^2}$, and $P_s:=2^{2s}$.  Thus, on any fixed
set of size at most $s$, there are at most $R_s$ binary relations and at most
$P_s$ ordered pairs of subsets.  For $0\le b\le\beta$, let $\mathscr S_b$
be the set of isomorphism types of pairs $(J,S)$ such that $J$ is a
hypergraph of rank at most $r$ on at most $2\beta$ vertices with a
pointwise-labelled boundary of size $b$, and $S$ is a subset of that
boundary.  The sets $\mathscr S_b$ are finite and can be enumerated.  Hence
\[
 M:=1+\max_{0\le b\le\beta}|\mathscr S_b|
\]
is computable; the additional value counted by $1$ is $\infty$.  Every
signature occurring at $x$ is one of these $M$ possibilities, after the
labelled boundary is identified with $B_x$, because $|B_x|\le\beta$ and a
finite signature has at most $|L|+|B_x|\le2\beta$ vertices.

Fix $P\in\mathsf{Cert}(x)$ and write $\Gamma:=\Gamma_x(P)$, so
$|\Gamma|\le s$.  If $\mathscr R_\Gamma$ is the set of binary relations on
$\Gamma$ and $\mathscr P_\Gamma:=\mathcal P(\Gamma)^2$, then, disregarding
conditions that can only reduce the respective sets,
\[
 \begin{aligned}
 \mathcal B_x^{\rm int}(P)
 &\subseteq
 \mathscr P_\Gamma\times
 (\mathscr R_\Gamma\times\mathscr P_\Gamma)\times
 (\mathscr S_{|B_x|}\cup\{\infty\}),\\
 \mathcal B_x^{\rm ext}(P)
 &\subseteq
 (\mathscr R_\Gamma\times\mathscr P_\Gamma)\times
 (\mathscr S_{|B_x|}\cup\{\infty\}).
 \end{aligned}
\]
In particular,
\[
 |\mathcal B_x^{\rm int}(P)|\le P_s^2R_sM,
 \qquad
 |\mathcal B_x^{\rm ext}(P)|\le P_sR_sM.
\]
Here it is essential that the tables are sets with repetitions omitted: the
vertex $u\in cmp(x)$ is a bound variable in their definition and is not an
element of a stored table row.  Thus the two tables range over subsets of
fixed finite universes, and the number of possible pairs of tables is at
most
\[
 2^{P_s^2R_sM+P_sR_sM}.
\]
Finally, $I_x(P)$ and $I_x^{\rm anch}(P)$ are subsets of $[q]$,
$\mathcal A_x(P)$ is one of at most $R_s$ relations, and
$\mathcal C_x(P)$ is a subset of a set of at most $R_s$ relations.
Consequently
\[
 |\mathfrak T_x|
 \le 2^{2q}R_s2^{R_s}
       2^{P_s^2R_sM+P_sR_sM},
\]
which is a computable function of $f,k,r,d$.  Notice that this argument does
not assert, or require, that $\mathsf{Cert}(x)$ itself has parameter-bounded
cardinality.
\end{proof}

\begin{lemma}
\label{lem:bk-composition}
For each $y\in Y_x$, let $P_y,P'_y\in\mathsf{Cert}(y)$ satisfy
$\tau_y(P_y)=\tau_y(P'_y)$.  Fix the same local data at $x$ for both families,
and let $P$ and $P'$ be the resulting data below $x$.  Then
\[
 P\in\mathsf{Cert}(x)\quad\Longleftrightarrow\quad
 P'\in\mathsf{Cert}(x).
\]
If these conditions hold, then $\tau_x(P)=\tau_x(P')$.  Moreover, validity
and the common state are computable from the local data and the child states.
\end{lemma}
\begin{proof}
{
We first show that the child states and local data determine validity and
retained ancestry at $x$.  To recover the bag tables, we must express a
parent ancestry extension and marker position on each child interface,
taking account of pending colours that become anchored when the
certificates are combined.  The corresponding child signatures can then
be combined using their boundary labels and the locally assigned edges.
\par}

The child states and the local data determine how the support components meet
at $x$, and hence whether the incidence and anchoring conditions hold and the
resulting parent arcs form a forest.  Thus they determine membership in
$\mathsf{Cert}(x)$; assume it holds.  Replacing every path through a child by
the ancestry recorded at its interface preserves reachability and directed
cycles, and therefore determines $\mathcal A_x(P)$.

Fix an ancestry extension $\lambda$ at $x$ and a child $y$.  Define
$h_y:\Gamma_y(P_y)\to V(F_x(P))$ by $h_y(a)=a$ for $a\in B_y$, while, for
$i\in I_y(P_y)$, $h_y(i)$ is the representative in $P$ of the component of
$Q_i$ containing the open component of $Q_i^y$.  Pulling back
$\preceq_{x,\lambda}$ along $h_y$, that is, setting
$a\mathrel{\lambda_y}b$ if and only if
$h_y(a)\preceq_{x,\lambda}h_y(b)$, gives the induced ancestry extension at
$y$.  By~\eqref{eq:bk-edge-split},
$\lambda\in\mathcal C_x(P)$ precisely when
$\lambda_y\in\mathcal C_y(P_y)$ for every child $y$ and every edge of
$\mathcal E_x^0$ is a chain under $\preceq_{x,\lambda}$.  Thus
$\mathcal C_x(P)$ is determined.

Finally fix a row of either parent bag table, indexed by $(\lambda,D,U)$.
For an external row and a child $y$, define
\[
 \begin{aligned}
 D_y&:=\{a\in\Gamma_y(P_y)\mid
       h_y(a)\preceq_{x,\lambda}d\text{ for some }d\in D\},\\
 U_y&:=\{a\in\Gamma_y(P_y)\mid
       e\preceq_{x,\lambda}h_y(a)\text{ for some }e\in U\}.
 \end{aligned}
\]
Then $(D_y,U_y)$ is a position over $\lambda_y$, and the corresponding child
row is the external row indexed by $(\lambda_y,D_y,U_y)$.

Now suppose that the parent row is internal and marked by $u$.  Let
\[
 R_u:=\operatorname{rtc}\bigl(
 \preceq_{x,\lambda}\mathbin\cup(D\times\{u\})
 \mathbin\cup(\{u\}\times U)\bigr),
\]
and define
\[
 \begin{aligned}
 a\mathrel{\lambda_y^u}b
   &\quad\Longleftrightarrow\quad h_y(a)\mathrel{R_u}h_y(b),\\
 D_y^u&:=\{a\in\Gamma_y(P_y)\mid h_y(a)\mathrel{R_u}u\},\\
 U_y^u&:=\{a\in\Gamma_y(P_y)\mid u\mathrel{R_u}h_y(a)\}.
 \end{aligned}
\]
Here $\lambda_y^u\in\mathcal C_y(P_y)$ and $(D_y^u,U_y^u)$ is a position
over $\lambda_y^u$.  If $u\in cmp(y)$, this position extends
$\pi_y(P_y,u)$, and the corresponding child row is the internal row indexed
by $\bigl(\pi_y(P_y,u),(\lambda_y^u,D_y^u,U_y^u)\bigr)$.  If
$u\notin cmp(y)$, the corresponding child row is the external row indexed by
$(\lambda_y^u,D_y^u,U_y^u)$.

If a selected child signature is $\infty$, then so is the parent signature,
since the child retained set is contained in the parent retained set.
Otherwise, write $(J_y,S_y)$ for the selected signature at $y$ and put
$L_y=(V(J_y)\setminus B_y)\cup S_y$.  Writing $X^-,X^+$ for the two sets
defining the parent signature, $L_x(X^-,X^+)\setminus B_x$ is the disjoint
union of the sets $L_y\setminus B_x$.  Its intersection with $B_x$ is the
union of the sets $S_y\cap B_x$ together with the terms in~\eqref{eq:bk-L}
arising from $X^-\cap X^+$ and the edges in $\mathcal E_x^0$.  These data
determine $|L_x(X^-,X^+)|$.  If it is at most $\beta$, the hypergraphs $J_y$
and the edges in $\mathcal E_x^0$ determine
$H_x[L_x(X^-,X^+)\cup B_x]$ up to isomorphism fixing $B_x$ pointwise.  Thus
the parent signature, and hence $\tau_x(P)$, is determined effectively.
\end{proof}

{
\begin{example}
\label{ex:bk-merge}
We continue to follow the bag of $u$ from Example~\ref{ex:bk-rows}.
When the computation reaches $x$, its contributions from $y,z$ are
combined with the local edge $su$.

In both completions, $c$ is below $u$, so the edge $uc$ gives the
contribution $L_z=\{u\}$.  Since $u\notin cmp(z)$, this is an external
contribution.  We combine it with the contribution from $y$, identifying
the shared boundary vertex $u$, and add any new vertices retained by $su$:
\[
\begin{array}{c|c|c|c|c}
 &L_y&L_z&\text{new vertex contributed by }su&L_x\\ \hline
 F_A&\{a,u\}&\{u\}&\text{none}&\{a,u\}\\
 F_B&\{u\}&\{u\}&s&\{s,u\}
\end{array}
\]
In $F_A$, $s$ is below $u$, so $su$ adds nothing to $\{a,u\}$.
In $F_B$, $s$ is above $u$, so $su$ adds $s$.  Since $u\in mrg(x)$,
the combined contributions now belong to the internal table.  Both
signatures retain the boundary $\{s,u\}$, but mark only $u$ in the first
case and both vertices in the second.  Thus $s$, retained only for gluing
at $y$, becomes a bag member in the second case.

At the empty root, the signatures for the marker $u$ retain the hypergraphs
$H[\{a,u\}]$ and $H[\{s,u\}]$, with no boundary labels.
Both bags have size two, but
\[
 \rho_H(\{a,u\})=2,\qquad \rho_H(\{s,u\})=1.
\]
In $F_A$, no hyperedge contains both $a$ and $u$, so covering the bag
requires two edges.  In $F_B$, the local edge $su$ covers the entire bag.
\end{example}
}

{
We now combine the bound on the number of states with the composition
lemma to obtain a fixed-parameter algorithm for finding an optimal
elimination forest.
\par}
Write $\exp_1(x):=2^x$ and $\exp_{i+1}(x):=2^{\exp_i(x)}$ for $i\ge1$.

\begin{theorem}[Direct optimisation]
\label{th:bk-optimisation}
For every manageable width function $f$, there is a computable function
$g_f$ and an algorithm that, given a hypergraph $H$ and an integer $k\ge0$,
runs in time
\[
 g_f(k,\rank(H),\Delta(H))\,\|H\|^{O(1)}
\]
and either correctly reports that $\fwidth(H)>k$, or computes
$\fwidth(H)$ and outputs an optimal elimination forest and its induced
tree decomposition.
In particular, for $f\in\{\rho,\rho^*\}$, one may take
\[
 g_f(k,r,d)\le
 \exp_3\bigl(2^{O(kr)}+(r+d)^{O(1)}\bigr)
\]
where $r:=\rank(H)$ and $d:=\Delta(H)$.
\end{theorem}

\begin{proof}
Compute a tree decomposition of width at most $\beta-1$.  If none exists,
\eqref{eq:bk-size-bound} implies $\fwidth(H)>k$; otherwise root it at an empty
bag and binarise it.  Fix a choice $\ell$ of local data at $x$.
Lemma~\ref{lem:bk-composition} induces a partial map
\[
 \Phi_{x,\ell}:\prod_{y\in Y_x}\mathfrak T_y\rightharpoonup\mathfrak T_x.
\]
Indeed, for states $\sigma_y\in\mathfrak T_y$, choose
$P_y\in\mathsf{Cert}(y)$ with $\tau_y(P_y)=\sigma_y$ and form the datum $P$
below $x$ using $\ell$.  The map is defined precisely when
$P\in\mathsf{Cert}(x)$, and then has value $\tau_x(P)$.  The composition lemma
says that this is independent of the representatives $P_y$.  Conversely,
Proposition~\ref{prop:bk-certificate-restriction} and the definition of the
local data show that $\mathfrak T_x$ consists exactly of the defined values of
these maps, as $\ell$ ranges over the local choices at $x$ and
$(\sigma_y)_{y\in Y_x}$ ranges over the tuples of child states.  The tables are
therefore computed bottom-up by evaluating all these choices.

We now characterise the root states computed by the algorithm.  Let
$P\in\mathsf{Cert}(\rho)$.  Since $I_\rho(P)=\emptyset$, every support
component is closed and hence anchored, so $F:=F_\rho(P)$ is a rooted forest
on $V(H)=cmp(\rho)$.  For each $u\in V(H)$, let $C_u$ be the component of
$Q_{c(u)}$ anchored at $u$.
If $v$ is a child of $u$ in $F$, then $p(v)=c(u)$ and
$\mbag(v)\in V(C_u)$.  Since also $\mbag(u)\in V(C_u)$ and $C_u$ is
connected,
\begin{equation}
 \mathit{Stain}(u,({\bf t},F))\subseteq V(C_u).
 \label{eq:bk-stain-support}
\end{equation}
Distinct vertices of colour $i$ anchor distinct components of $Q_i$, so
their stains are disjoint.  Thus $c$ is a proper colouring of the stains of
$F$, although the support components may be larger than the stains.

At the root there is a unique ancestry extension $\lambda_\rho$, and, by
definition,
\begin{equation}
\lambda_\rho\in\mathcal C_\rho(P)
 \quad\Longleftrightarrow\quad
 F\text{ is an elimination forest of }H.
 \label{eq:bk-root-chain}
\end{equation}
When it holds, the interface is empty, so its only position is
$(\emptyset,\emptyset)$.  Equation~\eqref{eq:bk-L} gives, for every
$u\in V(H)$,
\[
 L_\rho(A^-_u(\emptyset),A^+_u(\emptyset))=\bag_F(u).
\]
Thus its signature is finite exactly when this bag has at most $\beta$
vertices; it then records the induced hypergraph on the bag and an empty
marked set.

A state $\tau\in\mathfrak T_\rho$ is \emph{feasible} if
$\lambda_\rho$ belongs to its $\mathcal C_\rho$-coordinate and all its internal
bag signatures are finite.  Let $w(\tau)$ be the maximum $f_J(V(J))$ over its
internal signatures $(J,\emptyset)$, with value $0$ for the empty hypergraph.
By invariant locality, $w(\tau)$ is the width of every complete certificate
with state $\tau$.  Hence
$w(\tau)\ge\fwidth(H)$.

Conversely, if $\fwidth(H)\le k$, take an optimal elimination forest $F$ with
properly $q$-coloured stains.  Let $Q_i$ be the union of the stains of colour
$i$, let $c(v)$ be the colour of the stain of $v$, and let $p(v)$ be the
colour of its parent's stain (or $0$ at a root).  The stain properties make
this a complete certificate whose root forest is $F$; proper colouring
makes each $Q_i$-component have exactly one anchor.  Its root state is
feasible by~\eqref{eq:bk-size-bound}, and has value $\fwidth(H)$.  Therefore
\[
 \min\{w(\tau)\mid\tau\text{ is a feasible root state}\}=\fwidth(H).
\]
Reject if this set is empty or its minimum exceeds $k$.  For each state
inserted into a table, record one local choice $\ell$ and one tuple of child
states mapped to it.  Tracing these choices from a minimising root state
reconstructs a complete certificate, hence an optimal elimination forest and
its induced tree decomposition.

Lemma~\ref{lem:bk-finite-types} bounds every table.  Since
$mrg(x)\subseteq B_x$ and $|B_x|\le\beta$, the number of local choices at
$x$ is at most
\[
 [q(q+1)]^\beta 4^q,
\]
and is therefore bounded by a function of $k,r,d$.  Moreover
$|\mathcal E_x^0|\le2^\beta$ and the host decomposition has polynomial size.
Thus the running time is
$g_f(k,r,d)\|H\|^{O(1)}$ for a computable
{$g_f$}.

\paragraph{Complexity.}
 Each finite signature contains
a hypergraph on at most $2\beta$ vertices, a labelled boundary, and a
marked subset of that boundary.  Thus the number $M$ in the proof of
Lemma~\ref{lem:bk-finite-types} is at most $2^{2^{O(\beta)}}$.
Substituting this estimate into that proof gives
\[
 N_f(k,r,d)\le \exp_2\bigl(O((\beta+q)^2)+2^{O(\beta)}\bigr).
\]
The colouring bound for the conflict graph following Theorem~\ref{th:cgchrom},
applied to a host decomposition with bags of size at most $\beta$, allows
us to take
\[
 q\le 2^{\,2^{O(\beta)}+(r+d)^{O(1)}}.
\]
Consequently, the number of states is at most
$\exp_3\bigl(2^{O(\beta)}+(r+d)^{O(1)}\bigr)$.

At each node, there are at most $N_f(k,r,d)^2$ pairs of child states and
$[q(q+1)]^\beta4^q$ local choices.  To bound the cost of a transition
explicitly, retain one representative partial certificate for each state.
For each candidate, its validity and state can be computed directly from
the definitions by enumerating interface relations and positions and
canonicalising the induced hypergraphs on at most $2\beta$ vertices.
These operations, and comparison of the resulting states, take time
polynomial in the state bound and the input size.
We can use Bodlaender's
algorithm~\cite{DBLP:journals/siamcomp/Bodlaender96} to construct the host tree
decomposition in time $2^{\beta^{O(1)}}\|H\|^{O(1)}$, and observe that the factor $2^{\beta^{O(1)}}$ is absorbed
by the state bound.
Hence the combinatorial part of the algorithm runs in time
\begin{equation}
 \exp_3\bigl(2^{O(\beta)}+(r+d)^{O(1)}\bigr)\,\|H\|^{O(1)}.
 \label{eq:bk-running-time}
\end{equation}

For an arbitrary manageable $f$, add the cost of computing $\beta_f(k,r)$
and multiply this bound by $1+t_f(\beta)$, where $t_f$ is a monotone
upper bound on the evaluation time in the automatisability condition.
For $f\in\{\rho,\rho^*\}$, we may take $\beta=\lceil kr\rceil$. There, enumeration of covers or linear programming bases, respectively, gives
evaluation time $2^{\beta^{O(1)}}$, which is absorbed by
\eqref{eq:bk-running-time}.  Thus the dependence on $kr$ is at most
quadruple exponential, whereas the dependence on $d$ is at most triple
exponential.
\end{proof}

\section{Conclusion \& Future Work}
\label{sec:conclusion}
In this paper, we have developed a general framework for fixed-parameter tractable
computation of hypergraph width measures. By introducing the class of \emph{manageable}
width functions, we obtain as a special case the first exact FPT algorithms for deciding $\mathrm{ghw}(H)\le k$
and $\mathrm{fhw}(H)\le k$ under any non-trivial parameterisation.
Our structural results support both an MSO-based algorithm and a direct dynamic program, which computes the exact manageable $f$-width and an optimal decomposition whenever the width is at most the supplied threshold.
Our framework further yields an exact FPT algorithm for discretised adaptive
width, as well as an FPT approximation scheme for the relaxed variant
$\mathcal{F}^*_\delta$-width. These adaptive-width results require no degree parameter.

Finally, we observe that for a simple graph $G$ with treewidth at least $4$,
the adaptive width of $G$ equals its $\mathcal{F}^*_{\Delta(G)+1}$-width. This yields an FPT approximation scheme for adaptive width on such graphs, with maximum degree as a parameter.

A natural question for further research is whether
the parameter $\degree(H)$ is necessary for fixed-parameter tractability of $\fwidth$
computation. For the related discussion, we find it easier to confine ourselves
to the specific width measure $\ghw$. Then the open question is as follows:
is testing $\ghw(H) \leq k$ fixed-parameter tractable if parameterised by $k$ and $\rank(H)$?

In order to upgrade the machinery of this paper to the parameterisation by $k$ and $\rank(H)$ only, it is sufficient to get rid of $\Delta(H)$ in the statement of Lemma \ref{lem:existsneutral}. Formally, this leads to the following question.

\begin{question} \label{quest:neutral1}
Is there a computable integer-valued function $h$ such that the following holds?
Let $F$ be a reduced elimination forest of $H$ and let $(Red,X,Blue)$
be a partition of $V(H)$ such that no hyperedge meets both $Red$ and $Blue$.
Let $B$ be a context factor of $F$ contained in $Red\cup Blue$ and suppose that the number
of monochromatic intervals of the spine of $B$ is at least
$h(|X|,\ghw(F),\rank(H))$. Then $B$ admits a neutral swap with respect to $\rho$-width (cf. Lemma \ref{lem:existsneutral}).
\end{question}

A positive answer to Question \ref{quest:neutral1}
would extend the algorithm in this paper to parameterisation
by $k$ and $\rank(H)$. To refute such a bound, a counterexample
of the following kind would suffice. Pick constants $s,k,r>0$ and an infinite sequence of hypergraphs $H_1,H_2,\dots$,
their respective reduced elimination forests $F_1,F_2,\dots$ and respective context factors $B_1,B_2,\dots$
with separating partitions $(Red_i,X_i,Blue_i)$ as above, $B_i\subseteq Red_i\cup Blue_i$, and
$|X_i|\le s$, $\ghw(F_i)\le k$, $\rank(H_i)\le r$ for every $i$,
with the number of monochromatic intervals on the spines of $B_i$ tending to infinity, and yet
none of the factors $B_i$ admitting a neutral swap.

Our approach to computing adaptive width on graphs through its equality with $\mathcal{F}^*_{\Delta(G)+1}$-width naturally raises two further open questions.
First, can the proposed approach to designing an approximation scheme be extended to hypergraphs of rank greater than $2$?
Second, can the approach be upgraded to an exact algorithm for adaptive width computation?
The first question amounts to a generalisation of Theorem \ref{thm:adaptive-graphs} to hypergraphs of rank greater than $2$.
This is formally stated below.

\begin{question}
    Are there computable positive integer-valued functions $g,t$, with $g(r,D)\ge r$, such that the following holds?
    Let $H$ be a hypergraph with $\tw(H)\ge t(\rank(H),\Delta(H))$.
    Then the adaptive width of $H$
    is its $\mathcal{F}^*_{g(\rank(H),\Delta(H))}$-width.
\end{question}

For an exact algorithm, recall that the proposed
approximation scheme for computation of $\mathcal{F}^*_\delta$-width
is obtained by exact computation of $\mathcal{F}_{\delta,b}$-width for
sufficiently large $b$. Within this approach, it would suffice to bound $b$
so that the resulting discretised version of adaptive width equals the
actual adaptive width. Below, we formally state the question for simple graphs.

\begin{question}
    Are there a computable integer-valued function $h$, with $h(D)\ge D+1$, and an integer $t\ge4$ such that the following holds?
    Let $G$ be a simple graph with $\tw(G)\ge t$.
    Then the adaptive width of $G$ is its $\mathcal{F}_{\Delta(G)+1,h(\Delta(G))}$-width.
\end{question}

We note that the three formally stated questions in this section are of a purely combinatorial nature
and hence may be of independent interest.

\bibliographystyle{ACM-Reference-Format}
\bibliography{references}

\clearpage

\appendix

\section{Proof of Reduced Hypergraph Dealternation for Strongly Manageable Functions}
\label{sec:dealtstrict.appendix}
\label{sec:dealtstrict}

In this appendix we prove the following variant of the Reduced Hypergraph
Dealternation Lemma for strongly $q$-manageable functions.

\StrongDealternation*

The main difference is that we use bags instead of balls. The structural
lemmas of the preceding subsections remain applicable.
In what follows, fix a witnessing tuple ${\bf W}=(W_1,\dots,W_q)$ for $f_H$
and a partition $(Red,X,Blue)$ of $V(H)$ such that no hyperedge contains
both a red and a blue vertex. Let $k\ge 0$ and let $F$ be a reduced
elimination forest of $H$ with $\fwidth(F)\le k$.
We may take the bounded size witness $\beta$ monotone in its arguments.
We use the notation for bags introduced above.

\begin{proposition} \label{prop:strongadd-strict}
Let $G$ be any elimination forest of $H$ and let $u\in Red\cup Blue$.
Put $A=X\text{-bag}_G(u)\cup nbag_G(u)$. Then
\[
    f_H(\bag_G(u))=
    \begin{cases}
        f_H(A), & fbag_G(u)=\emptyset,\\
        g_f\bigl(f_H(A),f_H(fbag_G(u))\bigr),
            & fbag_G(u)\neq\emptyset.
    \end{cases}
\]
\end{proposition}

\begin{proof}
The set $A$ contains $u$ and is therefore nonempty.
If the foreign bag is empty, then $\bag_G(u)=A$.
Otherwise, $A$ and the foreign bag form a partition of
$\bag_G(u)$ into two nonempty sets. Apply monotone separability.
\end{proof}

\begin{definition} \label{def:homofactor-strict}
With the notation as in Definition \ref{def:contmono}, let $P'$ be a maximal
subpath of $P$ such that $X\text{-bag}_F(u)$ is the same for all $u\in V(P')$.
Throughout this appendix, we refer to $P'$ as an $X$-interval of $U$.
Let $U_0\subseteq Red\cup Blue$ be a context factor and let $P_0$ be its spine.
We say that $U_0$ is $X$-\emph{homogeneous} if all the vertices of $P_0$
have the same $X$-bag.
\end{definition}

\begin{lemma} \label{lem:boundedxintervals-strict}
The number of $X$-intervals of $U$ is at most
$\kappa_{0,\mathrm{str}}(|X|)=|X|+1$.
\end{lemma}

\begin{proof}
Let $u,v\in V(P)$ with $u$ a predecessor of $v$.
Since $V(P)\cap X=\emptyset$, the predecessors of $u$ and $v$ in $X$
are the same. Moreover, $V(F_v)\subseteq V(F_u)$.
It follows from the definition of bags that
$X\text{-bag}_F(v)\subseteq X\text{-bag}_F(u)$.
Thus the $X$-bags along $P$ form a decreasing sequence of subsets of $X$.
Each change removes at least one vertex, so there are at most $|X|+1$
$X$-intervals.
\end{proof}

We use Lemma \ref{lem:neutralaux0} unchanged: all the vertices of a colour
interval $P'$ have the same foreign bag, denoted by $fbag_F(P')$.
Now, let $P'_1,\dots,P'_m$ be the colour intervals of the spine $P_0$
of an $X$-homogeneous context factor $U_0$, in the order starting from its root.

For every $j\ge 2$, the foreign bag of $P'_j$ is nonempty.
Indeed, reducedness places the last vertex of $P'_{j-1}$
in the foreign bag of the first vertex of $P'_j$.
Lemma~\ref{lem:neutralaux0} gives the assertion for the whole interval.

\begin{lemma} \label{lem:neutralaux2-strict}
Assume that there is $3\le i\le m-3$ such that
\[
    f_H(fbag_F(P'_i)),\ f_H(fbag_F(P'_{i+1})),\
    f_H(fbag_F(P'_{i+2})),\ f_H(fbag_F(P'_{i+3}))
\]
forms a mutable quadruple of the sequence
$f_H(fbag_F(P'_1)),\dots,f_H(fbag_F(P'_m))$.
Then the swap of $P'_{i+1}$ and $P'_{i+2}$ is neutral.
\end{lemma}

\begin{proof}
Let $F'$ be the elimination forest resulting from the swap.
The arguments of Claims \ref{clm:neutralaux21}--\ref{clm:neutralaux25}
do not use balls or $X$-homogeneity and apply unchanged.
In particular, the bags outside $V(P'_{i+1})\cup V(P'_{i+2})$ remain the same,
and for each $u\in V(P'_{i+1})\cup V(P'_{i+2})$,
$nbag_F(u)=nbag_{F'}(u)$.
Furthermore, for $u\in V(P'_{i+1})$,
$fbag_{F'}(u)=fbag_F(P'_{i+3})$, and for $u\in V(P'_{i+2})$,
$fbag_{F'}(u)=fbag_F(P'_i)$.

The argument of Claim \ref{clm:neutralaux26} applies with $X$-bags
in place of $X$-balls. We provide the details.
Let $u_2$ and $u_3$ be the first vertices of $P'_{i+2}$ and $P'_{i+3}$,
respectively. We first observe that
$pred_F(u)\cap X=pred_{F'}(u)\cap X$ for each $u\in V(P'_{i+1})\cup V(P'_{i+2})$.
Let $u\in V(P'_{i+2})$. By construction,
$succ_{F'}(u)=succ_F(u)\cup Cont_F(P'_{i+1},\{u_2\})$.
Therefore, if the $X$-bag changes, there is
$v\in X\text{-bag}_{F'}(u)\setminus X\text{-bag}_F(u)$ adjacent to a vertex
of $Cont_F(P'_{i+1},\{u_2\})$.
For any $u_0\in V(P'_i)$, this implies $v\in X\text{-bag}_F(u_0)$.
By $X$-homogeneity, $X\text{-bag}_F(u_0)=X\text{-bag}_F(u)$,
a contradiction.

Now, let $u\in V(P'_{i+1})$. Then
$succ_{F'}(u)=succ_F(u)\setminus Cont_F(P'_{i+2},\{u_3\})$.
If the $X$-bag changes, there is
$v\in X\text{-bag}_F(u)\setminus X\text{-bag}_{F'}(u)$.
For any $u_0\in V(P'_{i+3})$, $X$-homogeneity implies
$v\in X\text{-bag}_F(u_0)$.
By the argument of Claim \ref{clm:neutralaux21}, $v\in\bag_{F'}(u_0)$.
Since $v$ is a predecessor of $u$ and
$V(F'_{u_0})\subseteq V(F'_u)$, this implies
$v\in X\text{-bag}_{F'}(u)$, a contradiction.
We conclude that the $X$-bags of all swapped vertices remain the same.

Let $u\in V(P'_{i+1})$, let $C$ be the common $X$-bag on $P_0$,
and put $A=C\cup nbag_F(u)$.
The bag identities above give $A=C\cup nbag_{F'}(u)$.
The foreign bags used below are nonempty by the observation above.
Thus Proposition~\ref{prop:strongadd-strict} gives
\[
\begin{aligned}
f_H(\bag_{F'}(u))
&=g_f\bigl(f_H(A),f_H(fbag_F(P'_{i+3}))\bigr)\\
&\le g_f\bigl(f_H(A),f_H(fbag_F(P'_{i+1}))\bigr)
 =f_H(\bag_F(u)).
\end{aligned}
\]
The inequality follows from the mutable-quadruple assumption
and monotonicity of $g_f$.
For $u\in V(P'_{i+2})$ the argument is the same with $P'_i$ used instead
of $P'_{i+3}$ and $P'_{i+2}$ used instead of $P'_{i+1}$.
\end{proof}

\begin{lemma} \label{lem:homogenousneutral-strict}
Let $p(k,r)=\max\{1,\lceil\beta(k,r)\rceil\}$ and set
\[
    \kappa_{1,\mathrm{str}}(k,r,q)=4(p(k,r)+1)q^{p(k,r)}2^{2^{p(k,r)}}.
\]
Let $U_0$ be an $X$-homogeneous context factor with spine $P_0$ that has
at least $\kappa_{1,\mathrm{str}}(k,\rank(H),q)$ colour intervals.
Then $U_0$ admits a neutral swap of two non-terminal (neither first nor last)
intervals.
\end{lemma}

\begin{proof}
Put $r=\rank(H)$ and $p=p(k,r)$.
For $j\ge 3$, the foreign bag of $P'_j$ is nonempty and has
at most $p$ vertices. By invariant $q$-locality, its width is
determined by the isomorphism type of the induced hypergraph
together with the classes of the witnessing tuple.
The number of such types is at most
\[
    \sum_{n=1}^{p}q^n2^{2^n-1}\le pq^p2^{2^p-1}.
\]
Thus the sequence
$f_H(fbag_F(P'_3)),\dots,f_H(fbag_F(P'_m))$
takes at most $pq^p2^{2^p-1}$ distinct values.
Since $m\ge\kappa_{1,\mathrm{str}}(k,r,q)$, we have
\[
    m-2\ge 4pq^p2^{2^p-1}.
\]
By Theorem~\ref{theor:exmut}, this sequence has a mutable quadruple,
whose starting index $i$ satisfies $3\le i\le m-3$.
By Lemma \ref{lem:neutralaux2-strict}, the swap of $P'_{i+1}$ and
$P'_{i+2}$ is neutral. Neither of these two is a terminal colour interval.
\end{proof}

It remains to gather the statements together. Set
\[
    \kappa_{\mathrm{str}}(|X|,k,r,q)
    =\kappa_{0,\mathrm{str}}(|X|)\cdot\kappa_{1,\mathrm{str}}(k,r,q).
\]
The proof of Lemma \ref{lem:existsneutral} applies with
Lemmas \ref{lem:boundedxintervals-strict} and \ref{lem:homogenousneutral-strict}
in place of Lemmas \ref{lem:boundedxintervals} and \ref{lem:homogenousneutral}.
Consequently, any context factor with more than
$\kappa_{\mathrm{str}}(|X|,k,\rank(H),q)$ colour intervals admits a neutral swap.
The proofs of Lemmas \ref{lem:seqswaps} and \ref{lem:localdealt} now apply
with $\kappa_{\mathrm{str}}$ in place of $\kappa$.
The induction keeps $k$ fixed, and the preservation and transitivity
properties follow from Lemma \ref{lem:singleswap} and Remark \ref{rem:strongrepl},
respectively. Thus there is a strong replacement $F'$ of $F$ such that
the maximal factorisation of both $Red$ and $Blue$ under $F'$ has size at most
\[
    \kappa_{2,\mathrm{str}}(|X|,k,r,q)
    =2(9+3|X|)\bigl(1+\kappa_{\mathrm{str}}(|X|,k,r,q)\bigr),
\]
where $r=\rank(H)$, using Lemma \ref{lem:transopr7.5} and
Corollary \ref{cor:monspine} as in the proof of Lemma \ref{lem:localdealt}.

\begin{proof}[Proof of Lemma \ref{lem:dealt-strict}]
Choose a reduced elimination forest $F$ with $\fwidth(F)=\fwidth(H)$,
as in the proof of Lemma \ref{lem:dealt}, and put $k=\fwidth(H)$ and
$w=\tw({\bf t})$.
Use the same canonical red-blue decompositions and bottom-up sequence
of nodes of ${\bf t}$ as in that proof, applying the strong replacement
just established at each step. Claim \ref{clm:inv111} applies unchanged,
so the factorisations obtained at earlier steps are preserved.
Each separator is $X=adh(x)$ and satisfies $|X|\le w+1$.
It follows that the resulting reduced elimination forest $F_m$ satisfies
\[
    \splitf({\bf t},F_m)
    \le \kappa_{2,\mathrm{str}}(w+1,k,\rank(H),q).
\]
Since each replacement is strong,
$\fwidth(H)\le\fwidth(F_m)\le\fwidth(F)=\fwidth(H)$, so $F_m$ is optimal.
Finally, Lemma \ref{lem:splitmir} bounds $\mir({\bf t},F_m)$ by
\[
    \beta(k,\rank(H))\cdot
    \bigl(9\kappa_{2,\mathrm{str}}(w+1,k,\rank(H),q)+3(w+1)\bigr).
\]
Taking a computable monotone upper bound for these two expressions
gives the required function $\gamma_{\mathrm{str}}$.
\end{proof}
 
\section{Proof of Theorem \ref{th:cgchrom}}
\label{sec:comb}

\begin{definition} \label{def:stinter}
With the notation as in Definition \ref{def:stain}
the \emph{stain intersection graph} 
$\mathsf{SG}({\bf t},F)$ is a bipartite graph with
parts corresponding to $V(T)$ and $V(H)$
with $x \in V(T)$ adjacent to $u \in V(H)$ if and only if
$x \in Stain(u,({\bf t},F))$.
\end{definition}

\begin{theorem} \label{theor:sgdegbound}
There is a function $h$ such that
the max-degree of a vertex of $V(T)$ in 
$\mathsf{SG}({\bf t},F)$ is at most 
$h(split({\bf t},F),mir({\bf t},F),\tw({\bf t}),\rank(H))$.
\end{theorem}
 
Proving Theorem \ref{theor:sgdegbound}
requires some preparation. 

\begin{definition}
Let $x \in V(T)$ and let $\{x,u\} \in E(\mathsf{SG}({\bf t},F))$. 
We say that $\{x,u\}$ is a \emph{margin edge} 
if $u \in mrg(x)$. 
\end{definition}

\begin{definition}
Let $x \in V(T)$. 
The \emph{characteristic partition} of $x$
denoted by $\textsc{cp}_{\bf t}(x)$
(the subscript may be omitted if clear from the context)
is the family of all non-empty sets 
among 
$$\{mrg(x), V(H) \setminus cmp(x)\} \cup \{cmp(y)| y \in children_T(x)\}.$$
Further on, the \emph{characteristic factorisation} of 
$x$ denoted by $\textsc{cf}_{{\bf t},F}(x)$
(again, the subscript may be omitted if clear from the context)
is $\bigcup_{U \in \textsc{cp}(x)} \textsc{mf}_F(U)$.  
\end{definition}

\begin{lemma} \label{lem:contfactbounded}
The number of context factors of 
$\textsc{cf}_{{\bf t},F}(x)$ that are not subsets of $mrg(x)$ is at most 
$mir({\bf t},F) \cdot split({\bf t},F)+3\cdot split({\bf t},F)$. 
\end{lemma}

\begin{proof}
Each context factor $U$ as specified in the statement can be 
unambiguously mapped to an element $U^*$ of
$\textsc{cp}_{\bf t}(x)$ such that $U \subseteq U^*$. 
We observe that the number of such $U^*$ is at most
$mir({\bf t},F)+1$. Indeed, $U^*$ can be $cmp(y)$ such 
that $y$ is a child of $x$ whose component is not well-formed
(and there are at most $mir({\bf t},F)$ such children)
or $U^*$ may be $V(H) \setminus cmp(x)$. 
Each element $cmp(y)$ has at most $split({\bf t},F)$ factors 
in its maximal factorisation, so the total number of context factors 
yielded as subsets of $cmp(y)$ for $y$ being a child of $x$ is
$mir({\bf t},F) \cdot split({\bf t},F)$. 
Further on, the total number of factors in the maximal factorisation 
of $cmp(x)$ is $split({\bf t},F)$,
so the number of factors in the maximal factorisation of its 
complement is at most $3 \cdot split({\bf t},F)$ by
Lemma~3.3 of \cite{DBLP:journals/lmcs/BojanczykP22}. Clearly, the number of context factors in
the factorisation is at most this amount.
\end{proof}

\begin{lemma} \label{lem:nonmargcontext}
Let $x \in V(T)$, $u \in V(H)$ and $\{x,u\} \in E(\mathsf{SG}({\bf t},F))$. 
Suppose that $\mbag(u) \neq x$ (or, to put it differently, $u \notin mrg(x)$). 
Let $U \in \textsc{cf}_{{\bf t},F}(x)$ be the element containing $u$. 
Then $U$ is a context factor of $F$ where $u$ is the parent of the appendices. 
\end{lemma}

\begin{proof}
By assumption, the only way for $x$ to be included 
in $Stain(u,({\bf t},F))$ is to be not a first vertex on a path
between $\mbag(u)$ and $\mbag(v)$ where $v$ is a child of $u$ in $F$. 
Let $W_0$ and $W_1$ be elements of $\textsc{cp}_{\bf t}(x)$
containing $u$ and $v$ respectively. 

\begin{claim}
$W_0 \neq W_1$. 
\end{claim}

\begin{proof}
If $W_1=mrg(x)$, we are done by assumption. 
Next, assume that $W_1=cmp(y)$ where $y$ is a child of $x$. 
It follows then that $\mbag(v) \in V(T_y)$. 
If we assume that $W_0=W_1$ then $\mbag(u) \in V(T_y)$.
Hence, the path between $\mbag(u)$ and $\mbag(v)$ fully lies in
$T_y$ and does not involve $x$, a contradiction. 
It remains to assume that $W_1=V(H) \setminus cmp(x)$. 
Again if we assume that $W_0=W_1$ then both $\mbag(u)$ and
$\mbag(v)$ lie outside of $T_x$ and hence the path between them
does not involve $x$ either causing the same contradiction.
\end{proof}

As $U \subseteq W_0$, it follows from the claim
that $u$ is contained in a factor while its child is not. 
It can only be possible if $U$ is a context factor and
$u$ is the appendices' parent. 
\end{proof}

\begin{proof}[Proof of \Cref{theor:sgdegbound}]
The number of margin edges incident to $x$
is at most the margin size of $x$ which is in turn
at most the bag size of ${\bf t}$, and hence at most
$\tw({\bf t})+1$.

Let $NM(x)$ be the neighbours of $x$ in $\mathsf{SG}({\bf t},F)$
connected to $x$ through a non-margin edge. 
Let ${\bf y}: NM(x) \rightarrow \textsc{cf}_{{\bf t},F}(x)$
be the mapping that maps each vertex of $NM(x)$ to the factor containing it. 
By Lemma \ref{lem:nonmargcontext}, each $u \in NM(x)$ is
mapped to a context factor where $u$ is the appendices' parent. 
As each context factor contains precisely one appendices' parent this 
mapping is injective. It follows that $|NM(x)|$ is upper bounded by the
number of context factors in $\textsc{cf}_{{\bf t},F}(x)$
that are not subsets of $mrg(x)$, which is at most
$(mir({\bf t},F)+3)split({\bf t},F)$ by
Lemma \ref{lem:contfactbounded}. Thus one may take
$h(s,m,t,r)=t+1+(m+3)s$.
\end{proof}

\begin{proof}[Proof of \Cref{th:cgchrom}]
The proof consists of the following steps.
\begin{enumerate}
\item The graph $\cg({\bf t},F)$ is chordal as
it is the intersection graph of connected subgraphs of a forest~\cite{golumbic2004algorithmic}.
\item It follows from the chordality of $\cg({\bf t},F)$
that its chromatic number is at most the size of its maximum clique~\cite{golumbic2004algorithmic}.
\item It is known that the intersection graph 
of connected subgraphs of a forest has the Helly property. 
In other words, if all pairwise intersections of stains are non-empty, then there is a vertex that belongs to all the stains. 
\item It follows that the maximum clique size of 
$\cg({\bf t},F)$ is upper bounded by the max-degree of a vertex of
$V(T)$ in $\mathsf{SG}({\bf t},F)$. 
Indeed, let $u_1, \dots u_q$ form a clique of 
$\cg({\bf t},F)$. This means that for each $i \neq j \in [q]$, 
$Stain(u_i,({\bf t},F)) \cap Stain(u_j,({\bf t},F)) \neq \emptyset$. 
It follows from the Helly property that 
$\bigcap_{i \in [q]} Stain(u_i,({\bf t},F)) \neq \emptyset$. 
Let $x \in \bigcap_{i \in [q]} Stain(u_i,({\bf t},F))$. 
Then $x$ is adjacent to each of $u_1, \dots, u_q$ in $\mathsf{SG}({\bf t},F)$. 
\item The theorem is now immediate from Theorem \ref{theor:sgdegbound}. 
\end{enumerate}
\end{proof}

\section{Verifying Hypergraph Widths in MSO} \label{sec:mso}
\begin{definition} \label{def:wellformtd}
Let $M$ be a model over the signature $\tau_{td}$.
We say that $M$ is \emph{well-formed} if there is a hypergraph $H$
and a tree decomposition ${\bf t}=(T,{\bf B})$ of $H$
such that $M=M(H,{\bf t})$. 
We then refer to $H$ and ${\bf t}$ as $H(M)$ and ${\bf t}(M)$ respectively. 
\end{definition}

It is important for our reasoning that  the treewidth of $M$
is not much larger than that of ${\bf t}$. 

\begin{lemma} \label{lem:mthtw}
Let $M$ be a well-formed model of $\tau_{td}$. 
Then $\tw(M) \leq \tw({\bf t}(M))+2$. 
\end{lemma}

\begin{proof}
Let $G(M)$ be a graph over $U(M)$ such that 
$\{u,v\} \in E(G(M))$ if and only if there is a binary relation
$X$ of $\tau_{td}$ such that either $X(u,v)$ or $X(v,u)$ is true for $M$.
Unary relations do not affect the Gaifman graph. It therefore suffices to demonstrate
that $\tw(G(M)) \leq \tw({\bf t}(M))+2$. 

For each $e \in E(H)$, fix 
a node $t(e) \in V(T)$ such that $e \subseteq {\bf B}(t(e))$. 
Let $T^*$ be the tree with the same root as $T$ and
obtained from $T$ by adding a new node $t_e$ for each $e \in E(H)$
and making $t_e$ a child of $t(e)$.

Let ${\bf B^*}: V(T^*) \rightarrow 2^{U(M)}$ be defined as follows.  
If $t$ is the root of $T^*$ then ${\bf B^*}(t)={\bf B}(t) \cup \{t\}$.
If $t$ is a non-root vertex of $T$ then ${\bf B^*}(t)={\bf B}(t) \cup \{t,parent(t)\}$. 
Finally, for each node $t_e$, ${\bf B^*}(t_e)=e \cup \{e\}$. 

It can be verified by a direct inspection that
$(T^*,{\bf B^*})$ is a tree decomposition of $G(M)$ of width at most $\tw({\bf t}(M))+2$.
\end{proof}

Fix a manageable width function $f$. Let $\mathbf{H}_{f,k,r}$ contain
one representative of each isomorphism type of hypergraphs $H$ of rank
at most $r$ such that $f_H(V(H))\leq k$.

\begin{lemma} \label{lem:msoaux1}
There is an algorithm with runtime upper bounded by a monotone function of $k$
and $r$ that returns $\mathbf{H}_{f,k,r}$.
\end{lemma}

\begin{proof}
Let $\beta=\beta[f]$.
By the bounded size property, $\mathbf{H}_{f,k,r}$
is a subset of the set of hypergraphs of rank at most $r$
and with at most $\beta(k,r)$ vertices. 
Moreover, by the automatisability property, the value $f_H(V(H))$
for each such graph can be tested by an algorithm whose runtime is upper bounded
by a monotone function of $k$ and $r$. Therefore, enumerate one representative of each hypergraph isomorphism
type of rank at most $r$ on at most $\beta(k,r)$ vertices, and retain
exactly those $H$ with $f_H(V(H))\leq k$.
\end{proof}

\begin{lemma} \label{lem:msoaux2}
Let $H$ be a hypergraph of rank at most $r$ and let $U \subseteq V(H)$.
Then $f_H(U)\leq k$ if and only if $H[U]$ is isomorphic to an element
of $\mathbf{H}_{f,k,r}$. 
\end{lemma}

\begin{proof}
We simply observe that, by the invariant locality property,
$f_H(U)=f_{H[U]}(V(H[U]))$. 
By definition of $\mathbf{H}_{f,k,r}$, this value is at most $k$
if and only if $H[U]$ is isomorphic to one of its representatives.
\end{proof}

For a model encoding an elimination forest, define
\[
\begin{split}
 \mathit{IsBag}(u,U):=\forall v\bigl(U(v)\leftrightarrow
  \bigl(u=v\lor\bigl(\mathit{Vertex}(v)\land \mathit{Descendant}(v,u)\land{}\\
  \exists w\bigl(\mathit{Descendant}(u,w)\land \mathit{Neighb}(v,w)\bigr)\bigr)\bigr)\bigr),
\end{split}
\]
where $\mathit{Descendant}(a,b)$ expresses reflexive reachability from $a$ to $b$
in $(\mathit{Vertex},\mathit{Child})$, which is MSO definable, and
\[
 \mathit{Neighb}(v,w):=v\ne w\land\exists e\bigl(\mathit{Edge}(e)\land
 \mathit{Incident}(v,e)\land \mathit{Incident}(w,e)\bigr).
\]
Thus $\mathit{IsBag}(u,U)$ holds exactly when $U=\bag_{F(M)}(u)$.

\begingroup
\begin{proof}[Proof of \Cref{lem:width-condition}]
Put $p=\beta(k,r)$ and let $\mathbf H=\mathbf H_{f,k,r}$.
Compute $\mathbf H$ using Lemma~\ref{lem:msoaux1}.

For each representative $G$, construct an MSO formula $\theta_G(B)$
expressing that $B\subseteq \mathit{Vertex}$ and $H(M)[B]\cong G$.
To do this, enumerate the vertices of $B$ by distinct variables and specify the
nonempty intersections of hyperedges with $B$.
More explicitly, for every nonempty subset of the enumerated vertices,
test whether it is the intersection of $B$ with some hyperedge of
$H(M)$, and require this precisely for the edges of $G$.
Each such test is first-order definable using $\mathit{Edge}$ and $\mathit{Incident}$.
Only the resulting intersections matter; different hyperedges having
the same intersection give the same edge of $H(M)[B]$.
Set
\[
    \theta_{f,k,r}(B)=\bigvee_{G\in\mathbf H}\theta_G(B).
\]
By invariant locality, as in Lemma~\ref{lem:msoaux2},
this formula holds if and only if $f_{H(M)}(B)\leq k$.

Now define
\[
    \psi_{f,k,r}:=
    \forall u\bigl(\mathit{Vertex}(u)\rightarrow
    \exists B\bigl(\mathit{IsBag}(u,B)\land\theta_{f,k,r}(B)\bigr)\bigr).
\]
For an elimination forest $F$, $\mathit{IsBag}(u,B)$ specifies its unique
bag at $u$. Thus $M(H,F)\models\psi_{f,k,r}$ if and only if every
bag has $f$-value at most $k$, which is exactly the required
threshold condition.

For the labelled version, enumerate instead all hypergraphs $G$ of
rank at most $r$ with at most $p$ vertices, together with weak partition
tuples $(W_1,\dots,W_\ell)$ of $V(G)$.
Retain a labelled type if its $f$-value on $V(G)$ is at most $k$.
In the corresponding type formula
$\theta_{G,\mathbf W}(B,U_1,\dots,U_\ell)$, also require that each
enumerated vertex belongs to $U_i$ precisely when its image belongs
to $W_i$. For a weak partition tuple $(U_1,\dots,U_\ell)$, this expresses
\[
    (H(M)[B],B\cap U_1,\dots,B\cap U_\ell)
    \cong (G,W_1,\dots,W_\ell).
\]
Take the disjunction over the retained labelled types and use it in
the same $\mathit{IsBag}$ formula. Labelled invariant locality gives the same
equivalence with the $f$-value of each bag. The sets $U_i$ remain free
variables throughout, yielding $\psi_{f,k,r}(U_1,\dots,U_\ell)$.

There are finitely many hypergraph types on at most $p$ vertices and
at most $\ell^p$ assignments of vertex classes to each such type.
All evaluations and formula constructions therefore take time
bounded by a computable function of $k,r,\ell,\beta$ and the
evaluation algorithm for $f$, as claimed.
\end{proof}
\endgroup

\smallskip
\noindent\textbf{MSO Transductions \quad}
Ultimately our FPT algorithm will be stated in terms of MSO transductions. 
A transduction is a mapping of a model over the input signature to a set of models over the output signature. We define first \emph{elementary transductions}. In this paper, we use only three kinds of elementary transductions: Colouring, MSO Interpretation, and Filtering.
\begin{description}
    \item[Colouring.] Colouring adds a unary predicate to the input structure, i.e. the output vocabulary has an additional unary predicate. 
    The input model $M$ is mapped to the set of all possible models obtained from $M$ by adding a unary relation over the universe of $M$. 
    We note that colouring is the only \emph{guessing} transduction mapping the input model to several output models. 
    \item[MSO Interpretation.] 
    For any relation name $R$ in the output, there exists a formula $\varphi_R(x_1, \cdots, x_k)$ with the same arity as $R$. 
    For a model $M$, the output is a singleton whose only element has the same universe as $M$
    in addition to relations $R$ such that $(x_1, \cdots, x_k) \in R$ if and only if $\varphi_R(x_1, \cdots, x_k)$ is true. 
    We note that the MSO interpretation is the elementary transduction that allows an \emph{arbitrary} signature change
    (the colouring also changes the signature but only in a specific way of adding a single unary relation). 

    \item[Filtering.] For an MSO formula $\varphi$ there is a filtering transduction which filters out structures which do not satisfy $\varphi$, i.e. a structure $M$ is mapped
    to $\{M\}$ if it satisfies $\varphi$ and to $\{\}$ otherwise. This is the elementary transduction 
    that can cause a model to be mapped to an empty set. 
   
\end{description}

Now, let ${\bf g}_1, \dots {\bf g}_p$ be a sequence of elementary transductions
such that for each $2 \leq i \leq p$ the input signature of ${\bf g}_i$ is
the same as the output signature of ${\bf g}_{i-1}$. 
Then the \emph{transduction} $Tr[{\bf g}_1, \dots, {\bf g}_p]$ is a function
over the models $M$ of the input signature of ${\bf g}_1$ 
(which is also the input signature of $Tr[{\bf g}_1, \dots, {\bf g}_p]$)
defined as follows. 
If $p=1$ then $Tr[{\bf g}_1](M)={\bf g}_1(M)$. 
Otherwise $Tr[{\bf g}_1, \dots, {\bf g}_p](M)=\bigcup_{M' \in {\bf g}_1(M)} Tr[{\bf g}_2, \dots, {\bf g}_p](M')$.
The size of an MSO transduction is the sum of the sizes of its elementary MSO transductions, where colouring has size 1 and, for the remaining types of elementary transduction, the size is the total size of MSO formulas (i.e., number of symbols) that make up its description. We write $\|\mathbf{Tr}\|$ for the size of transduction $\mathbf{Tr}$.

We note that in \cite{DBLP:journals/lmcs/BojanczykP22}, a transduction is defined as a binary relation with pairs $(M_L,M_R)$ where the `left-hand' model $M_L$ 
is over the input signature whereas the `right-hand' model $M_R$ is over the output signature.
Our definition is equivalent as a mapping of a left-hand model to a set of right-hand models determines
a binary relation and vice versa. 
For a detailed description of MSO logic and formal languages we refer the reader to \cite{DBLP:books/daglib/0030804}.

We use the following theorem to express the existence of a transduction output satisfying a given MSO formula.

\begin{theorem} \label{th:transdtomso}
    There is an algorithm $\mathbf{A}$ that takes as input a transduction 
    $\mathbf{Tr}$ and an MSO formula $\psi(U_1, \dots, U_q)$ 
    over the output signature of  $\mathbf{Tr}$
    (where $U_1, \dots, U_q$ are free set variables)
    and returns an MSO formula $\varphi(U_1, \dots, U_q)$ so that the following holds.
    Let $M$ be a model over the input signature of $\mathbf{Tr}$. 
    Then $M \models \varphi(U_1, \dots, U_q)$
    if and only if there is $M' \in \mathbf{Tr}(M)$
    such that $M' \models \psi(U_1, \dots, U_q)$. 
    The runtime of the algorithm depends on the size of $\mathbf{Tr}$
    and the size of $\psi(U_1, \dots, U_q)$.
\end{theorem}

\begin{proof}
Process the elementary transductions in reverse order, starting with $\psi$.
For a colouring introducing a predicate $C$, replace $C$ by a fresh set
variable $X$ and existentially quantify $X$. For filtering by a sentence
$\eta$, conjoin $\eta$. For an interpretation, replace every relation atom
by its defining formula, renaming bound variables to avoid capture.
Each step expresses that some output of that elementary transduction
satisfies the current formula, and preserves the free variables $U_i$.
Induction proves the equivalence. All operations are effective, with
runtime bounded by a computable function of the two description sizes.
\end{proof}

We note two nuances concerning Theorem \ref{th:transdtomso}.
The first is the use of the MSO formula $\psi$ in addition to the transduction.
We use the accompanying formula to accommodate free set parameters, which are not part of the transduction description.

For the second nuance, note that the instantiations of $U_1, \dots U_q$
must be subsets of the universes of both $M$ and $M'$. This is valid because
the elementary transductions as defined above \emph{do not} change the universe. 

\begingroup
\subsection{Transducing elimination forests}
\label{sec:transd}
\begingroup
\begin{lemma}\label{lem:transtwforest}
There is an algorithm whose input is $q>0$ and the output is a
transduction ${\bf Tr}$ of size $O(q)$ whose input signature is $\tau_{td}$
and the output signature is $\tau_{ef}$.
Let ${\bf t},F$ be a \tf-pair of hypergraph $H$ and let $M(H,{\bf t})$
be the relational structure over $\tau_{td}$ induced by $H$ and ${\bf t}$.
Then the following is true regarding ${\bf Tr}(M(H,{\bf t}))$.
\begin{enumerate}
\item Let $M^*\in{\bf Tr}(M(H,{\bf t}))$. Then $H(M^*)=H$ and
$F(M^*)$ is an elimination forest of $H$.
Moreover, $\chi(\cg({\bf t},F(M^*)))\leq q$.
\item If the chromatic number of $\cg({\bf t},F)$ is at most $q$,
then there is $M^*\in{\bf Tr}(M(H,{\bf t}))$ such that $F(M^*)=F$.
\end{enumerate}
\end{lemma}
\endgroup

\newcommand{\Stain}{\ensuremath{\mathit{Stain}}\xspace}
 The central part of the proof is the notion of a transduction tuple as defined below.
\begin{definition} \label{def:transdtuple}
Let $({\bf t}=(T,B),F)$ be a \textsc{tf} pair of a hypergraph $H$. 
A \emph{transduction tuple} $({\bf C},{\bf D}, {\bf K}, {\bf f}, {\bf g})$ is defined as follows. 
\begin{enumerate}
\item ${\bf C}$ is a partition of $V(H)$ into independent sets of $\mathsf{CG}({\bf t},F)$.
To put it differently ${\bf C}$ determines a proper colouring of $\mathsf{CG}({\bf t},F)$.
\item ${\bf D}$ is a weak partition of $V(H)$. 
\item ${\bf K}$ is a family of subgraphs. 
\item ${\bf f}$ is a mapping from ${\bf C}$ to ${\bf D}$ defined as follows. 
     \begin{enumerate}
		 \item There is exactly one element $D_0 \in {\bf D}$ such that ${\bf f}^{-1}(D_0)=\emptyset$. 
		 \item For each $C \in {\bf C}$, ${\bf f}(C)$ is the set of all children (in $F$) of the non-leaf elements of $C$. 
		 \end{enumerate}
\item ${\bf g}$ is a surjective mapping from ${\bf C}$ to ${\bf K}$ such that for each 
$C \in {\bf C}$, ${\bf g}(C)=\bigcup_{u \in C} T[\Stain(u,({\bf t},F))]$. 
\end{enumerate}

We also introduce additional related terminology. 
In particular, for each  $u \in V(H)$, $C(u)$ is the element of ${\bf C}$ containing $u$. 
Also, for each $C \in {\bf C}$, we define a function ${\bf h}_C$ on $C$ such that,
for each $u \in C$, ${\bf h}_C(u)$ is the connected component of ${\bf g}(C)$ containing $\mbag(u)$. 
\end{definition}

In the subsequent proof, we demonstrate that, on the one hand, a transduction tuple, together with ${\bf t}$ 
allows reconstruction of $F$
and, on the other hand, it can be guessed by a transduction. For the reconstruction of $F$, we need the 
following definition. 

\begin{definition} \label{def:child0}
With the notation as in Definition \ref{def:transdtuple},
let $child$ be a binary relation on $V(H)$ consisting of all pairs $(u,v)$ such that
the following is true. 
\begin{enumerate}
\item $v \in {\bf f}(C(u))$. 
\item $\mbag(v) \in {\bf h}_{C(u)}(u)$. 
\end{enumerate} 
\end{definition}

Clearly, the $child$ relation is completely determined by ${\bf t}$ and $({\bf C},{\bf D},{\bf K},{\bf f},{\bf g})$. 

\begin{lemma} \label{lem:reconstr}
$(V(H),child)=F$. 
\end{lemma}

\begin{proof}
We first need to clarify that the connected components of 
each ${\bf g}(C)$ are, in fact, stains of an arbitrary vertex of the component. 

\begin{claim} \label{clm:reconstr}
For each $C \in {\bf C}$ and each $u \in C$,
${\bf h}_C(u)=\Stain(u,({\bf t},F))$. 
\end{claim}
\begin{proof}
As $\mbag(u) \in \Stain(u,({\bf t},F)) \subseteq V({\bf g}(C))$
and $\Stain(u,({\bf t},F))$ is a connected subset of $T$,
we conclude that $\Stain(u,({\bf t},F)) \subseteq {\bf h}_C(u)$.
Assume that  $\Stain(u,({\bf t},F)) \neq {\bf h}_C(u)$. 
Then there is $v_0 \in \Stain(u,({\bf t},F))$ and 
$v_1 \notin \Stain(u,({\bf t},F))$ such  that ${\bf g}(C)$ has an
edge between $v_0$ and $v_1$. 
By construction, the edge is an edge of $T[\Stain(u',({\bf t},F))]$
for some $u' \in C$. We note that $u' \neq u$ since $v_1 \notin \Stain(u,({\bf t},F)) $
and that $v_0 \in \Stain(u,({\bf t},F))$. It follows that the stains of $u$ and $u'$
intersect but this is a contradiction to $C$ being an independent set. 
\end{proof}

Assume now that $(u,v) \in F$. 
By definition, $v \in {\bf f}(C(u))$
and also $\mbag(v) \in \Stain(u,({\bf t},F))$. 
By Claim \ref{clm:reconstr}, $\mbag(v) \in {\bf h}_{C(u)}(u)$.
We conclude that $(u,v) \in child$. 

Conversely, assume that $(u,v) \in child$.
Then, by definition, $v$ is not a root of $F$.
Let $parent_F(v)$ be its parent in $F$. Hence, $(parent_F(v),v)$ is an edge of $F$.
We have already observed that $(parent_F(v),v) \in  child$. 
We proceed to note that the in-degree of $v$ in child is exactly one,
which will imply that $parent_F(v)=u$ and hence $(u,v) \in F$. 
Indeed, assume that there are two distinct $u_1,u_2$ such that
both $(u_1,v)$ and $(u_2,v)$ are elements of $child$. 
If $C(u_1)\neq C(u_2)$ then $v$ has a parent in each of two disjoint
classes of $\mathbf C$, contradicting uniqueness of its parent in $F$;
hence we conclude that $C(u_1)=C(u_2)=C$. 
Now, it follows from Claim \ref{clm:reconstr} that $\mbag(v) \in \Stain(u_1,({\bf t},F)) \cap \Stain(u_2,({\bf t},F))$
in contradiction to the independence of $u_1$ and $u_2$. 
\end{proof}

The main idea of the proof of Lemma \ref{lem:transtwforest} is guessing a transduction tuple. 
A subtle aspect of the proof is guessing of $K$ that is a family of subgraphs of $T$ rather than
a family of sets. The guessing is done by guessing two sets per subgraph. One of them is the set of vertices of the subgraphs and the other is a subset of the first set consisting of the vertices that are not connected to their parents in the resulting subgraph. 

\begin{proof}[Proof of Lemma \ref{lem:transtwforest}]
The algorithm produces a transduction
$${\bf Tr}=Tr[{\bf c}_1,\dots, {\bf c}_{q},{\bf c^*},{\bf d}_0, \dots, {\bf d}_q,{\bf d^*},
{\bf k}^0_1, {\bf k}^1_1, \dots, {\bf k}^0_q,{\bf k}^1_q, {\bf k^*}, {\bf x}_0, {\bf x}_1]$$ 
The elementary transductions ${\bf c}_1, \dots, {\bf c}_q$ are all colouring transductions introducing respective
new unary relations $C_1,\dots, C_q$. Let $\tau_1$ be the signature obtained from
$\tau_{td}$ by adding $C_1,\dots, C_q$. Then ${\bf c^*}$ is the filtering 
transduction whose input signature is $\tau_1$. The associated MSO sentence is
true if and only if $C_1,\dots, C_q$ weakly partition $\mathit{Vertex}$ (that is, one or more 
of $C_i$ may be empty sets). 

The elementary transductions ${\bf d}_0, \dots, {\bf d}_q$ are all colouring transductions introducing respective
new unary relations $D_0,\dots, D_q$. Let $\tau_2$ be the signature obtained from 
$\tau_1$ by adding the relations $D_0, \dots, D_q$. 
Then ${\bf d^*}$ is the filtering transduction whose input is $\tau_2$.
The associated MSO sentence is true if and only if $D_0, \dots, D_q$ weakly partition
$\mathit{Vertex}$. 

The elementary transductions
${\bf k}^0_1, {\bf k}^1_1, \dots, {\bf k}^0_q,{\bf k}^1_q$
are all colouring transductions introducing respective new 
unary relations $K^0_1,K^1_1,\dots, K^0_q,K^1_q$. 
Let $\tau_3$ be the signature obtained from $\tau_2$ by 
introducing the relations $K^0_1,K^1_1,\dots,K^0_q,K^1_q$. 
Then ${\bf k^*}$ is the filtering transduction whose 
accompanying MSO sentence is true if and only if the following statements hold. 
\begin{enumerate}
\item Each $K_i^j$ is a subset of $\mathit{Node}$. Moreover, each $K_i^1$ is a subset 
of $K_i^0$. For each $1 \leq i \leq q$, and each vertex $u \in K^1_i$, $u$ is not a root
of $T$ and, moreover, $parent(u)$, the parent of $u$, belongs to $K^0_i$. 
\item For each $1 \leq i \leq q$, let us define $T_i$ as the subgraph 
of $T[K^0_i]$ obtained by removing from $T[K^0_i]$ the edges $(parent(u),u)$ for each
$u \in K^1_i$. Then there is a bijection $\gamma_i$ from $C_i$ to the connected
components of $T_i$ so that each $u \in C_i$ is mapped to the component 
containing $\mbag(u)$.  
\end{enumerate}

Next, ${\bf x}_0$ is an interpretation
transduction mapping a model of $\tau_3$ to the model over $\tau_{ef}$ 
over the same universe.
Thus the underlying tuple of MSO formulas is 
$(iv(x),ie(x),ia(x,y),\varphi(x,y))$.
The formulas $iv(x),ie(x),ia(x,y)$ are true if and only if, respectively
$x \in \mathit{Vertex}$, $x \in \mathit{Edge}$ and $(x,y)\in \mathit{Incident}$. 
In other words, these formulas simply copy $\mathit{Vertex}$, $\mathit{Edge}$, and $\mathit{Incident}$
relations to the new model.

The formula $\varphi(x,y)$ is true if and only if the following
statements hold. 
\begin{enumerate}
\item Let $i$ be such that $y \in D_i$. Then $i>0$ and $x \in C_i$. 
\item $\mbag(y)$ belongs to the connected component of $T_i$ that
contains $\mbag(x)$. 
\end{enumerate}

Finally, ${\bf x}_1$ is a filtering transduction over signature $\tau_{ef}$
associated with a formula that is true over a model $M$ of $\tau_{ef}$ if and only
if $F(M)$ is an elimination forest of $H(M)$. 

We observe that the transduction as above is straightforward to be produced by an algorithm
with input $q$, that the transduction size is linear in $q$ and that for each
$M_1 \in {\bf Tr}(M(H,{\bf t}))$, $M_1$ is well formed in the sense that the $\mathit{Child}$
relation of the model is a binary relation over $\mathit{Vertex}$. 
For each $i$, the margin node of every vertex in $C_i$ and the margin nodes
of all its children lie in its assigned component of $T_i$.
Hence its stain lies in that component, so distinct vertices in $C_i$
have disjoint stains. Thus every output forest satisfies the claimed chromatic bound.

It remains to demonstrate that, when $q \geq q_0$
where $q_0$ is the chromatic number of $\mathsf{CG}({\bf t},F)$ then 
${\bf Tr}(M(H,{\bf t}))$ contains an output $M^*$ such that $F(M^*)=F$.%

In this case, there is $M_1 \in Tr[{\bf  c}_1, \dots, {\bf c}_q](M(H,{\bf t}))$  
that is obtained by introduction of relations $C_1, \dots, C_q$
such that $C_1, \dots, C_{q_0}$ are colouring classes of $\mathsf{CG}({\bf t},F)$
and, if $q_0<q$, $C_{q_0+1}, \dots, C_q$ are all empty sets. 
By construction, $M_1$ satisfies the underlying sentence of ${\bf c^*}$. 

Next, let $D_0$ be the set consisting of the roots of $F$ and for each
$1 \leq i \leq q$, let $D_i$ consist of all the children of the non-leaf
elements of $C_i$. We observe that $D_0, \dots D_q$ constitute a weak 
partition of $V(H)$. Indeed, every non-root element appears in some 
$D_i$ due to $C_1, \dots, C_q$ being a weak partition of $V(H)$. 
On the other hand, an element appearing in two distinct $D_i$ and $D_j$
means that either an element is both a root of $F$ and not a root of $F$
or that an element has two parents. 
Therefore, there is $M_2 \in Tr[{\bf d}_0, \dots, {\bf d}_q](M_1)$ obtained 
from $M_1$ by adding relations $D_0, \dots, D_q$ and $M_2$ satisfies 
the underlying sentence of ${\bf d^*}$.

For each  $i \in [q]$, let $K^0_i=\bigcup_{u \in C_i} \Stain(u,({\bf t},F))$
and let $T_i=\bigcup_{u \in C_i} T[\Stain(u,({\bf t},F))]$. 
We observe that there is $K^1_i \subseteq K^0_i$ 
such that $T_i$ is obtained from $T[K^0_i]$ by removal of edges 
$\{(parent(u),u)|\allowbreak u \in K^1_i\}$. 
By definition, there is  $M_3 \in Tr[{\bf k}^0_1,{\bf k}^1_1, \dots, {\bf k}^0_q,{\bf k}^1_q](M_2)$
obtained from $M_2$ by introduction of relations $K^0_1,K^1_1, \dots, K^0_q,K^1_q$. 
By definition, $M_3$ satisfies the sentence ${\bf k^*}$.

Let ${\bf  C}=\{C_1, \dots, C_{q_0}\}$. 
Let ${\bf D}=\{D_i| 0 \leq i \leq q_0\}$. 
Let ${\bf K}=\{T_i|i \in [q_0]\}$. 
Let ${\bf f}$ be a function on ${\bf C}$ that maps each $C_i$ to $D_i$. 
Let ${\bf g}$ be a function on ${\bf C}$ that maps each $C_i$ to $T_i$. 

We observe that $({\bf C},{\bf D},{\bf K},{\bf f},{\bf g})$ 
is a transduction tuple of $({\bf t},F)$. Furthermore, the relation determined by $\varphi(x,y)$ of ${\bf x}_0$
is the $child$ relation for $({\bf C},{\bf D},{\bf K},{\bf f},{\bf g})$ 
as in Definition \ref{def:child0}.
It follows from Lemma \ref{lem:reconstr} that there is $M^* \in Tr[{\bf x}_0](M_3)$
such that $F(M^*)=F$. As $F$ is a proper elimination forest of $H$,
$M^* \in Tr[{\bf x}_0,{\bf x}_1](M_3)$ thus completing the proof of the lemma. 
\end{proof} 
 \endgroup

\begingroup
\begin{proof}[Proof of \Cref{lem:mso-blackbox}]
Let ${\bf Tr}$ be the transduction constructed in the proof of
Lemma~\ref{lem:transtwforest}, with input $q$.
Its outputs encode elimination forests of $H$, and every elimination
forest $F$ with $\chi(\cg({\bf t},F))\leq q$ occurs among these outputs.
Moreover, every output forest satisfies this chromatic bound.

The transduction preserves the universe and the predicates $\mathit{Vertex}$
and $\mathit{Edge}$. To evaluate $\psi$ on $M(H,F)$, relativise its quantifiers
to $\mathit{Vertex}\cup \mathit{Edge}$: restrict element quantifiers to this set and
set quantifiers to its subsets. Denote the resulting formula by
$\psi^\circ(U_1,\dots,U_\ell)$.
For every output $M^*$ encoding $F$ and every weak partition
$(U_1,\dots,U_\ell)$ of $V(H)$, we then have
\[
    M^*\models\psi^\circ(U_1,\dots,U_\ell)
    \quad\Longleftrightarrow\quad
    M(H,F)\models\psi(U_1,\dots,U_\ell).
\]
In particular, the sets $U_i$ remain unchanged.

Apply Theorem~\ref{th:transdtomso} to ${\bf Tr}$ and $\psi^\circ$.
It returns an MSO formula $\varphi(U_1,\dots,U_\ell)$ over $\tau_{td}$
which holds if and only if some output of ${\bf Tr}$ satisfies
$\psi^\circ$. The preceding observations give exactly the equivalence
in the lemma. Since ${\bf Tr}$ is constructed from $q$ and has size
$O(q)$, the construction time depends only on $q,\ell,\|\psi\|$.

\end{proof}
\endgroup

\end{document}